\documentclass[10pt,journal,compsoc]{IEEEtran}
\ifCLASSOPTIONcompsoc
  \usepackage[nocompress]{cite}
\else
  \usepackage{cite}
\fi
\usepackage[utf8]{inputenc} \usepackage{amssymb} \usepackage{multirow} \usepackage{graphicx} \usepackage{tabularx} \usepackage[acronym]{glossaries} \usepackage{makecell} \usepackage{rotating} \usepackage{threeparttable} \usepackage{booktabs} \usepackage{ragged2e} \usepackage{array}
\usepackage{bm}
\usepackage{mathrsfs}
\usepackage{amsfonts}
\usepackage{amssymb}
\usepackage{amsmath}
\usepackage{latexsym}
\usepackage{graphicx}
\usepackage{subfigure}
\usepackage{amsthm}
\usepackage{hyperref}
\usepackage{indentfirst}
\usepackage{algorithm} %%format of the algorithm
\usepackage{algorithmic} %%format of the algorithm
\usepackage{multirow} %%multirow for format of table
\usepackage{xcolor}
\usepackage{color}
\usepackage{longtable}
\usepackage{graphics}
\usepackage{graphicx}
\usepackage{epsfig}
\usepackage[switch]{lineno}
\usepackage{hyperref}

\newcommand{\tabincell}[2]{\begin{tabular}{@{}#1@{}}#2\end{tabular}}

\usepackage{hyperref}
\usepackage[hyphenbreaks]{breakurl}

\usepackage{stfloats}

\theoremstyle{plain}
\newtheorem{thm}{Theorem}[section]

\theoremstyle{definition}

\theoremstyle{remark}

\ifCLASSINFOpdf
\else
\fi
\begin{document}
%
% paper title
% Titles are generally capitalized except for words such as a, an, and, as,
% at, but, by, for, in, nor, of, on, or, the, to and up, which are usually
% not capitalized unless they are the first or last word of the title.
% Linebreaks \\ can be used within to get better formatting as desired.
% Do not put math or special symbols in the title.
\title{Residual-driven Fuzzy \emph{C}-Means Clustering for Image Segmentation}
%
%
% author names and IEEE memberships
% note positions of commas and nonbreaking spaces ( ~ ) LaTeX will not break
% a structure at a ~ so this keeps an author's name from being broken across
% two lines.
% use \thanks{} to gain access to the first footnote area
% a separate \thanks must be used for each paragraph as LaTeX2e's \thanks
% was not built to handle multiple paragraphs
%
%
%\IEEEcompsocitemizethanks is a special \thanks that produces the bulleted
% lists the Computer Society journals use for "first footnote" author
% affiliations. Use \IEEEcompsocthanksitem which works much like \item
% for each affiliation group. When not in compsoc mode,
% \IEEEcompsocitemizethanks becomes like \thanks and
% \IEEEcompsocthanksitem becomes a line break with idention. This
% facilitates dual compilation, although admittedly the differences in the
% desired content of \author between the different types of papers makes a
% one-size-fits-all approach a daunting prospect. For instance, compsoc
% journal papers have the author affiliations above the "Manuscript
% received ..."  text while in non-compsoc journals this is reversed. Sigh.

\author{Cong~Wang,
        Witold~Pedrycz,~\IEEEmembership{Fellow,~IEEE,}
        ZhiWu~Li,~\IEEEmembership{Fellow,~IEEE,}
        and~MengChu~Zhou,~\IEEEmembership{Fellow,~IEEE}% <-this % stops a space
\IEEEcompsocitemizethanks{
\IEEEcompsocthanksitem C. Wang is with the School of Electro-Mechanical Engineering, Xidian University, Xi'an 710071, China. E-mail: wangc0705@stu.xidian.edu.cn.
% note need leading \protect in front of \\ to get a newline within \thanks as
% \\ is fragile and will error, could use \hfil\break instead.
\IEEEcompsocthanksitem W. Pedrycz is with the Department of Electrical and Computer Engineering, University of Alberta, Edmonton, AB T6R 2V4, Canada, the School of Electro-Mechanical Engineering, Xidian University, Xi'an 710071, China, and also with the Faculty of Engineering, King Abdulaziz University, Jeddah 21589, Saudi Arabia. E-mail: wpedrycz@ualberta.ca.
\IEEEcompsocthanksitem Z. Li is with the School of Electro-Mechanical Engineering, Xidian University, Xi'an 710071, China, and also with the Institute of Systems Engineering, Macau University of Science and Technology, Macau, China. E-mail: zhwli@xidian.edu.cn.
\IEEEcompsocthanksitem M. Zhou is with the Institute of Systems Engineering, Macau University  of Science and Technology, Macau 999078, China and also with the Helen and John C. Hartmann Department of Electrical and Computer Engineering, New Jersey Institute of Technology, Newark, NJ 07102 USA. E-mail: zhou@njit.edu.
}% <-this % stops an unwanted space
%\thanks{Manuscript received April 19, 2005; revised August 26, 2015.}
}

% note the % following the last \IEEEmembership and also \thanks -
% these prevent an unwanted space from occurring between the last author name
% and the end of the author line. i.e., if you had this:
%
% \author{....lastname \thanks{...} \thanks{...} }
%                     ^------------^------------^----Do not want these spaces!
%
% a space would be appended to the last name and could cause every name on that
% line to be shifted left slightly. This is one of those "LaTeX things". For
% instance, "\textbf{A} \textbf{B}" will typeset as "A B" not "AB". To get
% "AB" then you have to do: "\textbf{A}\textbf{B}"
% \thanks is no different in this regard, so shield the last } of each \thanks
% that ends a line with a % and do not let a space in before the next \thanks.
% Spaces after \IEEEmembership other than the last one are OK (and needed) as
% you are supposed to have spaces between the names. For what it is worth,
% this is a minor point as most people would not even notice if the said evil
% space somehow managed to creep in.

% The paper headers
\markboth{}%Journal of IEEE Transactions on Pattern Analysis and Machine Intelligence,~Vol.~xx, No.~xx, xxxx~2020
{Wang \MakeLowercase{\textit{et al.}}: Residual-driven Fuzzy \emph{C}-Means Clustering for Image Segmentation}
% The only time the second header will appear is for the odd numbered pages
% after the title page when using the twoside option.
%
% *** Note that you probably will NOT want to include the author's ***
% *** name in the headers of peer review papers.                   ***
% You can use \ifCLASSOPTIONpeerreview for conditional compilation here if
% you desire.

% The publisher's ID mark at the bottom of the page is less important with
% Computer Society journal papers as those publications place the marks
% outside of the main text columns and, therefore, unlike regular IEEE
% journals, the available text space is not reduced by their presence.
% If you want to put a publisher's ID mark on the page you can do it like
% this:
%\IEEEpubid{0000--0000/00\$00.00~\copyright~2015 IEEE}
% or like this to get the Computer Society new two part style.
%\IEEEpubid{\makebox[\columnwidth]{\hfill 0000--0000/00/\$00.00~\copyright~2015 IEEE}%
%\hspace{\columnsep}\makebox[\columnwidth]{Published by the IEEE Computer Society\hfill}}
% Remember, if you use this you must call \IEEEpubidadjcol in the second
% column for its text to clear the IEEEpubid mark (Computer Society jorunal
% papers don't need this extra clearance.)

% use for special paper notices
%\IEEEspecialpapernotice{(Invited Paper)}

% for Computer Society papers, we must declare the abstract and index terms
% PRIOR to the title within the \IEEEtitleabstractindextext IEEEtran
% command as these need to go into the title area created by \maketitle.
% As a general rule, do not put math, special symbols or citations
% in the abstract or keywords.
\IEEEtitleabstractindextext{%
\begin{abstract}
Due to its inferior characteristics, an observed (noisy) image's direct use gives rise to poor segmentation results. Intuitively, using its noise-free image can favorably impact image segmentation. Hence, the accurate estimation of the residual between observed and noise-free images is an important task. To do so, we elaborate on residual-driven Fuzzy \emph{C}-Means (FCM) for image segmentation, which is the first approach that realizes accurate residual estimation and leads noise-free image to participate in clustering. We propose a residual-driven FCM framework by integrating into FCM a residual-related fidelity term derived from the distribution of different types of noise. Built on this framework, we present a weighted $\ell_{2}$-norm fidelity term by weighting mixed noise distribution, thus resulting in a universal residual-driven FCM algorithm in presence of mixed or unknown noise. Besides, with the constraint of spatial information, the residual estimation becomes more reliable than that only considering an observed image itself. Supporting experiments on synthetic, medical, and real-world images are conducted. The results demonstrate the superior effectiveness and efficiency of the proposed algorithm over existing FCM-related algorithms.
\end{abstract}

% Note that keywords are not normally used for peerreview papers.
}

% make the title area
\maketitle

% To allow for easy dual compilation without having to reenter the
% abstract/keywords data, the \IEEEtitleabstractindextext text will
% not be used in maketitle, but will appear (i.e., to be "transported")
% here as \IEEEdisplaynontitleabstractindextext when the compsoc
% or transmag modes are not selected <OR> if conference mode is selected
% - because all conference papers position the abstract like regular
% papers do.
\IEEEdisplaynontitleabstractindextext
% \IEEEdisplaynontitleabstractindextext has no effect when using
% compsoc or transmag under a non-conference mode.

% For peer review papers, you can put extra information on the cover
% page as needed:
% \ifCLASSOPTIONpeerreview
% \begin{center} \bfseries EDICS Category: 3-BBND \end{center}
% \fi
%
% For peerreview papers, this IEEEtran command inserts a page break and
% creates the second title. It will be ignored for other modes.
\IEEEpeerreviewmaketitle

\IEEEraisesectionheading{\section{Introduction}\label{sec:introduction}}
\IEEEPARstart{A}{S} an important approach to data analysis and processing, fuzzy clustering has been widely applied to a number of visible domains such as pattern recognition \cite{Baraldi1999I,Baraldi1999II}, data mining \cite{Subbalakshmi2016}, granular computing \cite{Zhuxiubin2017}, and image processing \cite{Yambal2013}. One of the most popular fuzzy clustering methods is a Fuzzy \emph{C}-Means (FCM) algorithm \cite{Dunn1973,Bezdek1981,Bezdek1984}. It plays a significant role in image segmentation; yet it only works well for noise-free images. In real-world applications, images are often contaminated by different types of noise, especially mixed or unknown noise, produced in the process of image acquisition and transmission. Therefore, to make FCM robust to noise, FCM is refined resulting in many modified versions in two main means, i.e., introducing spatial information into its objective function \cite{Ahmed2002,Chen2004,Szilagyi2003,Cai2007,Krinidis2010,
Celik2013} and substituting its Euclidean distance with a kernel distance (function) \cite{Gong2013,Lin2014,Elazab2015,Zhao2013,Guo2016,Zhao2014,Zhu2017,Wang2019}. Even though such versions improve its robustness to some extent, they often fail to account for high computing overhead of clustering. To balance the effectiveness and efficiency of clustering, researchers have recently attempted to develop FCM with the aid of mathematical technologies such as Kullback-Leibler divergence \cite{Gharieb2017,Wang2020KL}, sparse regularization \cite{Gu2018,Wang2020SR}, morphological reconstruction \cite{Wang2020KL,Vincent1993,Najman1996,Chen2012} and gray level histograms \cite{Lei2018,Lei2019}, as well as pre-processing and post-processing steps like image pixel filtering \cite{Wang2020RS}, membership filtering \cite{Lei2018} and label filtering \cite{Wang2020SR,Wang2020RS,Bai2019}. To sum up, the existing studies make evident efforts to improve its robustness mainly by means of noise removal in each iteration or before and after clustering. However, they fail to take accurate noise estimation into account and apply it to improve FCM.

Generally speaking, noise can be modeled as the residual between an observed image and its ideal value (\mbox{noise-free} image). Clearly, its accurate estimation is beneficial for image segmentation as noise-free image instead of observed one can then be used in clustering. Most of FCM-related algorithms suppress the impact of such residual on FCM by virtue of spatial information. So far, there are no studies focusing on developing FCM variants based on an in-depth analysis and accurate estimation of the residual. To the best of our knowledge, there is only one attempt \cite{Zhang2019} to improve FCM by revealing the sparsity of the residual. To be specific, since a large proportion of image pixels have small or zero noise/outliers, $\ell_{1}$-norm regularization can be used to characterize the sparsity of the residual, thus forming deviation-sparse FCM (DSFCM). When spatial information is used, it upgrades to its augmented version, named as DSFCM\_N. Their residual estimation is realized by using a soft thresholding operation. In essence, such estimation is equivalent to noise removal. Therefore, neither of them can achieve highly accurate residual estimation.

To address this issue, we elaborate on residual-driven FCM (RFCM) for image segmentation, which furthers FCM's performance. We first design an RFCM framework, as shown in Fig. \ref{fig:frame}(b), by introducing a fidelity term on residual as a part of the objective function of FCM. This term makes residual accurately estimated. It is determined by a noise distribution, e.g., an $\ell_{2}$-norm fidelity term corresponds to Gaussian noise and an $\ell_{1}$-norm one suits impulse noise. In real-world applications, since images are often corrupted by mixed or unknown noise, a specific noise distribution is difficult to be obtained. To deal with this issue, by analyzing the distribution of a wide range of mixed noise, especially a mixture of Poisson, Gaussian and impulse noise, we present a weighted $\ell_{2}$-norm fidelity term in which each residual is assigned a weight, thus resulting in an augmented version namely WRFCM for image segmentation with mixed or unknown noise. To obtain better noise suppression, we also consider spatial information of image pixels in  WRFCM since it is naturally encountered in image segmentation. In addition, we design a two-step iterative algorithm to minimize the objective function of WRFCM. The first step is to employ the Lagrangian multiplier method to optimize the partition matrix, prototypes and residual when fixing the assigned weights. The second step is to update the weights by using the calculated residual. Finally, based on the optimal partition matrix and prototypes, a segmented image is obtained.
\begin{figure}[htbp]
\centering
\begin{minipage}[t]{1\linewidth}
\centering
\includegraphics[width=1\textwidth]{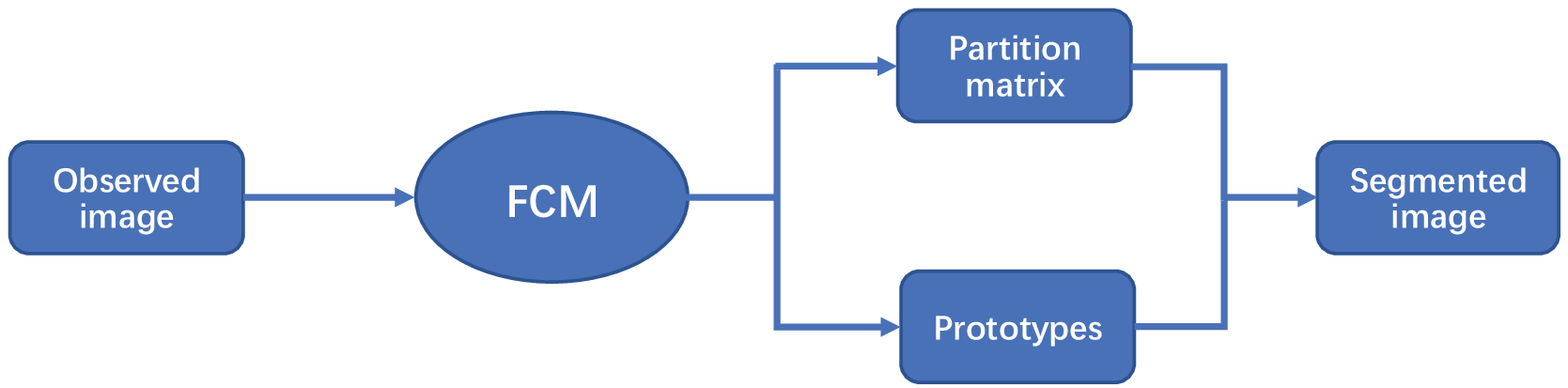}
\centerline{\footnotesize{(a)}}
\end{minipage}
\begin{minipage}[t]{1\linewidth}
\centering
\includegraphics[width=1\textwidth]{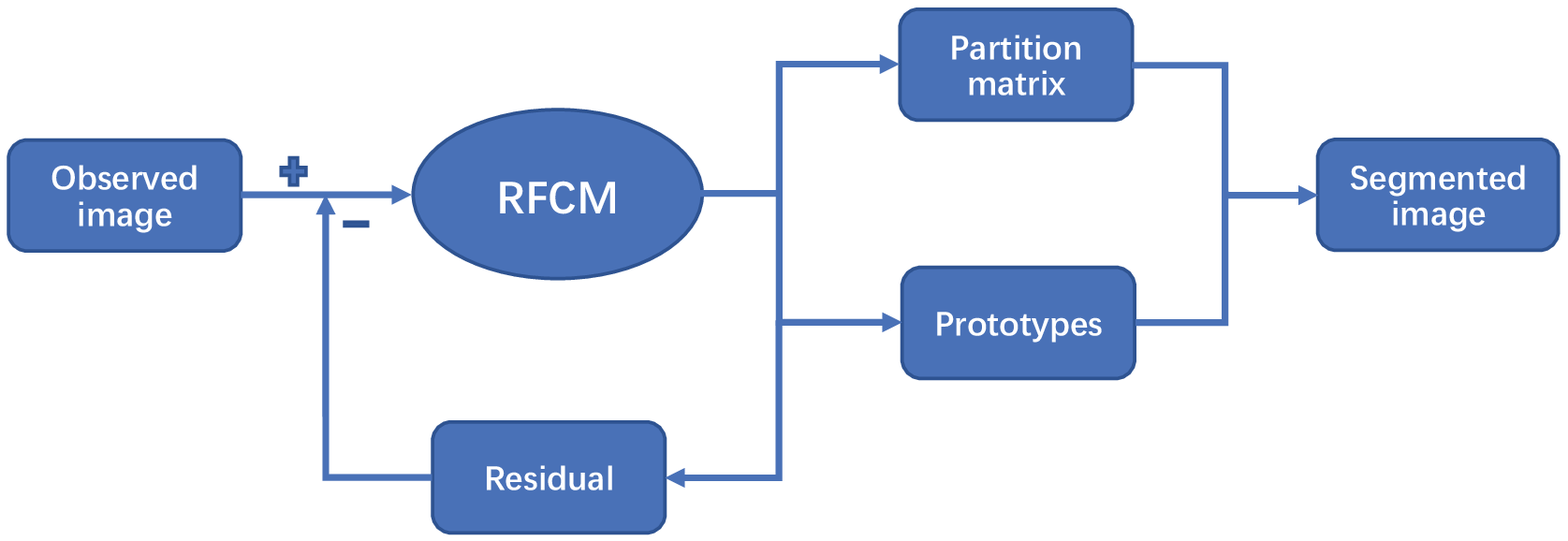}
\centerline{\footnotesize{(b)}}
\end{minipage}
\caption{A comparison between the frameworks of FCM and RFCM. (a) FCM; and (b) RFCM.}
\label{fig:frame}
\end{figure}

This study makes fourfold contributions to advance FCM for image segmentation:
\begin{itemize}
\item  For the first time, we propose an RFCM framework for image segmentation by introducing a fidelity term derived from a noise distribution into FCM. It relies on accurate residual estimation to greatly improve FCM's performance, which is absent from existing FCM-related algorithms.

%which are determined by different types of noise distributions. The terms make residuals accurately estimated, which means that a noise-free image instead of an observed one can be used in clustering.

\item  Built on an RFCM framework, we present WRFCM by weighting mixed noise distribution and incorporating spatial information. The use of spatial information makes resulting residual estimation more reliable. It is regarded as a universal RFCM algorithm for coping with mixed or unknown noise.

%We present a weighted $\ell_{2}$-norm fidelity term in which each residual is assigned a weight by analyzing the distribution of a mixture of Poisson, Gaussian, and impulse noise. We extend it to the accurate estimation of universal mixed or unknown noise, resulting in augmented RFCM (WRFCM).

%\item  Spatial information is incorporated in WRFCM. Due to its use, WRFCM's performance is further improved and resulting residual estimation becomes more reliable. As a result, we extend WRFCM to handle universal image segmentation problems.

\item  We design a two-step iterative algorithm to realize WRFCM. Since only $\ell_2$ vector norm is involved, it is fast by virtue of a Lagrangian multiplier method.

%The first step is to employ a Lagrangian multiplier method to optimize a partition matrix, prototypes and residual when fixing the assigned weights. The second step is to update the weights by using the calculated residual.
\item  WRFCM is validated to produce state-of-the-art performance on synthetic, medical and real-world images from four benchmark databases.
%3D MRI Simulated Brain Database, Berkeley Segmentation Data Set, Microsoft Research Cambridge Object Recognition Image Database, and NASA Earth Observation Database.
\end{itemize}

The originality of this work comes with a realization of accurate residual estimation from observed images, which benefits FCM's performance enhancement.
%realizes originally accurate residual estimation from observed image, which drives that noise-free image participates in clustering. Thus FCM's performance is improved greatly.
%In presence of a wide range of noise, especially mixed or unknown noise, this work originally improves FCM's robustness by realizing an accurate residual estimation from observed image, which makes noise-free image participate in clustering.The originality of this work is to realize accurate residual estimation from observed image so that noise-free image participates in clustering, which improves FCM's robustness to a wide range of noise, especially mixed or unknown noise.
In essence, the proposed algorithm is an unsupervised method. Compared with commonly used supervised methods such as convolutional neural networks (CNNs) \cite{Fakhry2017,Zhang2017,Kokkinos2019,Ren2019,Zhang2020,He2016} and dictionary learning \cite{Jiang2014,Zhou2019}, it realizes the residual estimation precisely by virtue of a fidelity term rather than using any image samples to train a residual estimation model. Hence, it needs low computing overhead and can be experimentally executed by using a \mbox{low-end} CPU rather than a high-end GPU, which means that its practicality is high. In addition, being free of the aid of mathematical techniques, it achieves the superior performance over some recently proposed comprehensive FCMs. Therefore, we conclude that WRFCM is a fast and robust FCM algorithm. Finally, in a mathematical sense, its minimization problem involves an $\ell_{2}$ vector norm only. Thus it can be easily solved by using a \mbox{well-known} Lagrangian multiplier method.

Section \ref{sec:relatedwork} reviews the state of the art relevant to this work. Section \ref{sec:methodology} details conventional FCM and the proposed methodology. Section \ref{sec:experiments} reports experimental results. Conclusions and some open issues are given in Section \ref{sec:conclusions}.

\section{Related Work}\label{sec:relatedwork}
In 1984, Bezdek et al. \cite{Bezdek1984} first proposed FCM. So far, it has evolved into the most popular fuzzy clustering algorithm. However, it cannot work well for segmenting observed (noisy) images. It has been improved by mostly considering spatial information \cite{Ahmed2002,Chen2004,Szilagyi2003,Cai2007,Krinidis2010,Celik2013}, kernel distances (functions) \cite{Gong2013,Lin2014,Elazab2015,Zhao2013,Guo2016,Zhao2014,Zhu2017,Wang2019}, and various mathematical techniques \cite{Gharieb2017,Wang2020KL,Gu2018,Wang2020SR,Vincent1993,Najman1996,Chen2012,Lei2018,Lei2019,Wang2020RS,Bai2019}. In this paper, we mainly focus on the improvement of FCM with regard to its robustness to noise for image segmentation. Therefore, we introduce related work about it in this section.

\subsection{FCM with Spatial Information}

Over the past two decades, using spatial information to improve FCM's robustness achieved remarkable successes, thus resulting in many improved versions \cite{Ahmed2002,Chen2004,Szilagyi2003,Cai2007,Krinidis2010,Celik2013}. For instance, Ahmed et al. \cite{Ahmed2002} introduce a neighbor term into the objective function of FCM so as to improve its robustness by leaps and bounds, thus yielding FCM\_S where S refers to ``spatial information". To further improve it, Chen and Zhang \cite{Chen2004} integrate mean and median filters into a neighbor term, thus resulting in two FCM\_S variants labeled as FCM\_S1 and FCM\_S2. However, their computing overhead is very high. To lower it, Szilagyi et al. \cite{Szilagyi2003} propose an enhanced FCM (EnFCM) where a weighted sum image is generated by the observed pixels and their neighborhoods. Based on it, Cai et al. \cite{Cai2007} substitute image pixels by gray level histograms, which gives rise to fast generalized FCM (FGFCM). Although it has a high computational efficiency, more parameters are required and tuned. Krinidis et al. \cite{Krinidis2010} come up with a fuzzy local information C-means algorithm (FLICM) for simplifying the parameter setting in FGFCM. Nevertheless, FLICM considers only non-robust Euclidean distance that is not applicable to arbitrary spatial information.

\subsection{FCM with Kernel Distance}

To address the serious shortcoming of FLICM \cite{Krinidis2010}, kernel distances (functions) are used to replace Euclidean distance in FCM. They realize the transformation from an original data space to a new one. As a result, a collection of kernel-based FCMs have been put forward \cite{Gong2013,Lin2014,Elazab2015,Zhao2013,Guo2016,Zhao2014,Zhu2017,Wang2019}. For example, Gong et al. \cite{Gong2013} propose an improved version of FLICM, namely KWFLICM, which augments a tradeoff weighted fuzzy factor and a kernel metric into FCM. Even though it is generally robust to extensive noise, it is more \mbox{time-consuming} than most of existing FCMs. Zhao et al. \cite{Zhao2014} take a neighborhood weighted distance into account, thus presenting a novel algorithm called NWFCM. Although it runs faster than KWFLICM, its segmentation performance is worse. Moreover, it exhibits lower computational efficiency than other FCMs. More recently, Wang et al. \cite{Wang2019} consider tight wavelet frames as a kernel function so as to present wavelet \mbox{frame-based} FCM (WFCM), which takes full advantage of the feature extraction capacity of tight wavelet frames. In spite of its rarely low computational cost, its segmentation effects can be further improved by using various mathematical techniques.

\subsection{Comprehensive FCM}

To keep a sound trade-off between performance and speed of clustering, comprehensive FCMs involving various mathematical techniques has been put forward \cite{Gharieb2017,Wang2020KL,Gu2018,Wang2020SR,Vincent1993,Najman1996,Chen2012,Lei2018,Lei2019,Wang2020RS,Bai2019}. For instance, Gharieb et al. \cite{Gharieb2017} present an FCM framework based on Kullback-Leibler (KL) divergence. It uses KL divergence to optimize the membership similarity between a pixel and its neighbors. Yet it has slow clustering speed. Gu et al. \cite{Gu2018} report a fuzzy double C-Means algorithm (FDCM) through the utility of sparse representation, which addresses two datasets simultaneously, i.e., a basic feature set associated with an observed image and a feature set learned from a spare self-representation model. Overall, FDCM is robust and applicable to a wide range of image segmentation problems. However, its computational efficiency is not satisfactory. Lei et al. \cite{Lei2018} present a fast and robust FCM algorithm (FRFCM) by using gray level histograms and morphological gray reconstruction. In spite of its fast clustering, its performance is sometimes unstable since morphological gray reconstruction may cause the loss of useful image features. More recently, Lei et al. \cite{Lei2019} propose an automatic fuzzy clustering framework (AFCF) by incorporating threefold techniques, i.e., superpixel algorithms, density peak clustering and prior entropy. It overcomes two difficulties in existing algorithms \cite{Wang2019,Gu2018,Lei2018}. One is to select the number of clusters automatically. The other one is to employ superpixel algorithms and the prior entropy to improve image segmentation performance. However, AFCF's results are unstable.

In this work, the proposed algorithm differs from all algorithms mentioned above in the sense that we take a wide range of mixed noise estimations as the starting point and directly minimize the objective function of WRFCM formulated by using fidelity without dictionary learning and CNNs and archives outstanding performance in image segmentation tasks.

\section{FCM and Proposed Methodology}\label{sec:methodology}
\subsection{Fuzzy \emph{C}-Means (FCM)}
Given a set $\bm{X}=\{\bm{x}_{j}\in\mathbb{R}^{L}:j=1,2,\cdots,K\}$, where $\bm{x}_{j}$ contains $L$ channels, i.e., $\bm{x}_{j} =(x_{j1},x_{jl},\cdots,x_{jL})^{T}$. FCM is applied to cluster $\bm{X}$ by minimizing:
\begin{equation} \label{GrindEQ__1_}
J({\bm U},{\bm V})=\sum\limits_{i=1}^{c}\sum\limits_{j=1}^{K}u_{ij}^{m} \|{\bm x}_{j} -{\bm v}_{i} \|^{2}
\end{equation}
where ${\bm U}=[u_{ij}]_{c\times K} $ is a partition matrix under a constraint $\sum_{i=1}^{c}u_{ij}=1$ for $j=1,2,\cdots,K$, $\bm{V}=\{\bm{v}_{i}:i=1,2,\cdots,c\}$ is a prototype set, $\|\cdot\|$ stands for Euclidean distance, and $m$ denotes a fuzzification exponent ($m>1$).

An alternating iteration scheme \cite{Bezdek1984} is used to minimize \eqref{GrindEQ__1_}. Each iteration is realized as follows:
\begin{equation*}
u_{ij}^{(t+1)} =\frac{(\|{\bm x}_{j} -{\bm v}_{i}^{(t)} \|^{2} )^{-1/(m-1)} }{\sum\limits_{q=1}^{c}(\|{\bm x}_{j} -{\bm v}_{q}^{(t)} \|^{2} )^{-1/(m-1)} }
\end{equation*}

\begin{equation*}
v_{il}^{(t+1)} =\frac{\sum\limits_{j=1}^{K}\left(u_{ij}^{(t+1)} \right)^{m} x_{jl}  }{\sum\limits_{j=1}^{K}\left(u_{ij}^{(t+1)} \right)^{m}}
\end{equation*}
Here, $t=0,1,2,\cdots$ is an iterative step and $l=1,2,\cdots,L$. By presetting a threshold $\varepsilon$, the procedure stops when $\|{\bm U}^{(t+1)} -{\bm U}^{(t)} \|<\varepsilon$.

\subsection{Noise Model}
Consider an observed image $\bm{X}$ with $K$ pixels. It is denoted as ${\bm X}=\{ {\bm x}_{j} :j=1,2,\cdots ,K\} $, where ${\bm x}_{j} =\{ x_{jl}:l=1,2,\cdots,L\} $. When $L=1$, ${\bm X}$ represents a gray image. For $L=3$, ${\bm X}$ is a Red-Green-Blue color image. Since there is noise in an observed image, ${\bm X}$ can be modeled as a sum of a noise-free image $\widetilde{{\bm X}}$ and noise ${\bm R}$:
\begin{equation} \label{GrindEQ__2_}
{\bm X}=\widetilde{{\bm X}}+{\bm R}
\end{equation}

Mathematically speaking, $\widetilde{{\bm X}}=\{\tilde{{\bm x}}_{1} ,\tilde{{\bm x}}_{2} ,\cdots ,\tilde{{\bm x}}_{K}\}$ is an ideal value of ${\bm X}$ and thus is unknown. ${\bm R}=\{ {\bm r}_{1} ,{\bm r}_{2} ,\cdots ,{\bm r}_{K} \} $ is viewed as the residual between $\bm{X}$ and $\widetilde{\bm{X}}$. Its accurate estimation can make $\widetilde{{\bm X}}$ instead of ${\bm X}$ participate in clustering so as to improve FCM's robustness. Hence, it is a necessary step to formulate a noise model before constructing an FCM model. In image processing, the models of single noise such as Gaussian, Poisson and impulse noise are widely used. In this work, in order to construct robust FCM, we mostly consider mixed or unknown noise since it is often encountered in real-world applications. Its specific model is unfortunately hard to be formulated. Therefore, a common solution is to assume the type of mixed noise in advance. In universal image processing, two kinds of mixed noise are the most common, refer to mixed Poisson-Gaussian noise and mixed Gaussian and impulse noise. Beyond them, we focus on a mixture of a wide range of noise, i.e., a mixture of Poisson, Gaussian, and impulse noise. We investigate an FCM-related model based on the analysis of the mixed noise model and extend it to image segmentation with mixed or unknown noise.

Formally speaking, a noise-free image $\widetilde{{\bm X}}$ is defined in a domain $\Omega =\{1,2,\cdots,K\}$. It is first corrupted by Poisson noise, thus resulting in $\overline{{\bm X}}=\{ \bar{{\bm x}}_{1} ,\bar{{\bm x}}_{2} ,\cdots ,\bar{{\bm x}}_{K} \} $ that obeys a Poisson distribution, or, $\overline{{\bm X}}\sim {\bf P}(\widetilde{{\bm X}})$. Then additive zero-mean white Gaussian noise ${\bm R}'=\{ {\bm r}'_{1} ,{\bm r}'_{2} ,\cdots ,{\bm r}'_{K} \} $ with standard deviation $\sigma$ is added. Finally, impulse noise ${\bm R}''=\{ {\bm r}''_{j} ,{\bm r}''_{2} ,\cdots ,{\bm r}''_{K} \}$ with a given probability $p\in (0,1)$ is imposed. Hence, for $j\in \Omega $, an arbitrary element in observed image $\bm X$ is expressed as:
\begin{equation} \label{GrindEQ__3_}
{\bm x}_{j} =\left\{\begin{array}{l} {\bar{{\bm x}}_{j} +{\bm r}'_{j} {\kern 1pt} {\kern 1pt} {\kern 1pt} {\kern 1pt} {\kern 1pt} {\kern 1pt} j\in \Omega _{1} } \\ {{\bm r}''_{j} {\kern 1pt} {\kern 1pt} {\kern 1pt} {\kern 1pt} {\kern 1pt} {\kern 1pt} {\kern 1pt} {\kern 1pt} {\kern 1pt} {\kern 1pt} {\kern 1pt} {\kern 1pt} {\kern 1pt} {\kern 1pt} {\kern 1pt} {\kern 1pt} {\kern 1pt} {\kern 1pt} {\kern 1pt} {\kern 1pt} {\kern 1pt} {\kern 1pt} {\kern 1pt} {\kern 1pt} {\kern 1pt} {\kern 1pt} {\kern 1pt} j\in \Omega_{2} :=\Omega \backslash \Omega _{1} } \end{array}\right.
\end{equation}
where the subset $\Omega_{2} $ of $\Omega $ denotes the region including the missing information of $\overline{{\bm X}}$ and is assumed to be unknown with each element being drawn from the whole region $\Omega$ by Bernoulli trial with $p$. In image segmentation, mixed noise model (\ref{GrindEQ__3_}) is for the first time presented.

\subsection{Residual-driven FCM}
Since there exists an unknown amount of noise in an observed image, the segmentation accuracy of FCM is greatly impacted without properly handling it. It is natural to understand that taking a noise-free image (the ideal value of an observed image) as data to be clustered can achieve better segmentation effects. In other words, if noise (residual) can be accurately estimated, the segmentation effects of FCM should be greatly improved. To do so, we introduce a fidelity term on residual into the objective function of FCM. Consequently, an RFCM framework is first presented:
\begin{equation} \label{GrindEQ__4_}
J({\bm U},{\bm V},{\bm R})=\sum _{i=1}^{c}\sum _{j=1}^{K}u_{ij}^{m} \|{\bm x}_{j} -{\bm r}_{j} -{\bm v}_{i} \|^{2}+{\bm \beta }\cdot \Gamma ({\bm R})
\end{equation}
where ${\bm \beta }=\{\beta _{l} :l=1,2,\cdots ,L\}$ is a parameter set, which controls the impact of fidelity term $\Gamma ({\bm R})$ on FCM. We rewrite ${\bm R}$ as $\{{\bm R}_{l} :l=1,2,\cdots ,L\}$ with ${\bm R}_{l}=(r_{1l} ,r_{2l} ,\cdots ,r_{Kl} )^{T}$, which indicates that ${\bm R}$ has $L$ channels and each of them contains $K$ pixels. In this work, $L=1$ (gray) or $3$ (Red-Green-Blue). From a channel perspective, we have:
\begin{equation} \label{GrindEQ__5_}
{\bm \beta }\cdot \Gamma ({\bm R})=\sum _{l=1}^{L}\beta _{l} \Gamma ({\bm R}_{l} )
\end{equation}

The fidelity term $\Gamma ({\bm R})$ guarantees that the solution accords with the degradation process of the minimization of \eqref{GrindEQ__4_}. It is determined by a specified noise distribution. For example, when considering Gaussian noise estimation, we use an $\ell_{2} $-norm fidelity term:
\[\Gamma ({\bm R}_{l} )=\|{\bm R}_{l} \|_{\ell _{2} }^{2} =\sum_{j=1}^{K}|r_{jl} |^{2}\]
where $\|\cdot \|_{\ell_{2} } $ stands for an $\ell _{2}$ vector norm. In the presence of impulse noise, we choose an $\ell_{1}$-norm fidelity term:
\[\Gamma ({\bm R}_{l} )={\kern 1pt} {\kern 1pt} \|{\bm R}_{l} \|_{\ell _{1} } =\sum _{j=1}^{K}|r_{jl} | \]
where $\|\cdot \|_{\ell_{1}}$ denotes an $\ell_{1}$ vector norm. For Poisson noise, we take the Csisz\'{a}r's I-divergence \cite{Le2007} of $\bm R$ from $\bm X$ as a fidelity term, i.e.,
\[\Gamma ({\bm R}_{l} )=\sum _{j=1}^{K}\left((x_{jl} -r_{jl} )-x_{jl} \log (x_{jl} -r_{jl} )\right) \]

For common single noise, i.e., Gaussian, Poisson, and impulse noise, the above fidelity terms lead to a maximum a posteriori (MAP) solution to such noise estimations. In \mbox{real-world} applications, images are generally contaminated by mixed or unknown noise rather than a single noise. The fidelity terms for single noise estimation become inapplicable since the distribution of mixed or unknown noise is difficult to be modeled mathematically. Therefore, one of the main purposes of this work is to design a universal fidelity term for mixed or unknown noise estimation.

\subsection{Analysis of Mixed Noise Distribution}
To reveal the essence of mixed noise distributions, we here consider generic and representative mixed noise, i.e., a mixture of Poisson, Gaussian, and impulse noise. Let us take an example to exhibit its distribution. Here, we impose Gaussian noise ($\sigma =10$) and a mixture of Poisson, Gaussian ($\sigma=10$) and random-valued impulse noise ($p=20\%$) on image `Lena' with size $512\times512$, respectively. We show original and two observed images in Fig. \ref{fig:noise}.

\begin{figure}[htb]
\centering
\begin{minipage}[t]{0.3\linewidth}
\centering
\includegraphics[width=1\textwidth]{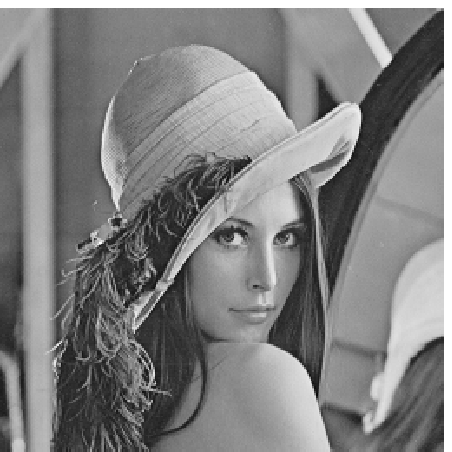}
%\centerline{\footnotesize (a)}
\end{minipage}
\begin{minipage}[t]{0.3\linewidth}
\centering
\includegraphics[width=1\textwidth]{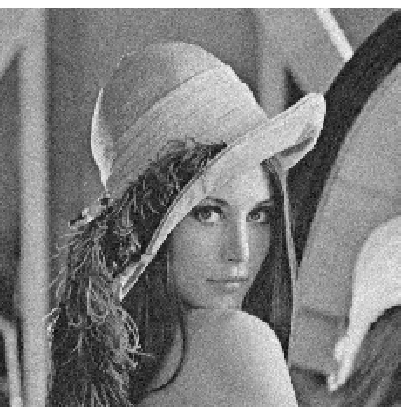}
%\centerline{\footnotesize (b)}
\end{minipage}
%\begin{minipage}[t]{0.19\linewidth}
%\centering
%\includegraphics[width=1\textwidth]{Lena_PG.eps}
%%\centerline{\footnotesize (c)}
%\end{minipage}
%\begin{minipage}[t]{0.19\linewidth}
%\centering
%\includegraphics[width=1\textwidth]{Lena_GI.eps}
%%\centerline{\footnotesize (d)}
%\end{minipage}
\begin{minipage}[t]{0.3\linewidth}
\centering
\includegraphics[width=1\textwidth]{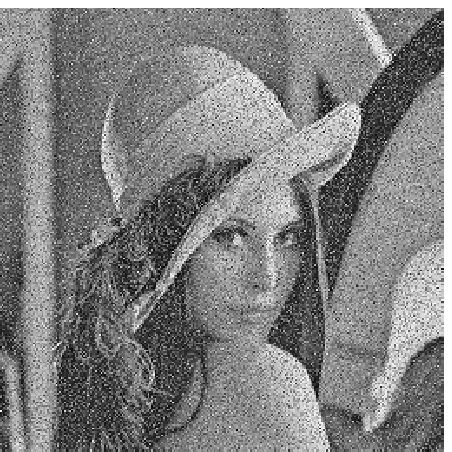}
%\centerline{\footnotesize (e)}
\end{minipage}
\begin{minipage}[t]{0.3\linewidth}
\centering
\includegraphics[width=1\textwidth]{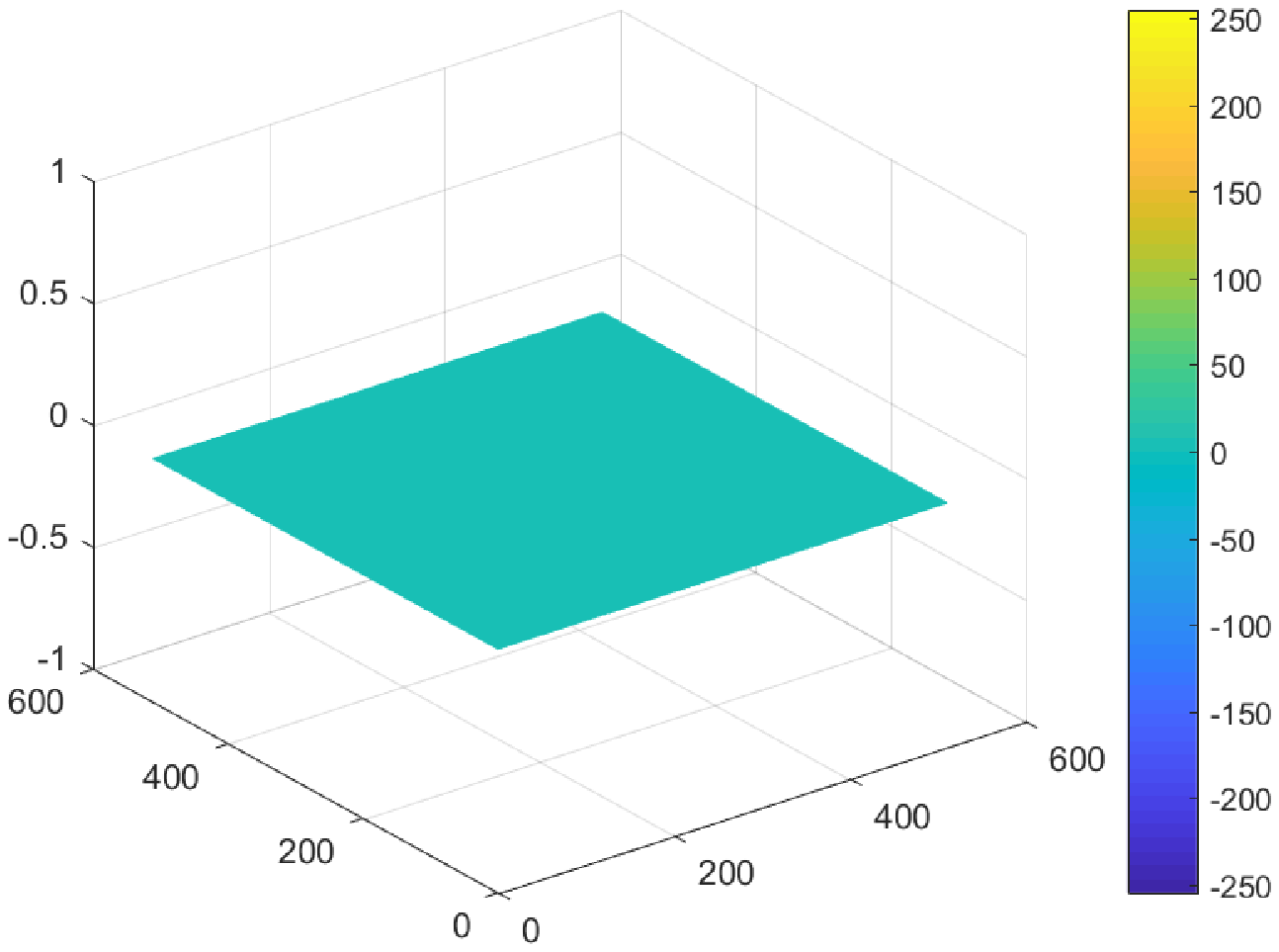}
\centerline{ \footnotesize (a)}
\end{minipage}
\begin{minipage}[t]{0.3\linewidth}
\centering
\includegraphics[width=1\textwidth]{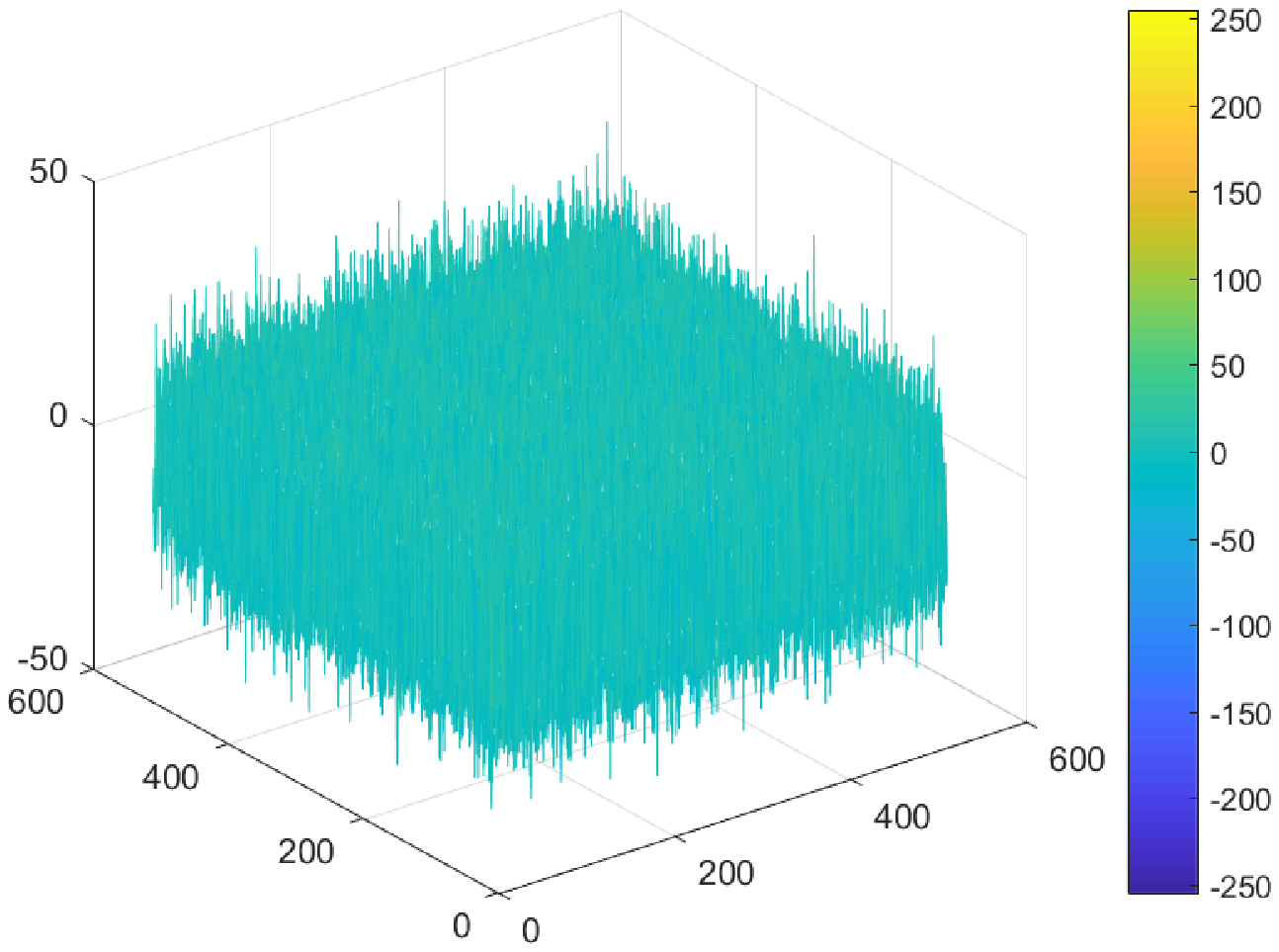}
\centerline{ \footnotesize (b)}
\end{minipage}
%\begin{minipage}[t]{0.19\linewidth}
%\centering
%\includegraphics[width=1\textwidth]{noise_Lena_PG.eps}
%\centerline{ \footnotesize (c)}
%\end{minipage}
%\begin{minipage}[t]{0.19\linewidth}
%\centering
%\includegraphics[width=1\textwidth]{noise_Lena_GI.eps}
%\centerline{ \footnotesize (d)}
%\end{minipage}
\begin{minipage}[t]{0.3\linewidth}
\centering
\includegraphics[width=1\textwidth]{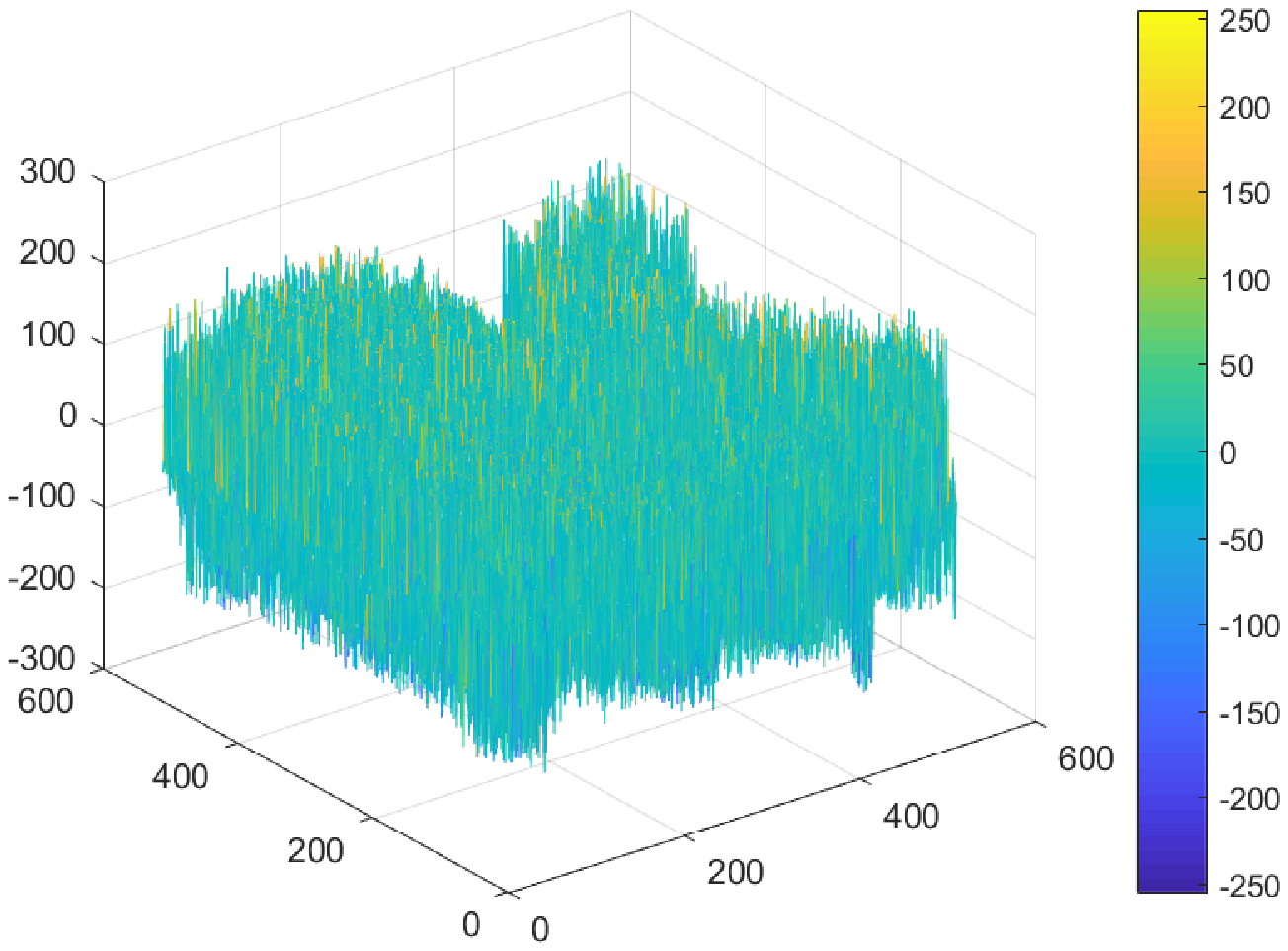}
\centerline{ \footnotesize (c)}
\end{minipage}
\caption{Noise-free image and two observed ones corrupted by Gaussian and mixed noise, respectively. The first row: (a) noise-free image; (b) observed image with Gaussian noise; and (c) observed image with mixed noise. The second row portrays noise included in three images.}
\label{fig:noise}
\end{figure}

As Fig. \ref{fig:noise}(b) shows, Gaussian noise is overall organized. As a common sense, Poisson distribution is a Gaussian-like one under the condition of enough samples. Therefore, due to impulse noise, mixed noise is disorganized as shown in Fig. \ref{fig:noise}(c). In Fig. \ref{fig:distribution}, we portray the distributions of Gaussian and mixed noise, respectively.

\begin{figure}[htb]
\setlength{\abovecaptionskip}{0pt}
\setlength{\belowcaptionskip}{0pt}
\centering
\begin{minipage}[t]{1\linewidth}
\centering
\includegraphics[width=1\textwidth]{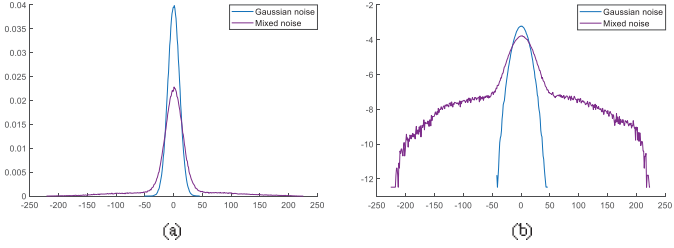}
%\centerline{(a)}
\end{minipage}
\caption{Distributions of Gaussian and mixed noise in different domains. (a) linear domain; and (b) logarithmic domain.}
\label{fig:distribution}
\end{figure}

Fig. \ref{fig:distribution}(a) shows noise distribution in a linear domain. To illustrate a heavy tail intuitively, we present it in a logarithmic domain as shown in Fig. \ref{fig:distribution}(b). Clearly, Poisson noise leads to a Gaussian-like distribution. Nevertheless, impulse noise gives rise to a more irregular distribution with a heavy tail. Therefore, neither $\ell_{1}$ norm nor $\ell_{2}$ norm can precisely characterize the residual $\bm{R}$ in the sense of the MAP estimation.

\subsection{Residual-driven FCM with Weighted $\ell_{2}$-norm Fidelity}
Intuitively, if the fidelity term can be modified so as to make mixed noise distribution more Gaussian-like, we can still use $\ell_{2}$ norm to characterize residual $\bm{R}$. It means that mixed noise can be more accurately estimated. Therefore, we adopt robust estimation techniques \cite{Huber1973,Huber1981} to weaken the heavy tail, which makes mixed noise distribution more regular. In the sequel, we assign a proper weight $w_{jl} $ to each residual $r_{jl} $, which forms a weighted residual $w_{jl} r_{jl} $ that almost obeys a Gaussian distribution. Given Fig. \ref{fig:weightdistribution}, we use an example for showing the effect of weighting.

\begin{figure}[htb]
\setlength{\abovecaptionskip}{0pt}
\setlength{\belowcaptionskip}{0pt}
\centering
\begin{minipage}[t]{1\linewidth}
\centering
\includegraphics[width=1\textwidth]{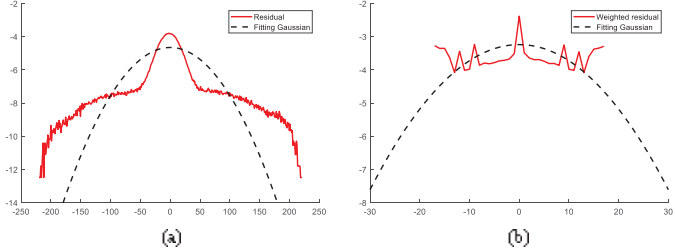}
%\centerline{(a)}
\end{minipage}
\caption{Distributions of residual $r_{jl} $ and weighted residual $w_{jl} r_{jl} $, as well as the fitting Gaussian function in the logarithmic domain.}
\label{fig:weightdistribution}
\end{figure}

Fig. \ref{fig:weightdistribution}(a) shows the distribution of $r_{jl} $ and the fitting Gaussian function based on the variance of $r_{jl} $. Fig. \ref{fig:weightdistribution}(b) gives the distribution of $w_{jl} r_{jl} $ and the fitting Gaussian function based on the variance of $w_{jl} r_{jl} $. Clearly, the distribution of $w_{jl} r_{jl} $ in Fig. \ref{fig:weightdistribution}(b) is more Gaussian-like than that in Fig. \ref{fig:weightdistribution}(a), which means that $\ell_{2} $-norm fidelity can work on weighted residual $w_{jl} r_{jl} $ for a MAP-like solution of $\bm{R}$.

By analyzing Fig. \ref{fig:weightdistribution}, for $l=1,2,\cdots ,L$, we propose a weighted $\ell _{2} $-norm fidelity term for mixed or unknown noise estimation:
\begin{equation} \label{GrindEQ__6_}
\Gamma ({\bm R}_{l} )= \|{\bm W}_{l} \circ {\bm R}_{l} \|_{\ell _{2} }^{2} =\sum_{j=1}^{K}|w_{jl} r_{jl} |^{2}
\end{equation}
where $\circ $ performs element-by-element multiplication of ${\bm R}_{l} =(r_{1l} ,r_{2l} ,\cdots ,r_{Kl} )^{T} $ and ${\bm W}_{l} =(w_{1l} ,w_{2l} ,\cdots ,w_{Kl} )^{T} $. For $l=1,2,\cdots ,L$, ${\bm W}_{l} $ makes up a weight matrix ${\bm W}=[w_{jl}]_{K\times L} $. Each element $w_{jl} $ is assigned to location $(j,l)$. Since it is inversely proportional to residual $r_{jl} $, it can be automatically determined. In this work, we adopt the following expression:
\begin{equation} \label{GrindEQ__7_}
w_{jl} =e^{-\xi r_{jl}^{2} }
\end{equation}
where $\xi $ is a positive parameter, which aims to control the decreasing rate of $w_{jl} $.

By substituting \eqref{GrindEQ__6_} into \eqref{GrindEQ__4_} combined with \eqref{GrindEQ__5_}, we present RFCM with weighted $\ell_{2} $-norm fidelity (WRFCM) for image segmentation:
\begin{equation} \label{GrindEQ__8_}
\begin{array}{l} {J(\bm{U},\bm{V},\bm{R},\bm{W})}\\
{{\kern 5pt}=\sum\limits_{i=1}^{c}\sum\limits_{j=1}^{K}u_{ij}^{m} \|{\bm x}_{j} -{\bm r}_{j} -{\bm v}_{i} \|^{2}+\sum\limits_{l=1}^{L}\beta_{l} \|{\bm W}_{l} \circ {\bm R}_{l} \|_{\ell _{2} }^{2}  } \\ {{\kern 5pt}=\sum\limits_{i=1}^{c}\sum\limits_{j=1}^{K}u_{ij}^{m} \|{\bm x}_{j} -{\bm r}_{j} -{\bm v}_{i} \|^{2}+\sum\limits_{l=1}^{L}\beta _{l} \sum\limits_{j=1}^{K}|w_{jl} r_{jl} |^{2}   } \end{array}
\end{equation}

When coping with image segmentation problems, since each image pixel is closely related to its neighbors, using spatial information has a positive impact on FCM as shown in \cite{Ahmed2002,Zhang2019}. If there exists a small distance between a target pixel and its neighbors, they most likely belong to a same cluster. Therefore, we introduce spatial information into \eqref{GrindEQ__8_}, thus resulting in our final objective function:
\begin{equation} \label{GrindEQ__9_}
\begin{array}{l} {J({\bm U},{\bm V},{\bm R},{\bm W})=\sum\limits_{i=1}^{c}\sum\limits_{j=1}^{K}u_{ij}^{m} \left(\sum\limits_{n\in {\rm {\mathcal N}}_{j} }\frac{\|{\bm x}_{n} -{\bm r}_{n} -{\bm v}_{i} \|^{2} }{1+d_{nj} }  \right)} \\ {{\kern 1pt} {\kern 1pt} {\kern 1pt} {\kern 1pt} {\kern 1pt} {\kern 1pt} {\kern 1pt} {\kern 1pt} {\kern 1pt} {\kern 1pt} {\kern 1pt} {\kern 1pt} {\kern 1pt} {\kern 1pt} {\kern 1pt} {\kern 1pt} {\kern 1pt} {\kern 1pt} {\kern 1pt} {\kern 1pt} {\kern 1pt} {\kern 1pt} {\kern 1pt} {\kern 1pt} {\kern 1pt} {\kern 1pt} {\kern 1pt} {\kern 1pt} {\kern 1pt} {\kern 1pt} {\kern 1pt} {\kern 1pt} {\kern 1pt} {\kern 1pt} {\kern 1pt} {\kern 1pt} {\kern 1pt} {\kern 1pt} {\kern 1pt} {\kern 1pt} {\kern 1pt} {\kern 1pt} {\kern 1pt} {\kern 1pt} {\kern 1pt} {\kern 1pt} {\kern 1pt} {\kern 1pt} {\kern 1pt} {\kern 1pt} {\kern 1pt} {\kern 1pt} {\kern 1pt} {\kern 1pt} {\kern 1pt} {\kern 1pt} {\kern 1pt} {\kern 1pt} {\kern 1pt} {\kern 1pt} {\kern 1pt} {\kern 1pt} {\kern 1pt} {\kern 1pt} {\kern 1pt} {\kern 1pt} {\kern 1pt} {\kern 1pt} {\kern 1pt} {\kern 1pt} {\kern 1pt} {\kern 1pt} {\kern 1pt} {\kern 1pt} {\kern 1pt} {\kern 1pt} {\kern 1pt} {\kern 1pt} +\sum\limits_{l=1}^{L}\beta _{l} \sum\limits_{j=1}^{K}\sum\limits_{n\in {\rm {\mathcal N}}_{j} }\frac{|w_{nl} r_{nl} |^{2} }{1+d_{nj} } } \end{array}
\end{equation}

In \eqref{GrindEQ__9_}, an image pixel is sometimes loosely represented by its corresponding index even though this is not ambiguous. Thus, $n$ is a neighbor pixel of $j$, ${\rm {\mathcal N}}_{j} $ stands for a local window centralized in $j$, and $d_{nj} $ represents the Euclidean distance between $n$ and $j$.

\subsection{Minimization Algorithm}
Minimizing \eqref{GrindEQ__9_} involves four unknowns, i.e., $\bm U$, $\bm V$, $\bm R$ and $\bm W$. According to \eqref{GrindEQ__7_}, $\bm W$ is automatically determined by $\bm R$. Hence, we can design a two-step iterative algorithm to minimize \eqref{GrindEQ__9_}, which fixes $\bm W$ first to solve $\bm U$, $\bm V$ and $\bm R$, then uses $\bm R$ to update $\bm W$. The main task in each iteration is to solve the minimization problem in terms of $\bm U$, $\bm V$ and $\bm R$ when fixing $\bm W$. Assume that $\bm W$ is given. We can apply a Lagrangian multiplier method to minimize \eqref{GrindEQ__9_}. The Lagrangian function is expressed as:
\begin{footnotesize}
\begin{equation} \label{GrindEQ__10_}
\begin{array}{l} {{\rm {\mathcal L}}_{\Lambda } ({\bm U},{\bm V},{\bm R};{\bm W})=\sum\limits_{i=1}^{c}\sum\limits_{j=1}^{K}u_{ij}^{m} \left(\sum\limits_{n\in {\it {\mathcal N}}_{j} }\frac{\|{\bm x}_{n} -{\bm r}_{n} -{\bm v}_{i} \|^{2} }{1+d_{nj} }  \right)  } \\ { {\kern 15pt} +\sum\limits_{l=1}^{L}\beta_{l} \sum\limits_{j=1}^{K}\sum\limits_{n\in {\it {\mathcal N}}_{j} }\frac{|w_{nl} r_{nl} |^{2} }{1+d_{nj} } +\sum\limits_{j=1}^{K}\lambda_{j} \left(\sum\limits_{i=1}^{c}u_{ij}-1\right) } \end{array}
\end{equation}
\end{footnotesize}where $\Lambda =\{ \lambda_{j}:j=1,2,\cdots ,K\} $ is a set of Lagrangian multipliers. The two-step iterative algorithm for minimizing \eqref{GrindEQ__9_} is realized in Algorithm \ref{algorithm1}.

%\noindent\begin{tabular}{|p{3.39in}|}\hline
%\textbf{Algorithm 1. Two-step iterative algorithm\newline }Given a threshold $\varepsilon$, input ${\bm W}^{(0)} $. For $t=0,1,\cdots $, iterate: \\
%\quad\textbf{Step 1}: Find minimizers ${\bm U}^{(t+1)} $, ${\bm V}^{(t+1)} $, and ${\bm R}^{(t+1)} $:\newline
%\begin{footnotesize}
%\begin{equation}\label{GrindEQ__11_}
%\left({\bm U}^{(t+1)} ,{\bm V}^{(t+1)} ,{\bm R}^{(t+1)} \right)=\arg {\mathop{\min }\limits_{{\bm U},{\bm V},{\bm R}}} {\rm {\mathcal L}}_{\Lambda } ({\bm U},{\bm V},{\bm R};{\bm W}^{(t)} )
%\end{equation}
%\end{footnotesize}\quad\textbf{Step 2}: Update the weight matrix ${\bm W}^{(t{\rm +}1)} $\newline
%If $\|{\bm U}^{(t{\rm +}1)} -{\bm U}^{(t)}\|<\varepsilon $, stop; else update $t$ such that \\
%\quad\quad\quad\quad\quad\quad\quad\quad\quad\quad$0\le t\uparrow <+\infty.$
%\\\hline
%\end{tabular}
%\\

\begin{algorithm}[htb]
\caption{Two-step iterative algorithm}
\begin{algorithmic}
\STATE Given a threshold $\varepsilon$, input ${\bm W}^{(0)} $. For $t=0,1,\cdots $, iterate:
\STATE \quad\textbf{Step 1}: Find minimizers ${\bm U}^{(t+1)} $, ${\bm V}^{(t+1)} $, and ${\bm R}^{(t+1)} $:
\begin{footnotesize}
\begin{equation}\label{GrindEQ__11_}
\left({\bm U}^{(t+1)} ,{\bm V}^{(t+1)} ,{\bm R}^{(t+1)} \right)=\arg {\mathop{\min }\limits_{{\bm U},{\bm V},{\bm R}}} {\rm {\mathcal L}}_{\Lambda } ({\bm U},{\bm V},{\bm R};{\bm W}^{(t)} )
\end{equation}
\end{footnotesize}
\STATE \quad\textbf{Step 2}: Update the weight matrix ${\bm W}^{(t{\rm +}1)} $
\STATE If $\|{\bm U}^{(t{\rm +}1)} -{\bm U}^{(t)}\|<\varepsilon $, stop; else update $t$ such that \\
\quad\quad\quad\quad\quad\quad\quad\quad\quad\quad$0\le t\uparrow <+\infty$
\end{algorithmic}
\label{algorithm1}
\end{algorithm}

The minimization problem \eqref{GrindEQ__11_} can be divided into the following three subproblems:
\begin{equation} \label{GrindEQ__12_}
\left\{\begin{array}{l} {{\bm U}^{(t+1)} =\arg {\mathop{\min }\limits_{{\bm U}}} {\rm {\mathcal L}}_{\Lambda } ({\bm U},{\bm V}^{(t)} ,{\bm R}^{(t)} ;{\bm W}^{(t)} )} \\ {{\bm V}^{(t+1)} =\arg {\mathop{\min }\limits_{{\bm V}}} {\rm {\mathcal L}}_{\Lambda } ({\bm U}^{(t+1)} ,{\bm V},{\bm R}^{(t)} ;{\bm W}^{(t)} )} \\ {{\bm R}^{(t+1)} =\arg {\mathop{\min }\limits_{{\bm R}}} {\rm {\mathcal L}}_{\Lambda } ({\bm U}^{(t+1)} ,{\bm V}^{(t+1)} ,{\bm R};{\bm W}^{(t)} )} \end{array}\right.
\end{equation}

Each subproblem in \eqref{GrindEQ__12_} has a closed-form solution. We use an alternative optimization scheme similar to the one used in FCM to optimize $\bm{U}$ and $\bm{V}$. The following result is needed to obtain the iterative updates of $\bm{U}$ and $\bm{V}$.

\begin{thm}\label{lem1}
Consider the first two subproblems of \eqref{GrindEQ__12_}. By applying the Lagrangian multiplier method to solve them, the iterative solutions are presented as:
\begin{equation} \label{GrindEQ__13_}
u_{ij}^{(t+1)} =\frac{\left(\sum\limits_{n\in {\rm {\mathcal N}}_{j} }\frac{\|{\bm x}_{n} -{\bm r}_{n}^{(t)} -{\bm v}_{i}^{(t)} \|^{2} }{1+d_{nj} }  \right)^{-1/(m-1)} }{\sum\limits_{q=1}^{c}\left(\sum\limits_{n\in {\rm {\mathcal N}}_{j} }\frac{\|{\bm x}_{n} -{\bm r}_{n}^{(t)} -{\bm v}_{q}^{(t)} \|^{2} }{1+d_{nj} }  \right)^{-1/(m-1)}  }
\end{equation}
\begin{equation} \label{GrindEQ__14_}
{\bm v}_{i}^{(t+1)} =\frac{\sum\limits_{j=1}^{K}\left(\left(u_{ij}^{(t+1)} \right)^{m} \sum\limits_{n\in {\rm {\mathcal N}}_{j} }\frac{{\bm x}_{n} -{\bm r}_{n}^{(t)} }{1+d_{nj} }  \right) }{\sum\limits_{j=1}^{K}\left(\left(u_{ij}^{(t+1)} \right)^{m} \sum\limits_{n\in {\rm {\mathcal N}}_{j} }\frac{1}{1+d_{nj} }  \right) }
\end{equation}
\end{thm}

\begin{proof}
See Appendix.
\end{proof}

In the last subproblem of \eqref{GrindEQ__12_}, both ${\bm r}_{j} $ and ${\bm r}_{n} $ appear simultaneously. Since ${\bm r}_{n} $ is dependent to ${\bm r}_{j} $, it should not be considered as a constant vector. In other words, $n$ is one of neighbors of $j$ while $j$ is one of neighbors of $n$ symmetrically. Thus, $n\in {\rm {\mathcal N}}_{j} $ is equivalent to $j\in {\rm {\mathcal N}}_{n} $. Thus we have:
\begin{footnotesize}
\begin{equation} \label{GrindEQ__15_}
\sum _{j=1}^{K}u_{ij}^{m} \left(f({\bm r}_{j} )+\sum _{\begin{array}{l} {n\in {\rm {\mathcal N}}_{j} } \\ {{\kern 1pt} {\kern 1pt} n\ne j} \end{array}}f({\bm r}_{n} ) \right) =\sum _{j=1}^{K}\sum\limits_{n\in {\rm {\mathcal N}}_{j} }u_{in}^{m} (f({\bm r}_{j} ))
\end{equation}
\end{footnotesize}where $f$ represents a function in terms of ${\bm r}_{j} $ or ${\bm r}_{n} $. By \eqref{GrindEQ__15_}, we rewrite \eqref{GrindEQ__9_} as
\begin{equation} \label{GrindEQ__16_}
\begin{array}{l} {J({\bm U},{\bm V},{\bm R},{\bm W})=\sum\limits_{i=1}^{c}\sum\limits_{j=1}^{K}\sum\limits_{n\in {\rm {\mathcal N}}_{j} }\frac{u_{in}^{m} \|{\bm x}_{j} -{\bm r}_{j} -{\bm v}_{i} \|^{2} }{1+d_{nj} }}\\ {{\kern 1pt} {\kern 1pt} {\kern 1pt} {\kern 1pt} {\kern 1pt} {\kern 1pt} {\kern 1pt} {\kern 1pt} {\kern 1pt} {\kern 1pt} {\kern 1pt} {\kern 1pt} {\kern 1pt} {\kern 1pt} {\kern 1pt} {\kern 1pt} {\kern 1pt} {\kern 1pt} {\kern 1pt} {\kern 1pt} {\kern 1pt} {\kern 1pt}{\kern 1pt} {\kern 1pt} {\kern 1pt} {\kern 1pt} {\kern 1pt} {\kern 1pt} {\kern 1pt} {\kern 1pt} {\kern 1pt} {\kern 1pt} {\kern 1pt} {\kern 1pt} {\kern 1pt} {\kern 1pt} {\kern 1pt} {\kern 1pt} {\kern 1pt} {\kern 1pt} {\kern 1pt} {\kern 1pt}{\kern 1pt} {\kern 1pt} {\kern 1pt} {\kern 1pt} {\kern 1pt} {\kern 1pt} {\kern 1pt} {\kern 1pt} {\kern 1pt} {\kern 1pt} {\kern 1pt} {\kern 1pt} {\kern 1pt} {\kern 1pt} {\kern 1pt} {\kern 1pt} {\kern 1pt} {\kern 1pt} {\kern 1pt} {\kern 1pt} {\kern 1pt} {\kern 1pt} {\kern 1pt}{\kern 1pt} {\kern 1pt} {\kern 1pt}+\sum\limits_{l=1}^{L}\beta _{l} \sum\limits_{j=1}^{K}\sum\limits_{n\in {\rm {\mathcal N}}_{j} }\frac{|w_{jl} r_{jl} |^{2} }{1+d_{nj} } }  \end{array}
\end{equation}

According to the two-step iterative algorithm, we assume that ${\bm W}$ in \eqref{GrindEQ__16_} is fixed in advance. When $\bm{U}$ and $\bm{V}$ are updated, the last subproblem of \eqref{GrindEQ__12_} is separable and can be decomposed into $K\times L$ subproblems:
\begin{footnotesize}
\begin{equation} \label{GrindEQ__17_}
\begin{array}{l} {r_{jl}^{(t+1)} =\arg {\mathop{\min }\limits_{r_{jl} }} \sum\limits_{i=1}^{c}\left(\sum\limits_{n\in {\rm {\mathcal N}}_{j} }\frac{\left(u_{in}^{(t+1)} \right)^{m} \|x_{jl} -r_{jl} -v_{il}^{(t+1)} \|^{2} }{1+d_{nj} }  \right) }\\{{\kern 60pt} +\sum\limits_{n\in {\rm {\mathcal N}}_{j} }\frac{\beta _{l} |w_{jl}^{(t)} r_{jl} |^{2} }{1+d_{nj} }}  \end{array}
\end{equation}
\end{footnotesize}By zeroing the gradient of the energy function in \eqref{GrindEQ__17_} in terms of $r_{jl} $, the iterative solution to \eqref{GrindEQ__17_} is expressed as:
\begin{equation} \label{GrindEQ__18_}
r_{jl}^{(t+1)} =\frac{\sum\limits_{i=1}^{c}\sum\limits_{n\in {\rm {\mathcal N}}_{j} }\frac{\left(u_{in}^{(t{\rm +}1)} \right)^{m} \left(x_{jl} -v_{il}^{(t+1)} \right)}{1+d_{nj} }   }{\sum\limits_{i=1}^{c}\sum\limits_{n\in {\rm {\mathcal N}}_{j} }\frac{\left(u_{in}^{(t{\rm +}1)} \right)^{m} }{1+d_{nj} }   +\sum\limits_{n\in {\rm {\mathcal N}}_{j} }\frac{2\beta_{l} \left(w_{jl}^{(t)} \right)^{2} }{1+d_{nj}}}
\end{equation}

To show the impact of weighted $\ell_{2} $-norm fidelity on FCM, we show an example, as shown in Fig. \ref{fig:fidelity}. We impose a mixture of Poisson, Gaussian, and impulse noise ($\sigma =30,$ $p=20\% $) on a noise-free image in Fig. \ref{fig:fidelity}(a). We set $c$ to 4. The settings of $\xi $ and ${\bm \beta}$ are discussed in the later section.
\begin{figure}[htb]
\centering
\begin{minipage}[t]{0.24\linewidth}
\centering
\includegraphics[width=1\textwidth]{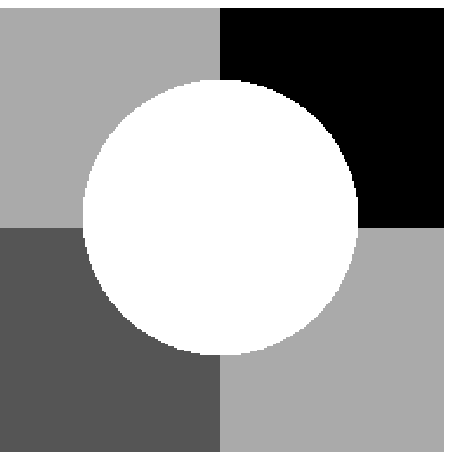}
\centerline{\footnotesize (a)}
\end{minipage}
\begin{minipage}[t]{0.24\linewidth}
\centering
\includegraphics[width=1\textwidth]{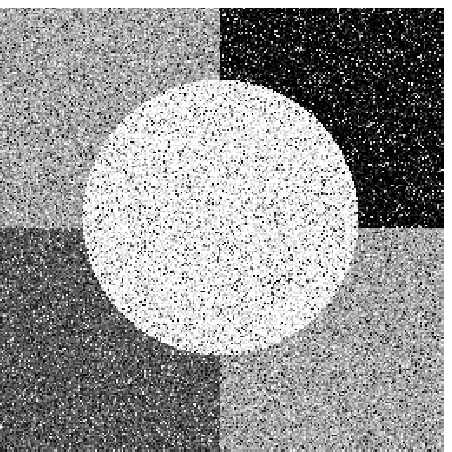}
\centerline{\footnotesize (b)}
\end{minipage}
\begin{minipage}[t]{0.24\linewidth}
\centering
\includegraphics[width=1\textwidth]{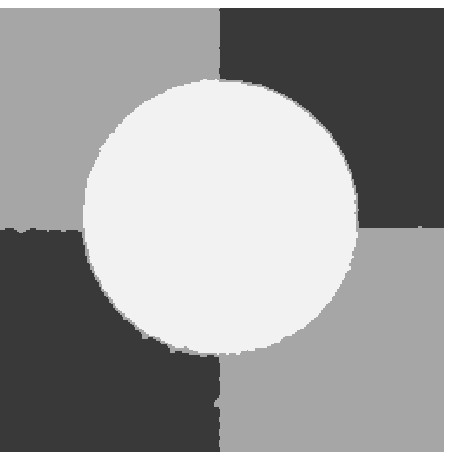}
\centerline{ \footnotesize (c)}
\end{minipage}
\begin{minipage}[t]{0.24\linewidth}
\centering
\includegraphics[width=1\textwidth]{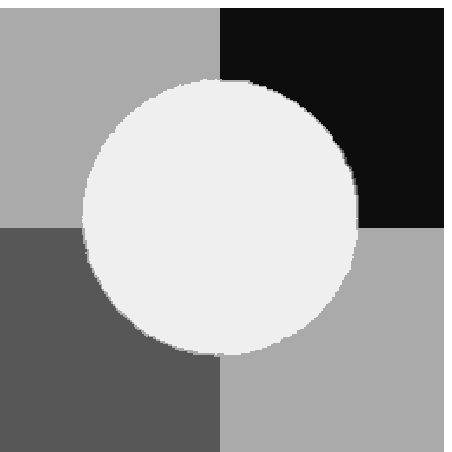}
\centerline{ \footnotesize (d)}
\end{minipage}
\begin{minipage}[t]{0.24\linewidth}
\centering
\includegraphics[width=1\textwidth]{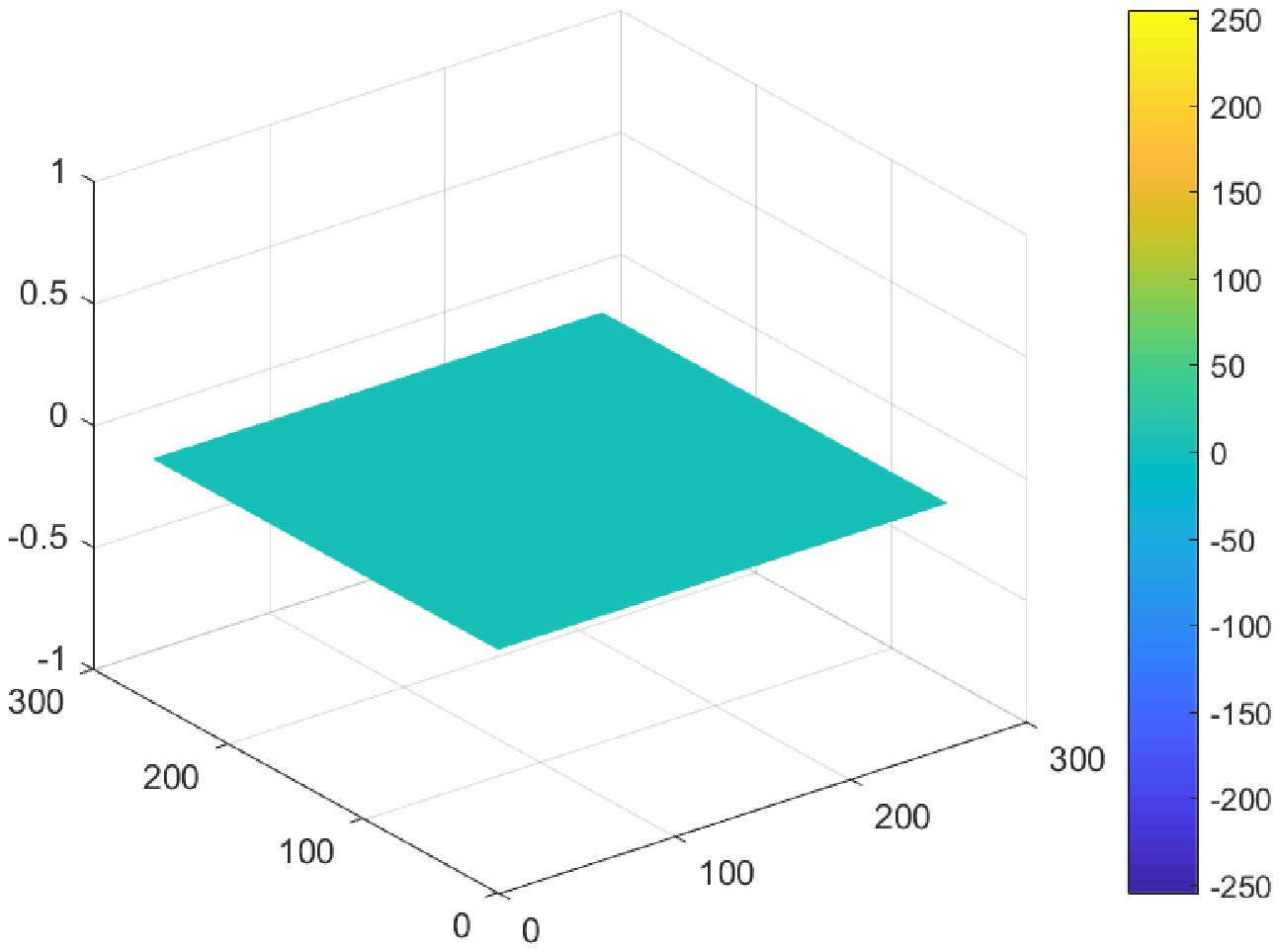}
\centerline{\footnotesize (e)}
\end{minipage}
\begin{minipage}[t]{0.24\linewidth}
\centering
\includegraphics[width=1\textwidth]{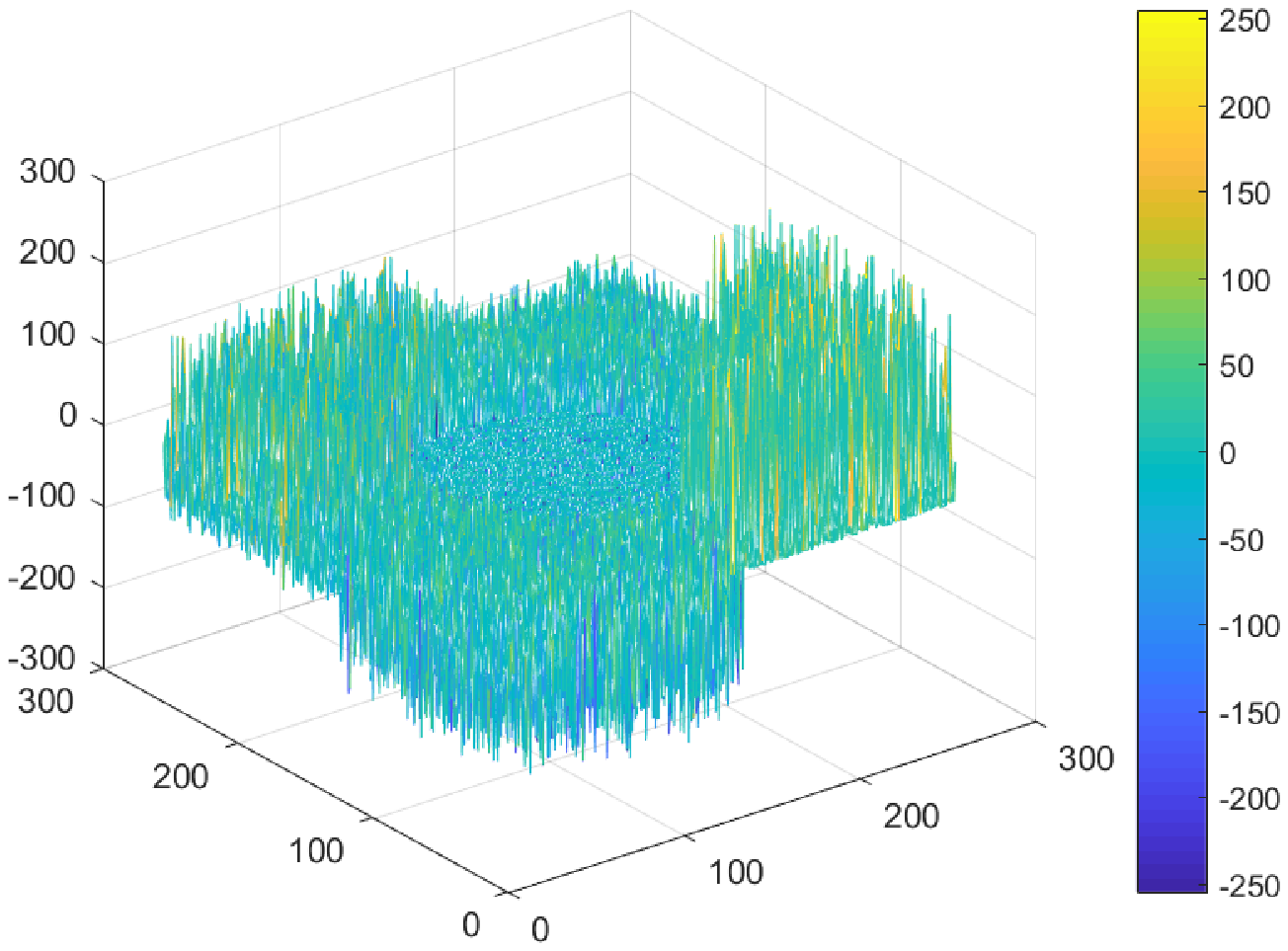}
\centerline{\footnotesize (f)}
\end{minipage}
\begin{minipage}[t]{0.24\linewidth}
\centering
\includegraphics[width=1\textwidth]{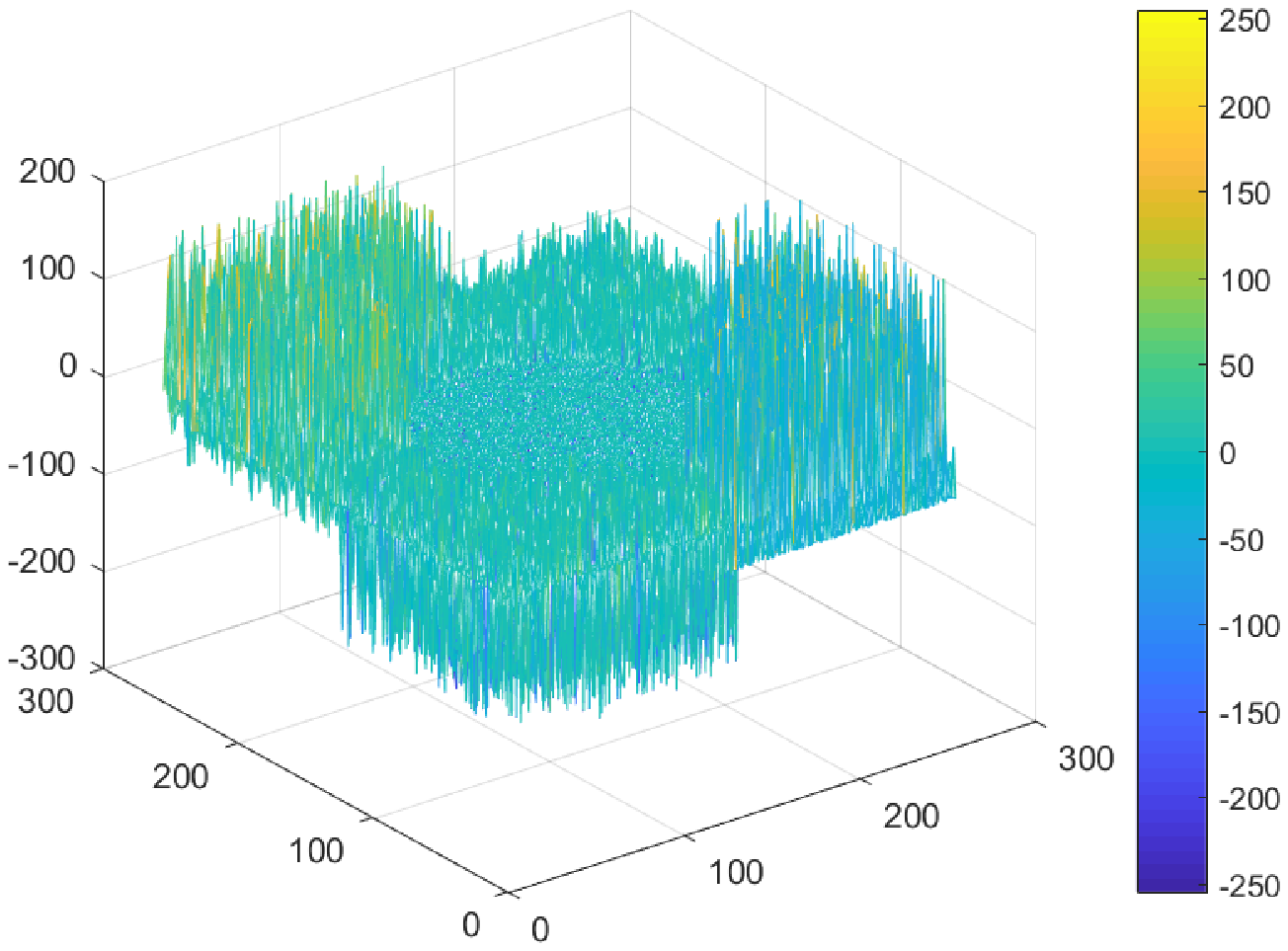}
\centerline{ \footnotesize (g)}
\end{minipage}
\begin{minipage}[t]{0.24\linewidth}
\centering
\includegraphics[width=1\textwidth]{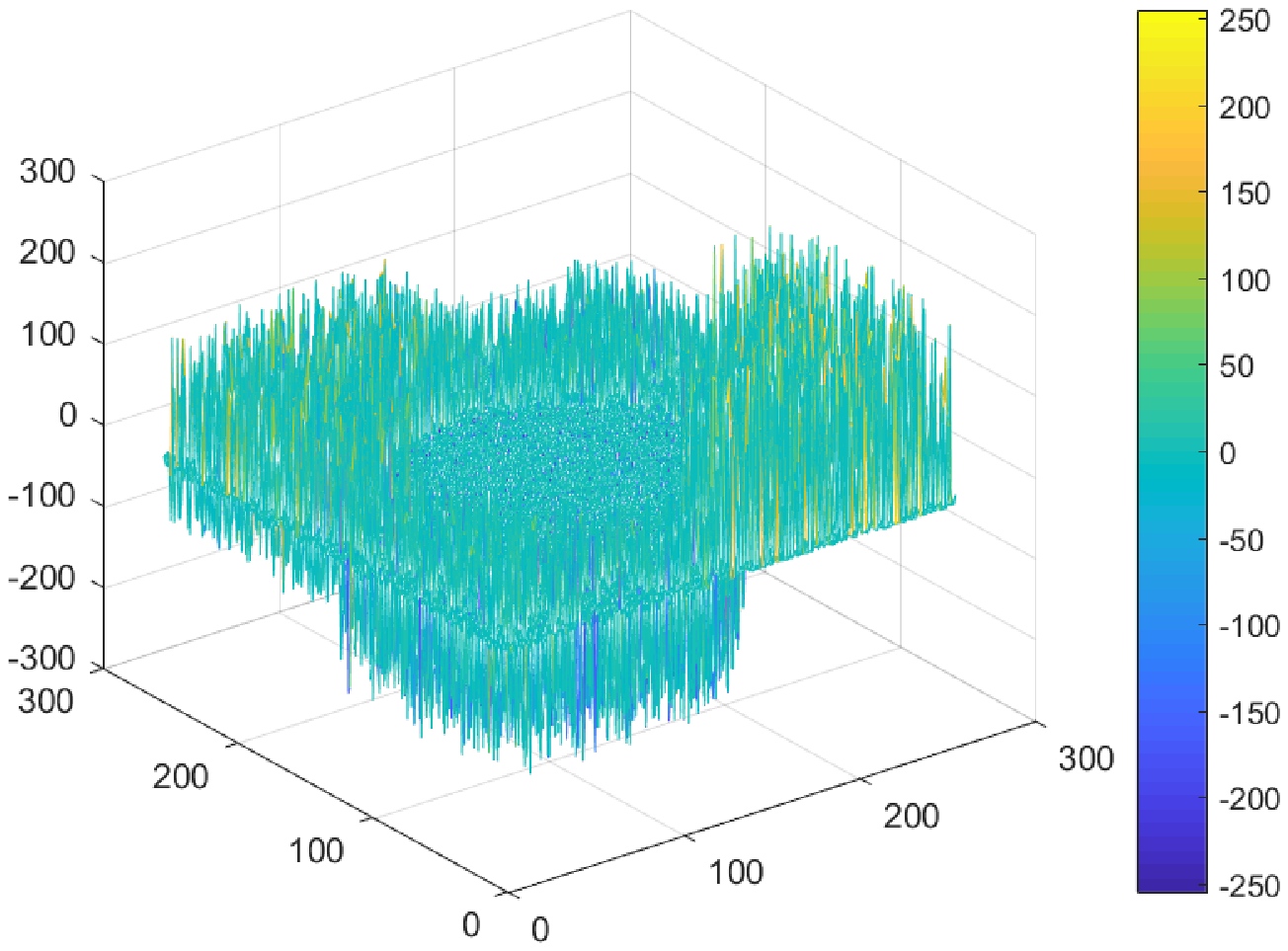}
\centerline{ \footnotesize (h)}
\end{minipage}
\caption{Noise estimation comparison between DSFCM\_N and WRFCM. (a) noise-free image; (b) observed image; (c) segmented image of DSFCM\_N; (d) segmented image of WRFCM; (e) noise in the noise-free image; (f) noise in the observed image; (g) noise estimation of DSFCM\_N; and (h) noise estimation of WRFCM.}
\label{fig:fidelity}
\end{figure}

As shown in Fig. \ref{fig:fidelity}, the noise estimation of DSFCM\_N in Fig. \ref{fig:fidelity}(g) is far from the true one in Fig. \ref{fig:fidelity}(f). However, WRFCM achieves a better noise estimation result as shown in Fig. \ref{fig:fidelity}(h). In addition, it has better performance for \mbox{noise-suppression} and feature-preserving than DSFCM\_N, which can be visually observed from Fig. \ref{fig:fidelity}(c) and (d).

Algorithm \ref{algorithm1} is terminated when $\|{\bm U}^{(t{\rm +}1)} -{\bm U}^{(t)} \|<\varepsilon$. Based on optimal $\bm U$ and $\bm V$, a segmented image $\widehat{{\bm X}}$ is obtained. WRFCM for minimizing \eqref{GrindEQ__9_} is realized in Algorithm \ref{algorithm2}.
\begin{algorithm}[htb]
\caption{Residual-driven FCM with weighted $\ell_{2}$-norm fidelity (WRFCM)}
\begin{algorithmic}[1]
\REQUIRE Observed image $\bm{X}$, fuzzification exponent $m$, number of clusters $c$, and threshold $\varepsilon$.
\ENSURE Segmented image $\widehat{{\bm X}}$.
\STATE Initialize ${\bm W}^{(0)} $ as a matrix of ones and generate randomly prototypes ${\bm V}^{(0)} $
\STATE $t\leftarrow 0$
\REPEAT
\STATE Calculate the partition matrix ${\bm U}^{(t+1)} $ via \eqref{GrindEQ__13_}
\STATE Calculate the prototypes ${\bm V}^{(t+1)} $ via \eqref{GrindEQ__14_}
\STATE Calculate the residual ${\bm R}^{(t{\rm +}1)} $ via \eqref{GrindEQ__18_}
\STATE Update the weight matrix ${\bm W}^{(t{\rm +}1)} $ via \eqref{GrindEQ__7_}
\STATE $t\leftarrow t+1$
\UNTIL {$ \|\bm U^{(t+1)}-\bm U^{(t)}\|< \varepsilon $}
\STATE \textbf{return} $\bm U$, $\bm V$, $\bm R$ and $\bm W$\\
\STATE Generate the segmented image $\widehat{{\bm X}}$ based on $\bm U$ and $\bm V$
\end{algorithmic}
\label{algorithm2}
\end{algorithm}

\subsection{Convergence Analysis}
In WRFCM, we set $\|{\bm U}^{(t{\rm +}1)} -{\bm U}^{(t)} \|<\varepsilon $ as the termination condition. In order to analyze the convergence of WRFCM, we take Fig. \ref{fig:fidelity} as a case study. We set $\varepsilon =1\times 10^{-6} $. In Fig. \ref{fig:convergence}, we draw the curves of $\theta=\|{\bm U}^{(t{\rm +}1)} -{\bm U}^{(t)} \|$ and $J$ versus iteration step $t$, respectively.
\begin{figure}[htb]
\setlength{\abovecaptionskip}{0pt}
\setlength{\belowcaptionskip}{0pt}
\centering
\begin{minipage}[t]{1\linewidth}
\centering
\includegraphics[width=1\textwidth]{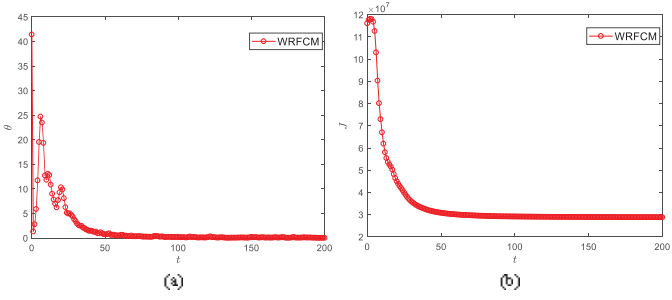}
%\centerline{(a)}
\end{minipage}
\caption{Convergence of WRFCM. (a) $\theta $ and (b) $J$ versus $t$.}
\label{fig:convergence}
\end{figure}

As Fig. \ref{fig:convergence} indicates, since the prototypes are randomly initialized, the convergence of WRFCM oscillates slightly at the beginning. Nevertheless, it reaches stability after a few iterations. In addition, even though $\theta$ exhibits an oscillating process, $J$ keeps decreasing until the iteration stops. To sum up, WRFCM has outstanding convergence since the weight $\ell_{2}$-norm fidelity makes mixed noise distribution estimated accurately so that the residual is gradually separated from observed data as iterations proceed.

\section{Experimental Studies}\label{sec:experiments}
In this section, to show the performance, efficiency and robustness of \mbox{WRFCM}, we provide numerical experiments on synthesis, medical, and other real-world images. To highlight the superiority and improvement of WRFCM over conventional FCM, we also compare it with seven FCM variants, i.e., FCM\_S1 \cite{Chen2004}, FCM\_S2 \cite{Chen2004}, FLICM \cite{Krinidis2010}, \mbox{KWFLICM} \cite{Gong2013}, FRFCM \cite{Lei2018}, WFCM \cite{Wang2019}, and DSFCM\_N \cite{Zhang2019}. They are the most representative ones in the field.

\subsection{Evaluation Indicators}
To quantitatively evaluate the performance of WRFCM, we adopt three objective evaluation indicators, i.e., segmentation accuracy (SA) \cite{Li2011}, Matthews correlation coefficient (MCC) \cite{Thanh2019Blood}, and Sorensen-Dice similarity (SDS) \cite{Thanh2019,Taha2015}. Note that a single one cannot fully reflect true segmentation results. SA is defined as:
$${\rm SA}=\sum _{i=1}^{c}|S_{i} \cap G_{i} | /K$$
where $S_{i} $ and $G_{i} $ are the $i$-th cluster in a segmented image and its ground truth, respectively. $|\cdot |$ denotes the cardinality of a set. MCC is computed as:
\begin{small}
\begin{equation*}
{\rm MCC}=\frac{T_P\cdot T_N-F_P\cdot F_N}{\sqrt{\left(T_P+F_P\right)\cdot \left(T_P+F_N\right)\cdot \left(T_N+F_P\right)\cdot \left(T_P+F_N\right)} }
\end{equation*}
\end{small}where $T_P$, $F_P$, $T_N$, and $F_N$ are the number of true positive, false positive, true negative, and false negative, respectively. SDS is formulated as:
$${\rm SDS}=\frac{2T_P}{2T_P+F_P+F_N}$$

\subsection{Dataset Descriptions}
Tested images except for synthetic ones come from four publicly available databases including a medical one and three \mbox{real-world} ones. The details are outlined as follows:
\begin{itemize}
\item[1)]  BrianWeb\footnote{\url{http://www.bic.mni.mcgill.ca/brainweb/}}: This is an online interface to a 3D MRI simulated brain database. The parameter settings are fixed to 3 modalities, 5 slice thicknesses, 6 levels of noise, and 3 levels of intensity non-uniformity. BrianWeb provides golden standard segmentation.

\item[2)]  Berkeley Segmentation Data Set (BSDS)\footnote{\url{https://www2.eecs.berkeley.edu/Research/Projects/CS/vision/grouping/resources.html}}  \cite{Arbelaez2011}: This database contains 200 training, 100 validation and 200 testing images. Golden standard segmentation is annotated by different subjects for each image of size $321\times481$ or $481\times321$.

\item[3)]  Microsoft Research Cambridge Object Recognition Image Database (MSRC)\footnote{\url{http://research.microsoft.com/vision/cambridge/recognition/}}: This database contains 591 images and 23 object classes. Golden standard segmentation is provided.

\item[4)]  NASA Earth Observation Database (NEO)\footnote{\url{http://neo.sci.gsfc.nasa.gov/}}: This database continually provides information collected by NASA satellites about Earth's ocean, atmosphere, and land surfaces. Due to bit errors appearing in satellite measurements, sampled images of size $1440 \times 720$ contain unknown noise. Therefore, their ground truth is unknown.
\end{itemize}

\subsection{Parameter Settings}
Prior to numerical simulations, we report the parameter settings of WRFCM and seven comparative algorithms. Since spatial information is used in all algorithm, a local window of size $3\times 3$ is selected for all. We set $m=2$ and $\varepsilon =1\times 10^{-6} $ across all algorithms. The setting of $c$ is presented in each experiment.

Except $m$, $\varepsilon$ and $c$, FLICM and KWFLICM are free of any other parameters. However, the remaining algorithms involve different parameters. In FCM\_S1 and FCM\_S2, $\alpha $ is set to 3.8, which controls the impact of spatial information on FCM  by following \cite{Chen2004}. In FRFCM, an observed image is taken as a mask image. A marker image is produced by a $3\times 3$ structuring element. WFCM requires one parameter $\mu \in [0.55,0.65]$ only, which constrains the neighbor term. For DSFCM\_N, $\lambda $ is set based on the standard deviation of each channel of image data.

As to WRFCM, it requires two parameters, i.e., $\xi $ in \eqref{GrindEQ__7_} and ${\bm \beta }$ in \eqref{GrindEQ__9_}. By analyzing mixed noise distributions, $\xi $ is experimentally set to 0.0008. Since the standard deviation of image data is related to noise levels to some extent \cite{Zhang2019}, we can set ${\bm \beta }$ in virtue of the standard deviation of each channel. Based on massive experiments, ${\bm \beta }=\{\beta _{l} :l=1,2,\cdots ,L\}$ is recommended to be chosen as follows:
$$
\beta_{l} =\frac{\phi \cdot \delta _{l} }{100} \,\, \textrm{for} \,\, \phi \in [5,10]
$$
where $\delta_{l} $ is the standard deviation of the $l$-th channel of $\bm{X}$. In fact, ${\bm \beta }$ is equivalently replaced by $\phi $. Here, we give an example to show the setting of $\phi $, refer to Fig. \ref{fig:parameter}. We impose a mixture of Poisson, Gaussian, and impulse noise on the first three synthetic images in the second row of \mbox{Fig. \ref{fig:syn}} respectively. The noise level is $\sigma =30$ and $p=20\%$.

\begin{figure}[htb]
\setlength{\abovecaptionskip}{0pt}
\setlength{\belowcaptionskip}{0pt}
\centering
\begin{minipage}[t]{1\linewidth}
\centering
\includegraphics[width=1\textwidth]{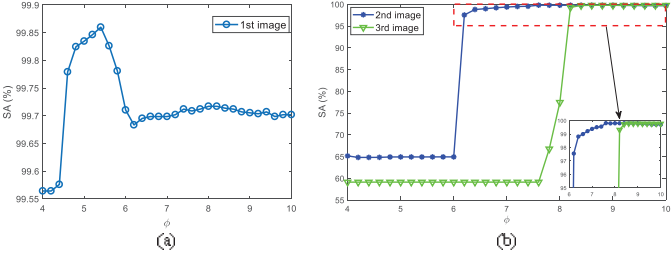}
%\centerline{(a)}
\end{minipage}
\caption{SA values versus $\phi$.}
\label{fig:parameter}
\end{figure}

As Fig. \ref{fig:parameter}(a) indicates, when coping with the first image, the SA value reaches its maximum gradually as the value of $\phi $ increases. Afterwards, it decreases rapidly and tends to be stable. As shown in Fig. \ref{fig:parameter}(b), for the other two images, after the SA value reaches its maximum, it has no apparent changes, implying that the segmentation performance is rather stable. In conclusion, for image segmentation, \mbox{WRFCM} can produce better and better performance as parameter $\phi$ increases from a small value.

\subsection{Experimental Results and Analysis}
\subsubsection{Results on Synthetic Images}
In the first experiment, we representatively choose five synthetic images of size $256\times 256$, as shown in the second row of \mbox{Fig. \ref{fig:syn}}. A mixture of Poisson, Gaussian, and impulse noise is considered for all cases. To be specific, Poisson noise is first added. Then we add Gaussian noise with $\sigma =30$. Finally, the random-valued impulse noise with $p=20\% $ is added since it is more difficult to detect than salt and pepper impulse noise. For five images, we set $c$ to 4, 4, 4, 3, and 3, respectively. The segmentation results are given in Fig. \ref{fig:syn} and Table \ref{tab:syn}. The best values are in bold.
\begin{figure}[htb]
\centering
\begin{minipage}[t]{0.18\linewidth}
\centering
\includegraphics[width=1\textwidth]{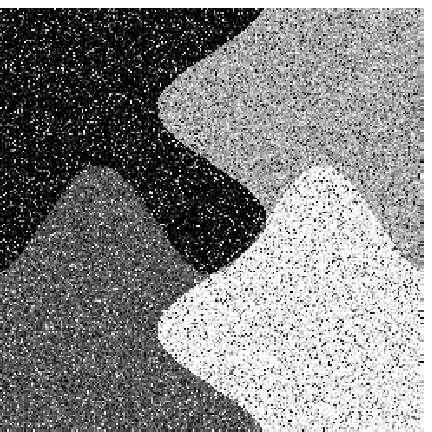}
%\centerline{(a)}
\end{minipage}
\begin{minipage}[t]{0.18\linewidth}
\centering
\includegraphics[width=1\textwidth]{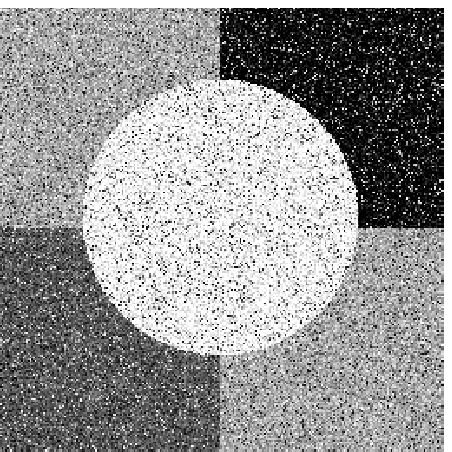}
%\centerline{(b)}
\end{minipage}
\begin{minipage}[t]{0.18\linewidth}
\centering
\includegraphics[width=1\textwidth]{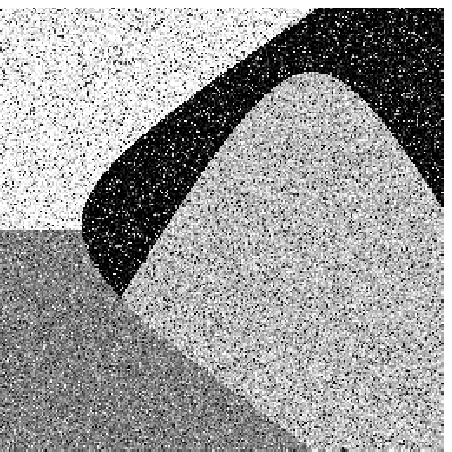}
%\centerline{(c)}
\end{minipage}
\begin{minipage}[t]{0.18\linewidth}
\centering
\includegraphics[width=1\textwidth]{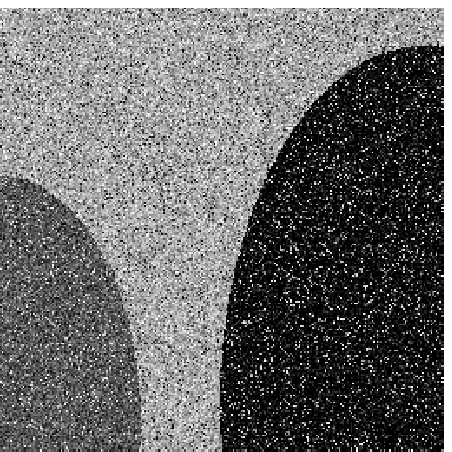}
%\centerline{(d)}
\end{minipage}
\begin{minipage}[t]{0.18\linewidth}
\centering
\includegraphics[width=1\textwidth]{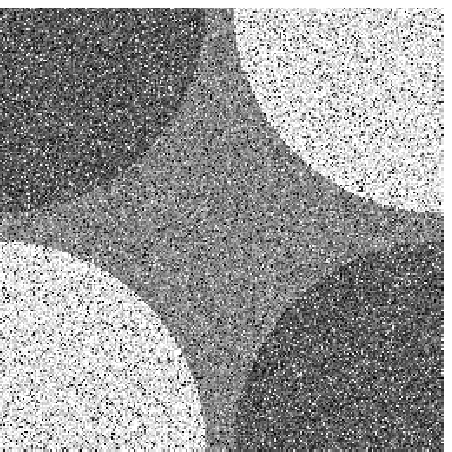}
%\centerline{(e)}
\end{minipage}\\
\begin{minipage}[t]{0.18\linewidth}
\centering
\includegraphics[width=1\textwidth]{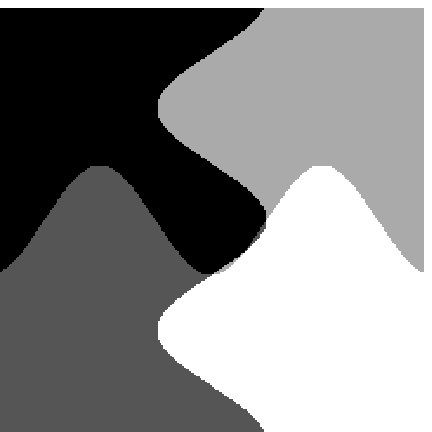}
%\centerline{(a)}
\end{minipage}
\begin{minipage}[t]{0.18\linewidth}
\centering
\includegraphics[width=1\textwidth]{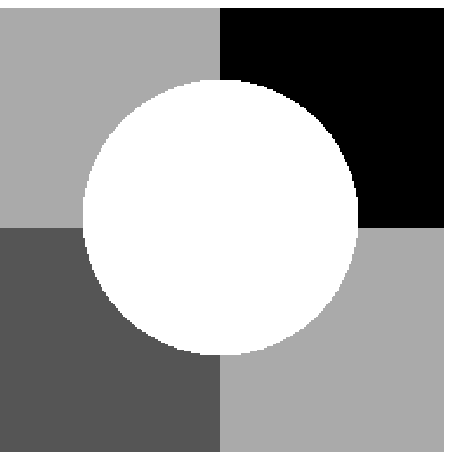}
%\centerline{(b)}
\end{minipage}
\begin{minipage}[t]{0.18\linewidth}
\centering
\includegraphics[width=1\textwidth]{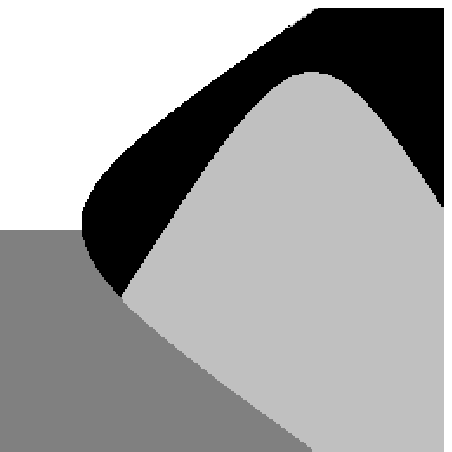}
%\centerline{(c)}
\end{minipage}
\begin{minipage}[t]{0.18\linewidth}
\centering
\includegraphics[width=1\textwidth]{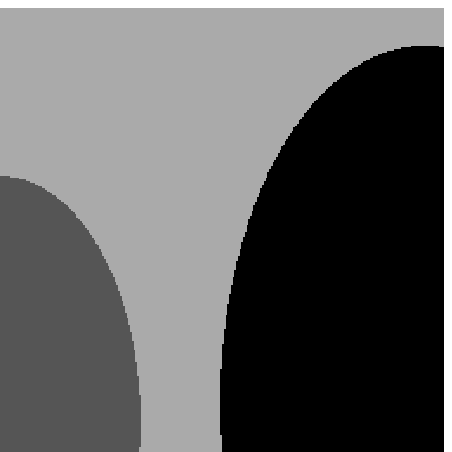}
%\centerline{(d)}
\end{minipage}
\begin{minipage}[t]{0.18\linewidth}
\centering
\includegraphics[width=1\textwidth]{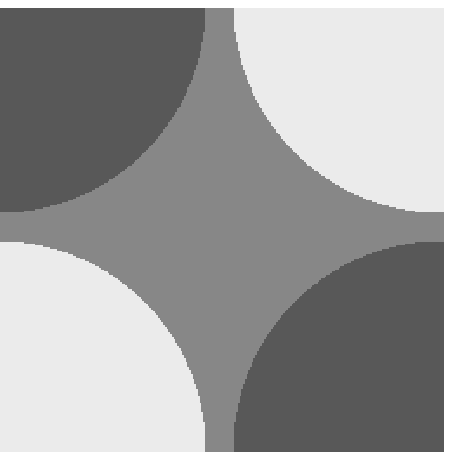}
%\centerline{(e)}
\end{minipage}\\
\begin{minipage}[t]{0.18\linewidth}
\centering
\includegraphics[width=1\textwidth]{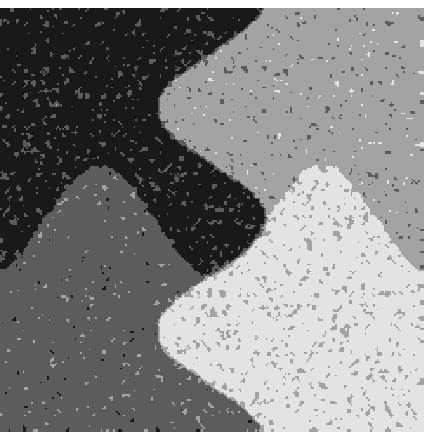}
%\centerline{(a)}
\end{minipage}
\begin{minipage}[t]{0.18\linewidth}
\centering
\includegraphics[width=1\textwidth]{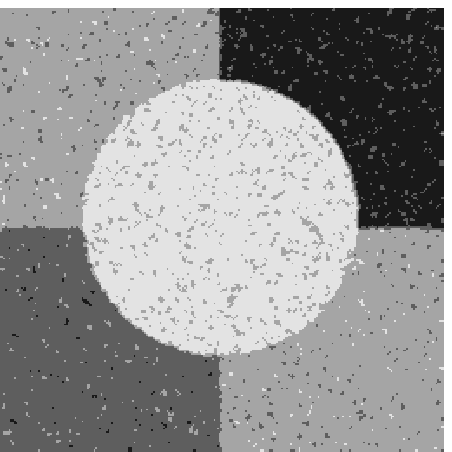}
%\centerline{(b)}
\end{minipage}
\begin{minipage}[t]{0.18\linewidth}
\centering
\includegraphics[width=1\textwidth]{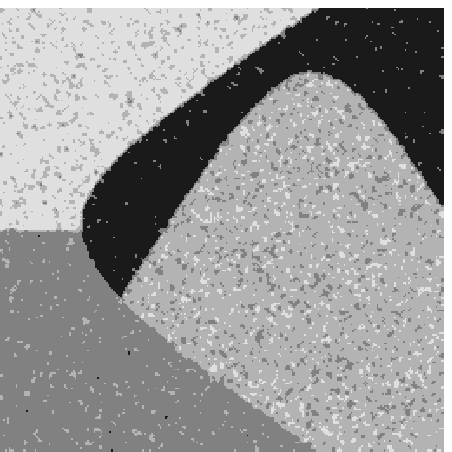}
%\centerline{(c)}
\end{minipage}
\begin{minipage}[t]{0.18\linewidth}
\centering
\includegraphics[width=1\textwidth]{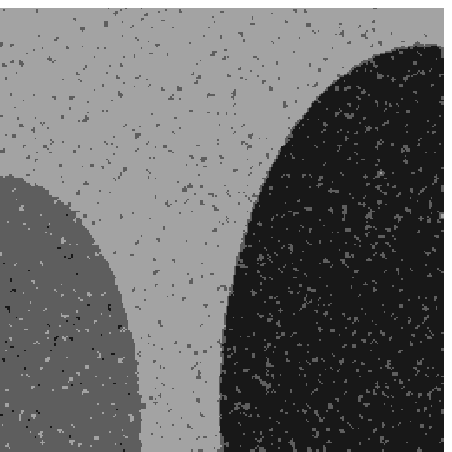}
%\centerline{(d)}
\end{minipage}
\begin{minipage}[t]{0.18\linewidth}
\centering
\includegraphics[width=1\textwidth]{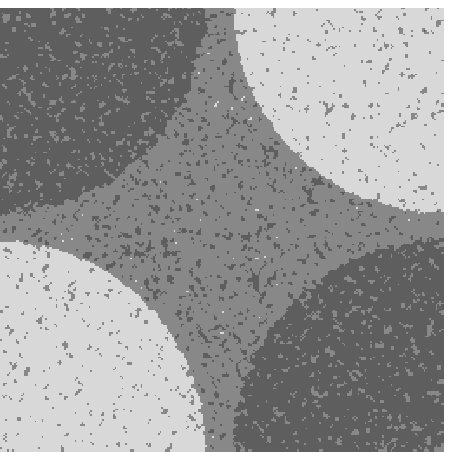}
%\centerline{(e)}
\end{minipage}\\
\begin{minipage}[t]{0.18\linewidth}
\centering
\includegraphics[width=1\textwidth]{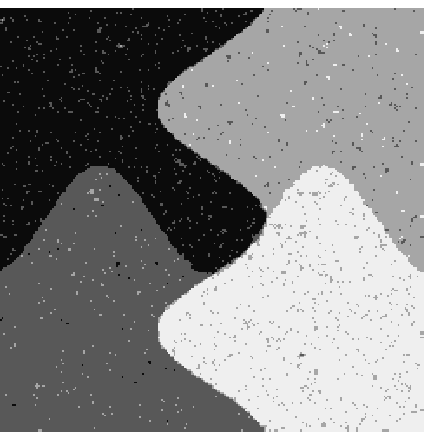}
%\centerline{(a)}
\end{minipage}
\begin{minipage}[t]{0.18\linewidth}
\centering
\includegraphics[width=1\textwidth]{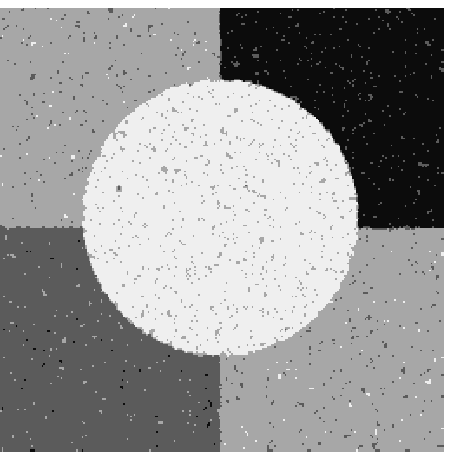}
%\centerline{(b)}
\end{minipage}
\begin{minipage}[t]{0.18\linewidth}
\centering
\includegraphics[width=1\textwidth]{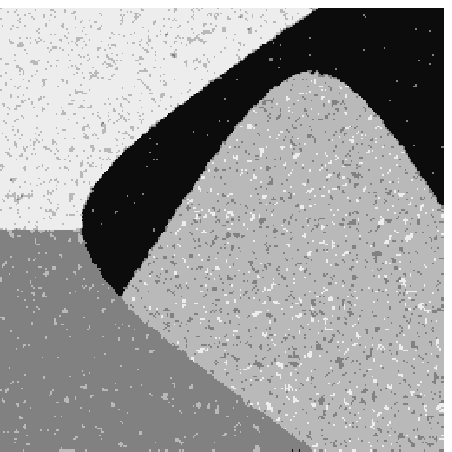}
%\centerline{(c)}
\end{minipage}
\begin{minipage}[t]{0.18\linewidth}
\centering
\includegraphics[width=1\textwidth]{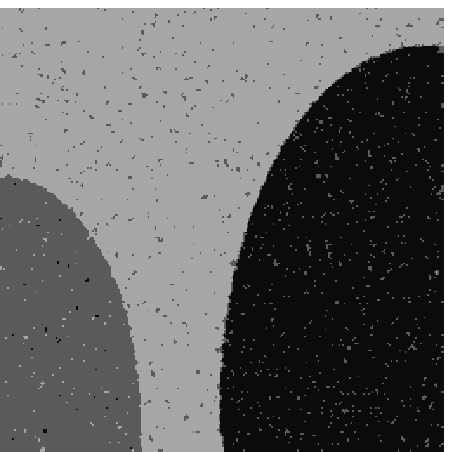}
%\centerline{(d)}
\end{minipage}
\begin{minipage}[t]{0.18\linewidth}
\centering
\includegraphics[width=1\textwidth]{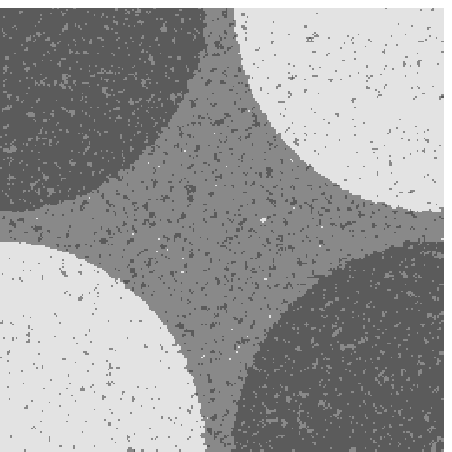}
%\centerline{(e)}
\end{minipage}\\
\begin{minipage}[t]{0.18\linewidth}
\centering
\includegraphics[width=1\textwidth]{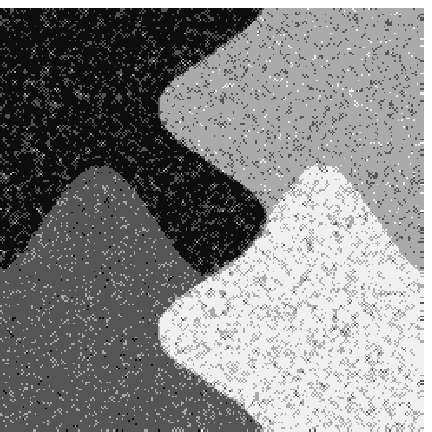}
%\centerline{(a)}
\end{minipage}
\begin{minipage}[t]{0.18\linewidth}
\centering
\includegraphics[width=1\textwidth]{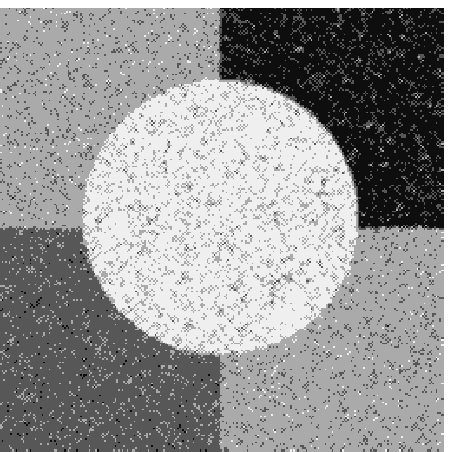}
%\centerline{(b)}
\end{minipage}
\begin{minipage}[t]{0.18\linewidth}
\centering
\includegraphics[width=1\textwidth]{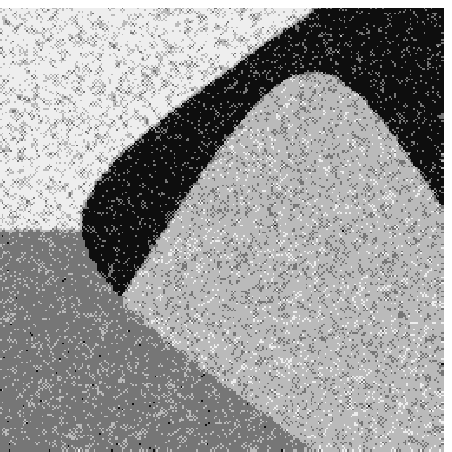}
%\centerline{(c)}
\end{minipage}
\begin{minipage}[t]{0.18\linewidth}
\centering
\includegraphics[width=1\textwidth]{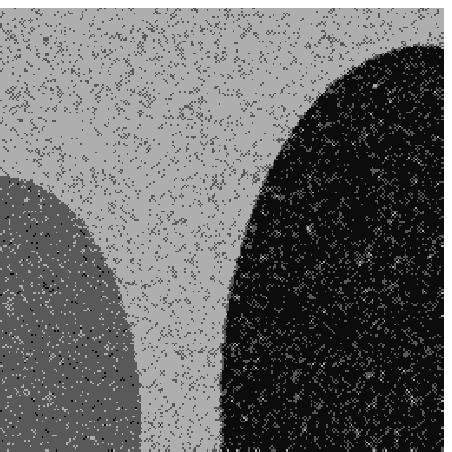}
%\centerline{(d)}
\end{minipage}
\begin{minipage}[t]{0.18\linewidth}
\centering
\includegraphics[width=1\textwidth]{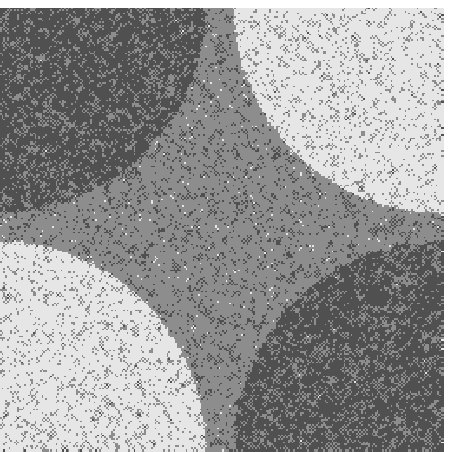}
%\centerline{(e)}
\end{minipage}\\
\begin{minipage}[t]{0.18\linewidth}
\centering
\includegraphics[width=1\textwidth]{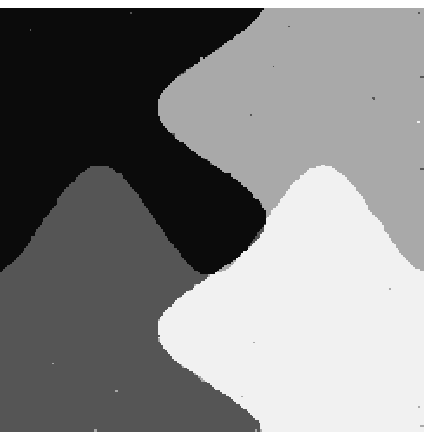}
%\centerline{(a)}
\end{minipage}
\begin{minipage}[t]{0.18\linewidth}
\centering
\includegraphics[width=1\textwidth]{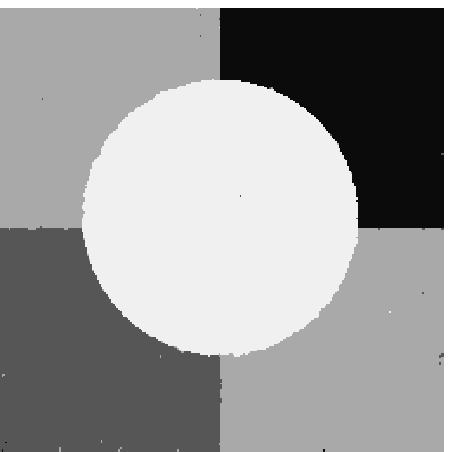}
%\centerline{(b)}
\end{minipage}
\begin{minipage}[t]{0.18\linewidth}
\centering
\includegraphics[width=1\textwidth]{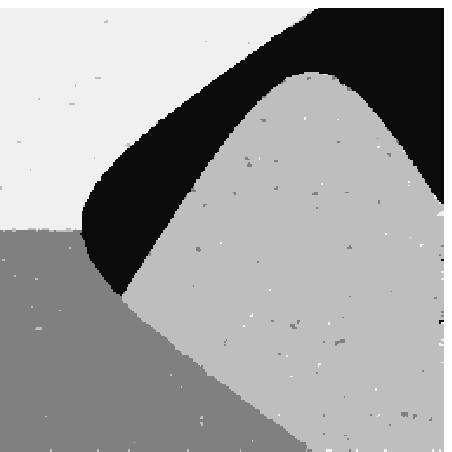}
%\centerline{(c)}
\end{minipage}
\begin{minipage}[t]{0.18\linewidth}
\centering
\includegraphics[width=1\textwidth]{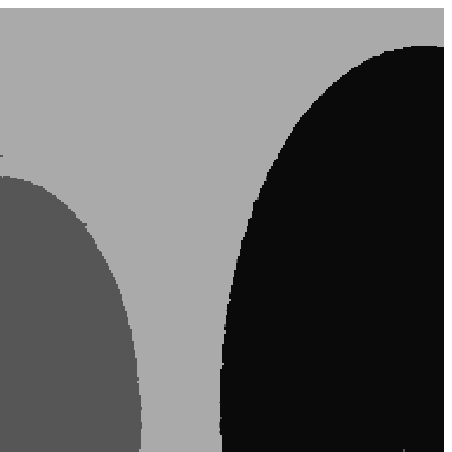}
%\centerline{(d)}
\end{minipage}
\begin{minipage}[t]{0.18\linewidth}
\centering
\includegraphics[width=1\textwidth]{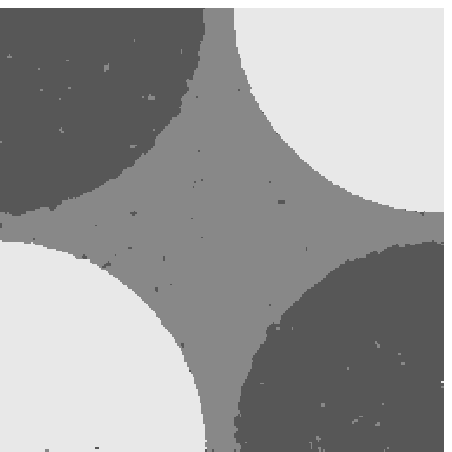}
%\centerline{(e)}
\end{minipage}\\
\begin{minipage}[t]{0.18\linewidth}
\centering
\includegraphics[width=1\textwidth]{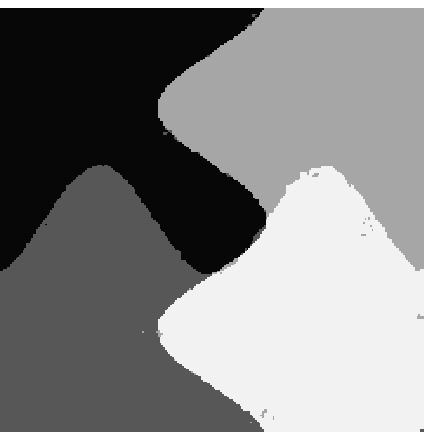}
%\centerline{(a)}
\end{minipage}
\begin{minipage}[t]{0.18\linewidth}
\centering
\includegraphics[width=1\textwidth]{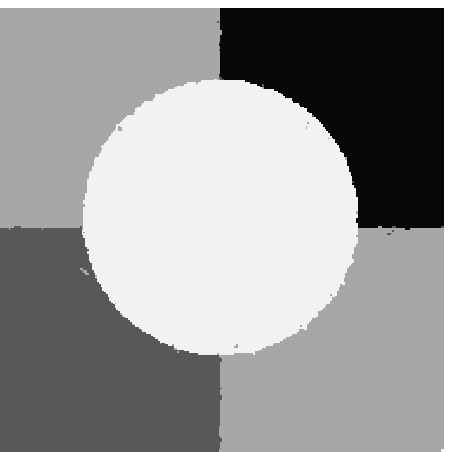}
%\centerline{(b)}
\end{minipage}
\begin{minipage}[t]{0.18\linewidth}
\centering
\includegraphics[width=1\textwidth]{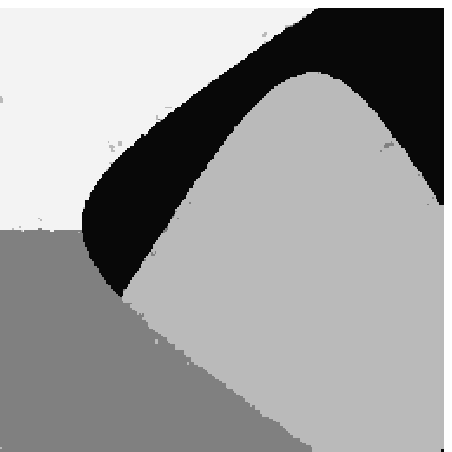}
%\centerline{(c)}
\end{minipage}
\begin{minipage}[t]{0.18\linewidth}
\centering
\includegraphics[width=1\textwidth]{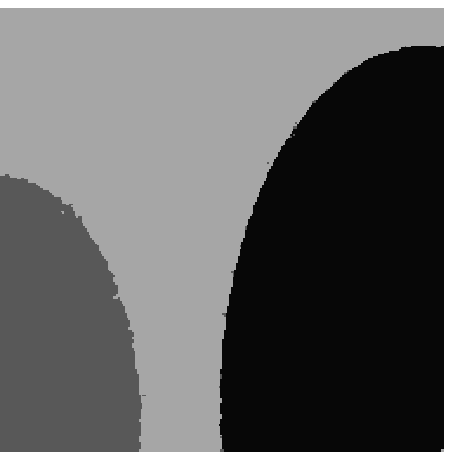}
%\centerline{(d)}
\end{minipage}
\begin{minipage}[t]{0.18\linewidth}
\centering
\includegraphics[width=1\textwidth]{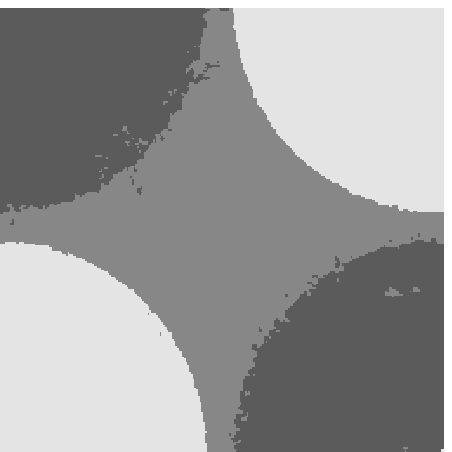}
%\centerline{(e)}
\end{minipage}\\
\begin{minipage}[t]{0.18\linewidth}
\centering
\includegraphics[width=1\textwidth]{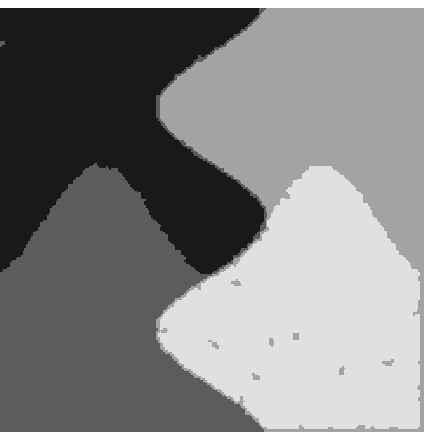}
%\centerline{(a)}
\end{minipage}
\begin{minipage}[t]{0.18\linewidth}
\centering
\includegraphics[width=1\textwidth]{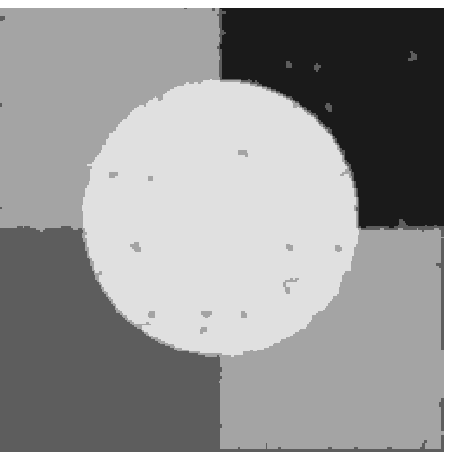}
%\centerline{(b)}
\end{minipage}
\begin{minipage}[t]{0.18\linewidth}
\centering
\includegraphics[width=1\textwidth]{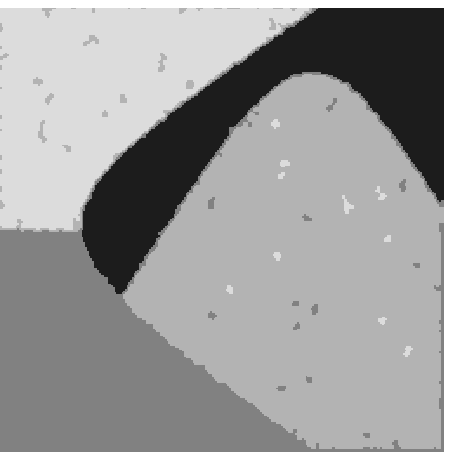}
%\centerline{(c)}
\end{minipage}
\begin{minipage}[t]{0.18\linewidth}
\centering
\includegraphics[width=1\textwidth]{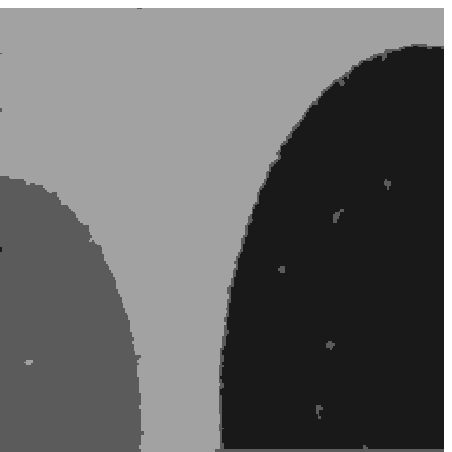}
%\centerline{(d)}
\end{minipage}
\begin{minipage}[t]{0.18\linewidth}
\centering
\includegraphics[width=1\textwidth]{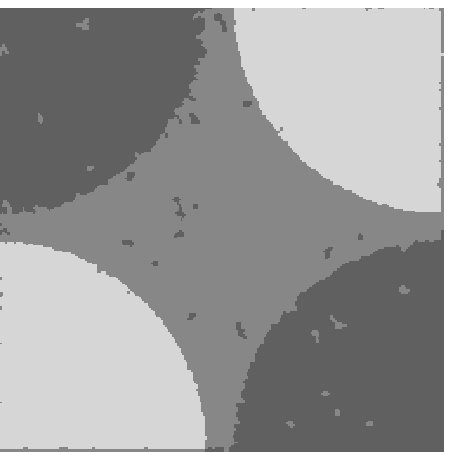}
%\centerline{(e)}
\end{minipage}\\
\begin{minipage}[t]{0.18\linewidth}
\centering
\includegraphics[width=1\textwidth]{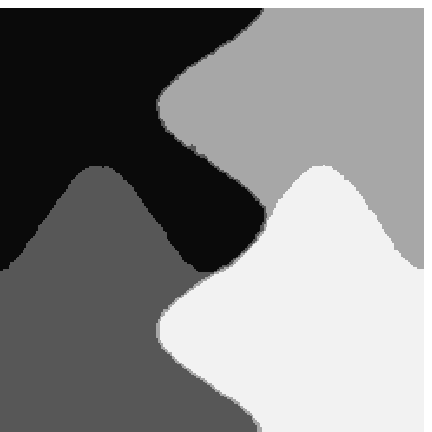}
%\centerline{(a)}
\end{minipage}
\begin{minipage}[t]{0.18\linewidth}
\centering
\includegraphics[width=1\textwidth]{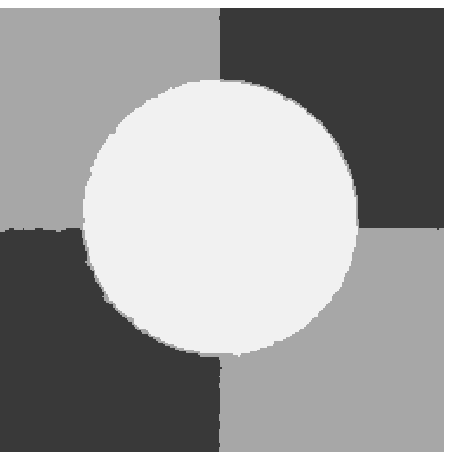}
%\centerline{(b)}
\end{minipage}
\begin{minipage}[t]{0.18\linewidth}
\centering
\includegraphics[width=1\textwidth]{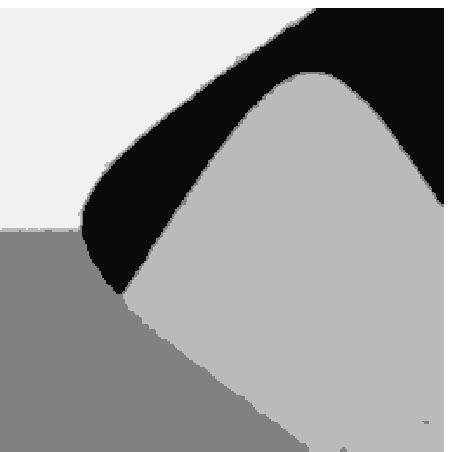}
%\centerline{(c)}
\end{minipage}
\begin{minipage}[t]{0.18\linewidth}
\centering
\includegraphics[width=1\textwidth]{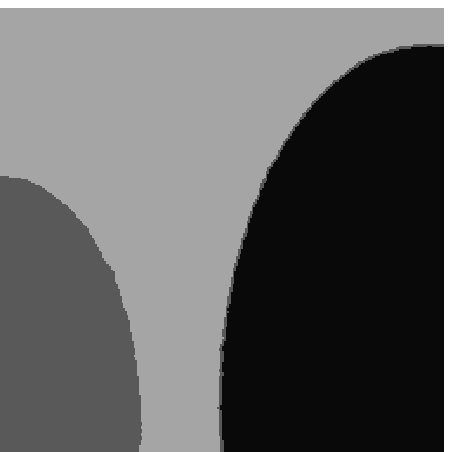}
%\centerline{(d)}
\end{minipage}
\begin{minipage}[t]{0.18\linewidth}
\centering
\includegraphics[width=1\textwidth]{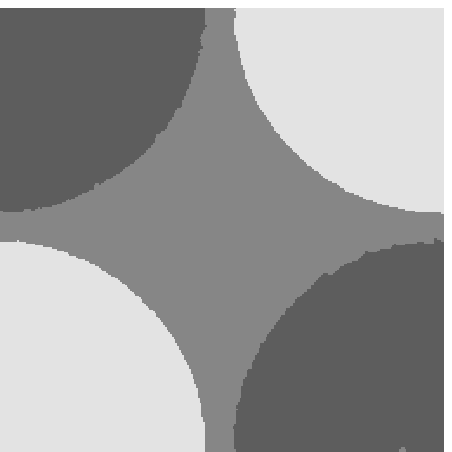}
%\centerline{(e)}
\end{minipage}\\
\begin{minipage}[t]{0.18\linewidth}
\centering
\includegraphics[width=1\textwidth]{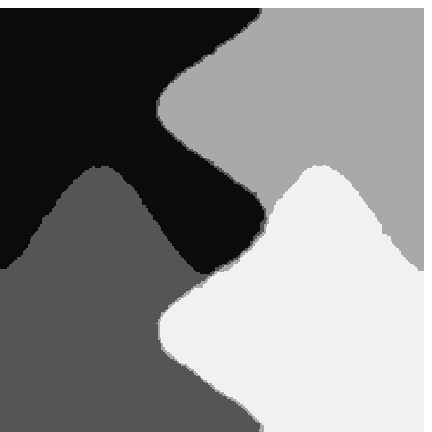}
%\centerline{(a)}
\end{minipage}
\begin{minipage}[t]{0.18\linewidth}
\centering
\includegraphics[width=1\textwidth]{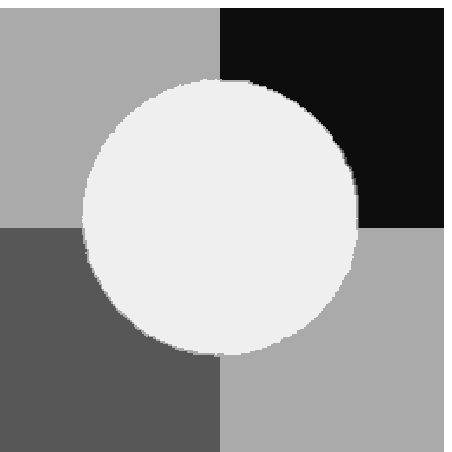}
%\centerline{(b)}
\end{minipage}
\begin{minipage}[t]{0.18\linewidth}
\centering
\includegraphics[width=1\textwidth]{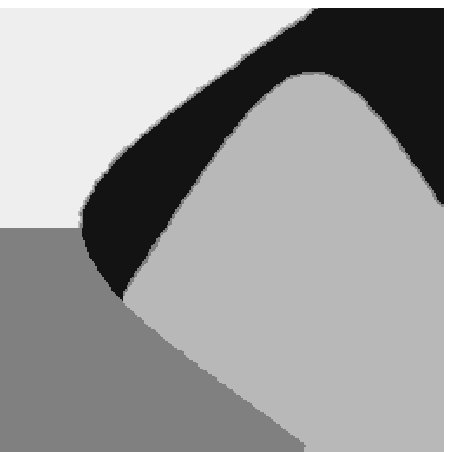}
%\centerline{(c)}
\end{minipage}
\begin{minipage}[t]{0.18\linewidth}
\centering
\includegraphics[width=1\textwidth]{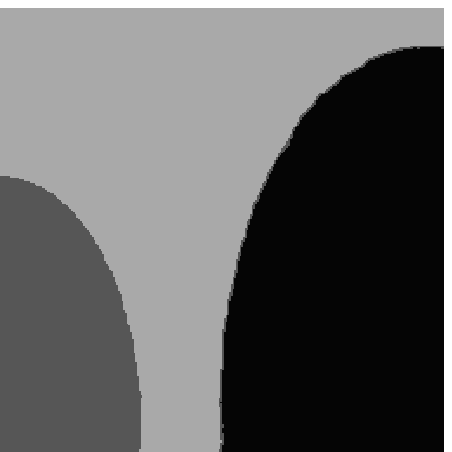}
%\centerline{(d)}
\end{minipage}
\begin{minipage}[t]{0.18\linewidth}
\centering
\includegraphics[width=1\textwidth]{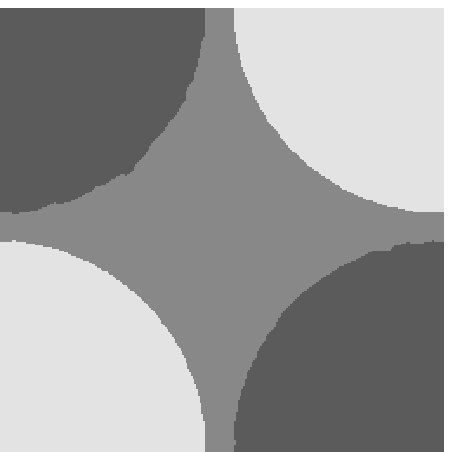}
%\centerline{(e)}
\end{minipage}
\caption{Segmentation results for five synthetic images. The parameters: $\phi_{1} ={\rm 5.58}$, $\phi_{2} ={\rm 7.45}$, $\phi_{3} ={\rm 8.17}$, $\phi_{4} ={\rm 5.79}$, and $\phi_{5} =9.99$. From top to bottom: noisy images, noise-free images, and results of FCM\_S1, FCM\_S2, FLICM, KWFLICM, FRFCM, WFCM, DSFCM\_N, and WRFCM.}
\label{fig:syn}
\end{figure}

As Fig. \ref{fig:syn} indicates, FCM\_S1, FCM\_S2 and FLICM achieve poor results in presence of such a high level of mixed noise. Compared with them, KWFLICM, FRFCM and WFCM suppress the vast majority of mixed noise. Yet they cannot completely remove it. DSFCM\_N visually outperforms other peers mentioned above. However, it generates several topology changes such as merging and splitting. By taking the second synthetic image as a case, we find that DSFCM\_N produces some unclear contours and shadows. Superior to seven peers, WRFCM not only removes all the noise but also preserves more image features.

\begin{table*}[htbp]
\setlength{\abovecaptionskip}{0pt}
\setlength{\belowcaptionskip}{0pt}
%\tabcolsep 0.008\textwidth
  \centering
  \caption{Segmentation performance (\%) on synthetic images}
  \scriptsize
    \begin{tabular}{c|ccc|ccc|ccc|ccc|ccc}
    \toprule
    \multirow{2}[3]{*}{Algorithm} &
      \multicolumn{3}{c|}{Fig. 8 column 1} &
      \multicolumn{3}{c|}{Fig. 8 column 2} &
      \multicolumn{3}{c|}{Fig. 8 column 3} &
      \multicolumn{3}{c|}{Fig. 8 column 4} &
      \multicolumn{3}{c}{Fig. 8 column 5}
      \\
\cmidrule{2-16}     &
      SA &
      SDS &
      MCC &
      SA &
      SDS &
      MCC &
      SA &
      SDS &
      MCC &
      SA &
      SDS &
      MCC &
      SA &
      SDS &
      MCC
      \\
      \midrule
    FCM\_S1 &
      92.902 &
      98.187 &
      96.362 &
      92.625 &
      98.414 &
      95.528 &
      87.289 &
      99.582 &
      97.606 &
      94.453 &
      97.405 &
      95.254 &
      90.178 &
      97.128 &
      95.740
      \\
    FCM\_S2 &
      96.157 &
      98.999 &
      97.991 &
      96.292 &
      99.127 &
      97.520 &
      92.345 &
      99.791 &
      98.808 &
      97.214 &
      84.356 &
      70.353 &
      92.737 &
      98.518 &
      97.769
      \\
    FLICM &
      85.081 &
      90.145 &
      95.082 &
      85.667 &
      95.894 &
      88.576 &
      81.502 &
      83.077 &
      54.764 &
      88.031 &
      92.855 &
      87.353 &
      82.401 &
      91.770 &
      88.350
      \\
    KWFLICM &
      99.706 &
      99.858 &
      99.715 &
      99.730 &
      99.904 &
      99.725 &
      99.310 &
      \textbf{99.938} &
      \textbf{99.648} &
      99.878 &
      99.880 &
      99.776 &
      99.240 &
      99.852 &
      99.774
      \\
    FRFCM &
      99.652 &
      99.920 &
      99.839 &
      99.675 &
      99.895 &
      99.698 &
      99.629 &
      99.924 &
      99.568 &
      99.751 &
      85.222 &
      72.048 &
      98.883 &
      99.726 &
      99.581
      \\
    WFCM &
      97.827 &
      99.325 &
      98.652 &
      98.079 &
      99.363 &
      98.197 &
      96.645 &
      99.735 &
      98.485 &
      98.570 &
      99.106 &
      98.353 &
      97.434 &
      98.515 &
      97.766
      \\
    DSFCM\_N &
      98.954 &
      99.545 &
      99.086 &
      99.226 &
      99.757 &
      99.303 &
      98.503 &
      99.756 &
      98.608 &
      99.205 &
      85.053 &
      71.730 &
      99.655 &
      99.863 &
      99.791
      \\
    WRFCM &
      \textbf{99.859} &
      \textbf{99.937} &
      \textbf{99.843} &
      \textbf{99.802} &
      \textbf{99.958} &
      \textbf{99.792} &
      \textbf{99.785} &
      99.931 &
      99.565 &
      \textbf{99.934} &
      \textbf{99.893} &
      \textbf{99.814} &
      \textbf{99.677} &
      \textbf{99.907} &
      \textbf{99.858}
      \\
    \bottomrule
    \end{tabular}%
  \label{tab:syn}%
\end{table*}%

Table \ref{tab:syn} shows the segmentation results of all algorithms quantitatively. It assembles the values of all three indictors. Clearly, WRFCM achieves better SA results for all images than other peers. In particular, its SA value comes up to 99.934\% for the fourth synthetic image. In most cases, it also gets better SDS and MCC results than its seven peers. For the third synthetic image, WRFCM is only slightly inferior to KWFLICM. Among its seven peers, KWFLICM obtains generally better results. In the light of Fig. \ref{fig:syn} and Table \ref{tab:syn}, we conclude that WRFCM performs better than its peers.
\subsubsection{Results on Medical Images}
Next, we representatively segment five medical images from BrianWeb. They are represented as five slices in the axial plane with a sequence of 70, 80, 90, 100 and 110, which are generated by T1 modality with slice thickness of 1mm resolution, 9\% noise and 20\% intensity non-uniformity. Here, we set $c=4$ for all cases. The comparison between WRFCM and its peers are shown in Fig. \ref{fig:med} and Table \ref{tab:med}. The best values are in bold.
\begin{figure}[htb]
\centering
\begin{minipage}[t]{0.187\linewidth}
\centering
\includegraphics[width=1\textwidth]{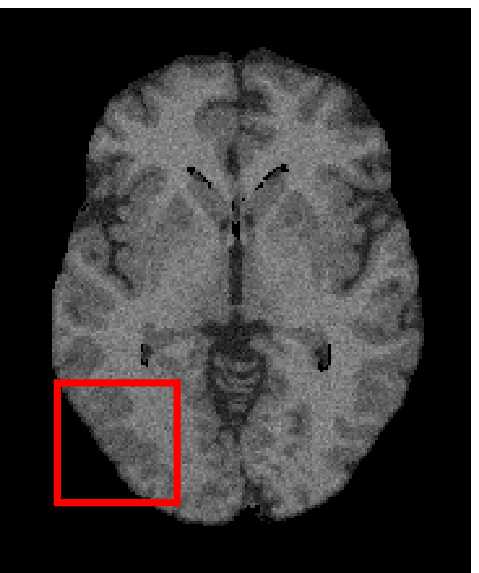}
%\centerline{(a)}
\end{minipage}
\begin{minipage}[t]{0.187\linewidth}
\centering
\includegraphics[width=1\textwidth]{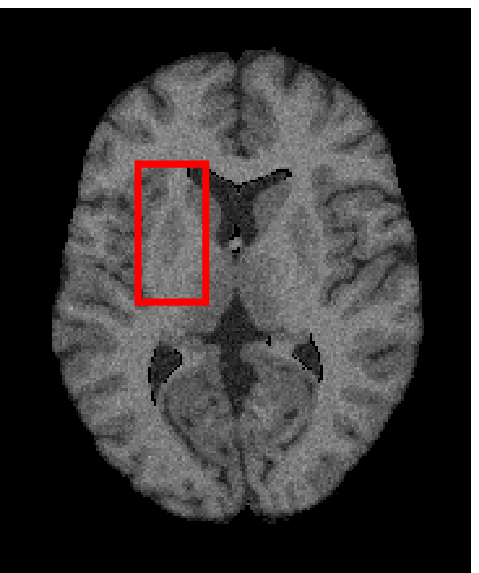}
%\centerline{(b)}
\end{minipage}
\begin{minipage}[t]{0.187\linewidth}
\centering
\includegraphics[width=1\textwidth]{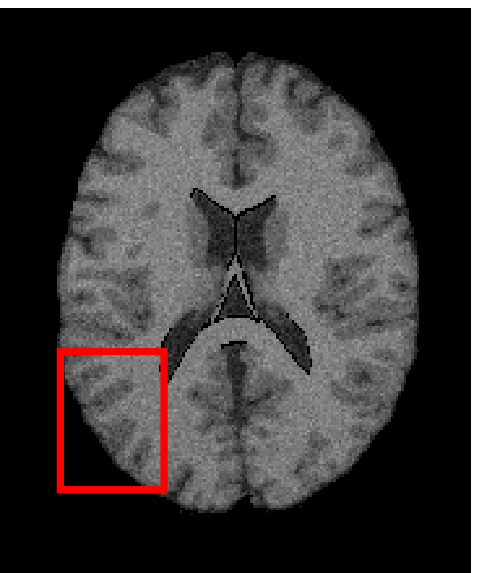}
%\centerline{(c)}
\end{minipage}
\begin{minipage}[t]{0.187\linewidth}
\centering
\includegraphics[width=1\textwidth]{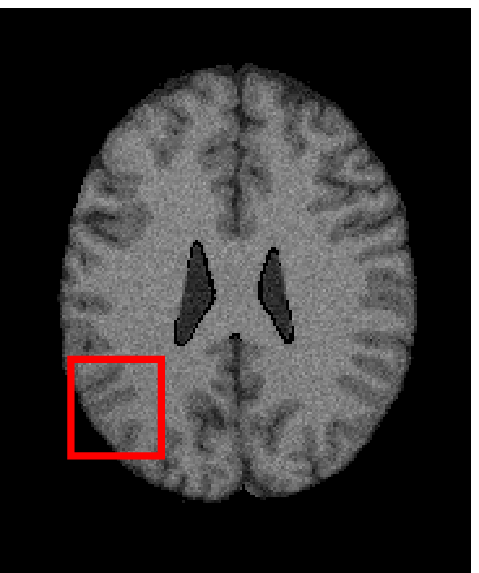}
%\centerline{(d)}
\end{minipage}
\begin{minipage}[t]{0.187\linewidth}
\centering
\includegraphics[width=1\textwidth]{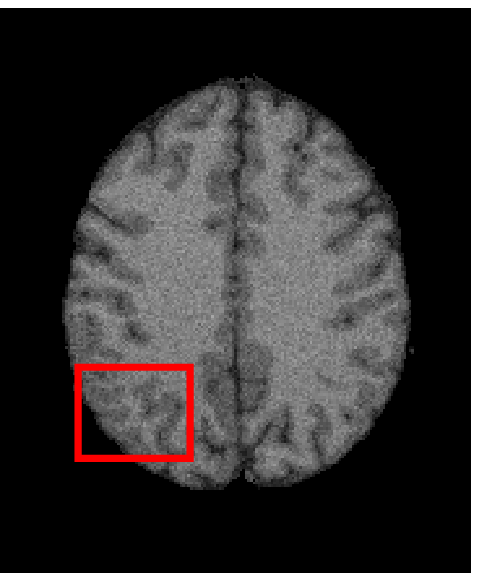}
%\centerline{(e)}
\end{minipage}\\
\begin{minipage}[t]{0.187\linewidth}
\centering
\includegraphics[width=1\textwidth]{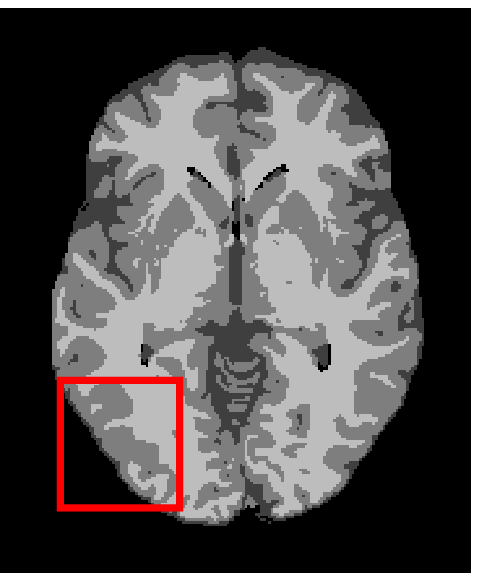}
%\centerline{(a)}
\end{minipage}
\begin{minipage}[t]{0.187\linewidth}
\centering
\includegraphics[width=1\textwidth]{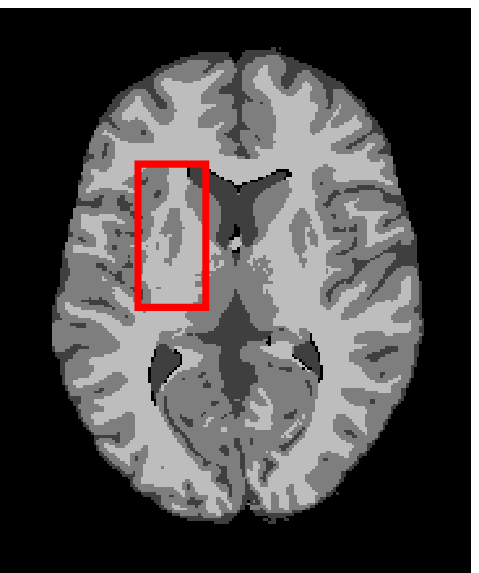}
%\centerline{(b)}
\end{minipage}
\begin{minipage}[t]{0.187\linewidth}
\centering
\includegraphics[width=1\textwidth]{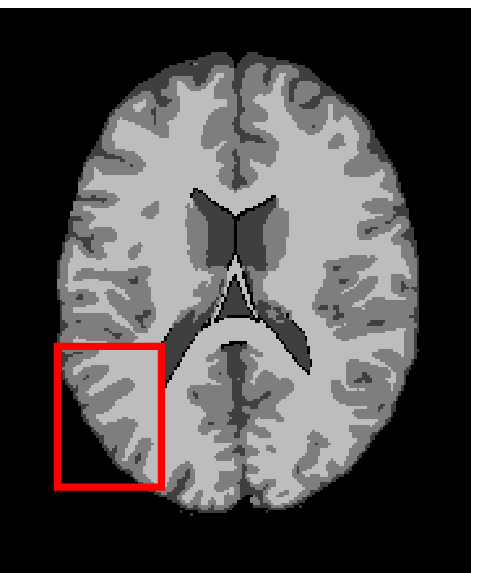}
%\centerline{(c)}
\end{minipage}
\begin{minipage}[t]{0.187\linewidth}
\centering
\includegraphics[width=1\textwidth]{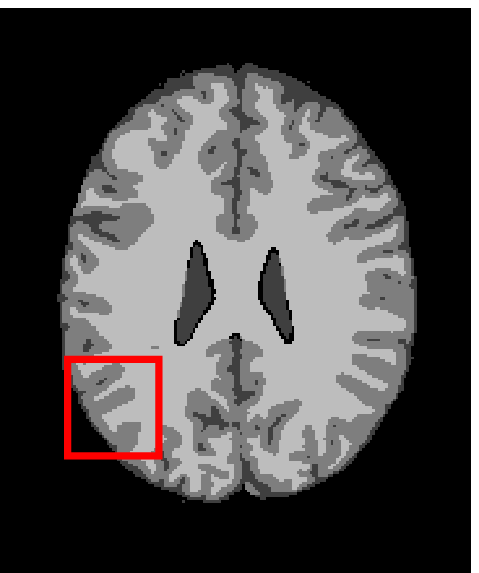}
%\centerline{(d)}
\end{minipage}
\begin{minipage}[t]{0.187\linewidth}
\centering
\includegraphics[width=1\textwidth]{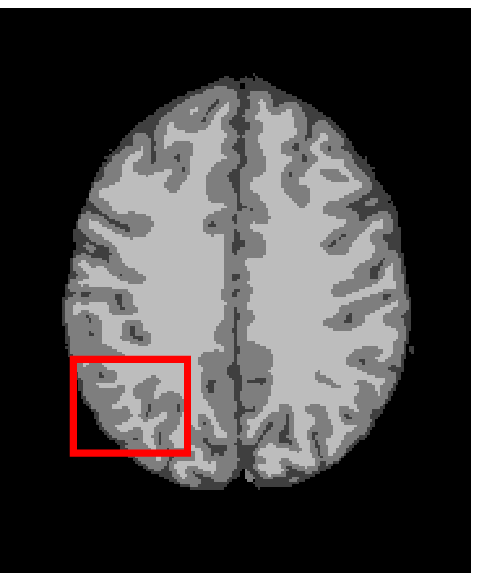}
%\centerline{(e)}
\end{minipage}\\
\begin{minipage}[t]{0.187\linewidth}
\centering
\includegraphics[width=1\textwidth]{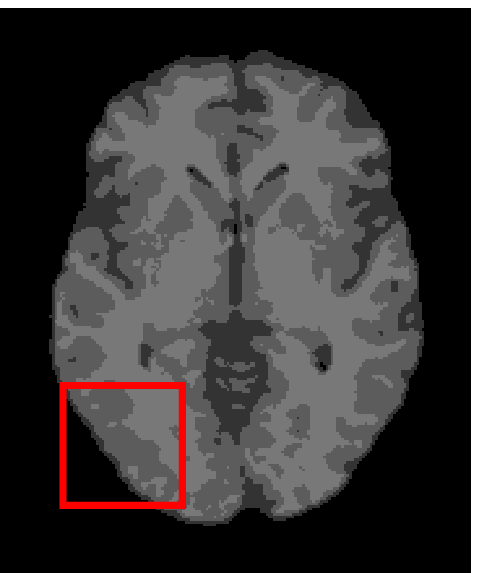}
%\centerline{(a)}
\end{minipage}
\begin{minipage}[t]{0.187\linewidth}
\centering
\includegraphics[width=1\textwidth]{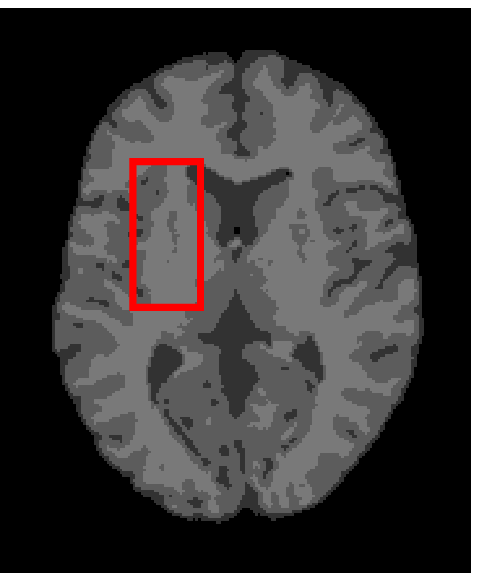}
%\centerline{(b)}
\end{minipage}
\begin{minipage}[t]{0.187\linewidth}
\centering
\includegraphics[width=1\textwidth]{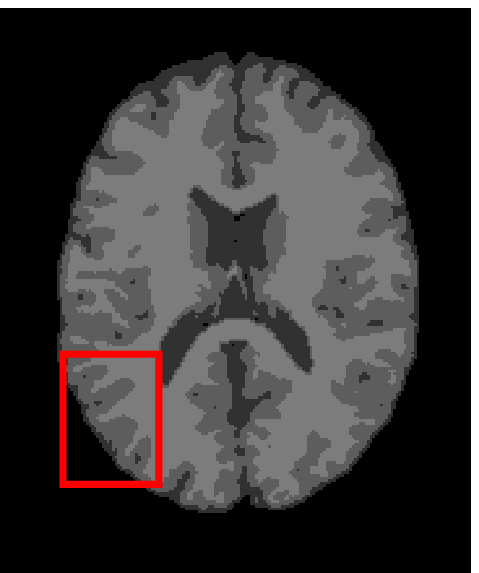}
%\centerline{(c)}
\end{minipage}
\begin{minipage}[t]{0.187\linewidth}
\centering
\includegraphics[width=1\textwidth]{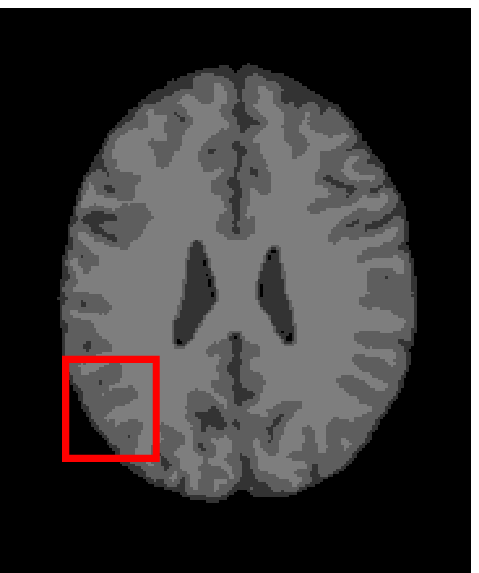}
%\centerline{(d)}
\end{minipage}
\begin{minipage}[t]{0.187\linewidth}
\centering
\includegraphics[width=1\textwidth]{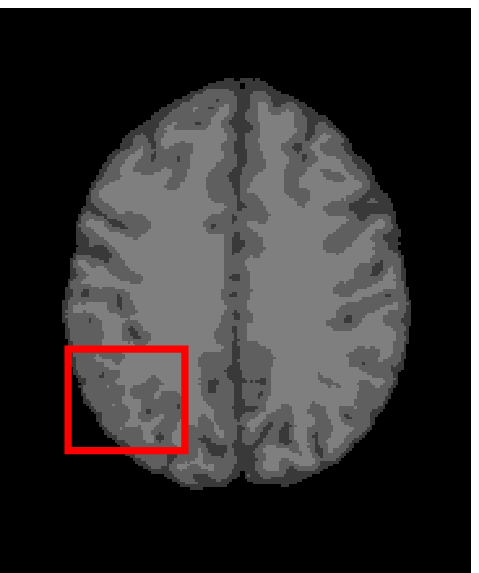}
%\centerline{(e)}
\end{minipage}\\
\begin{minipage}[t]{0.187\linewidth}
\centering
\includegraphics[width=1\textwidth]{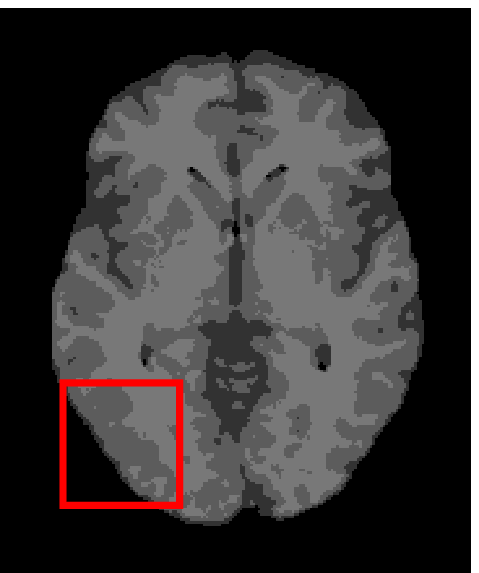}
%\centerline{(a)}
\end{minipage}
\begin{minipage}[t]{0.187\linewidth}
\centering
\includegraphics[width=1\textwidth]{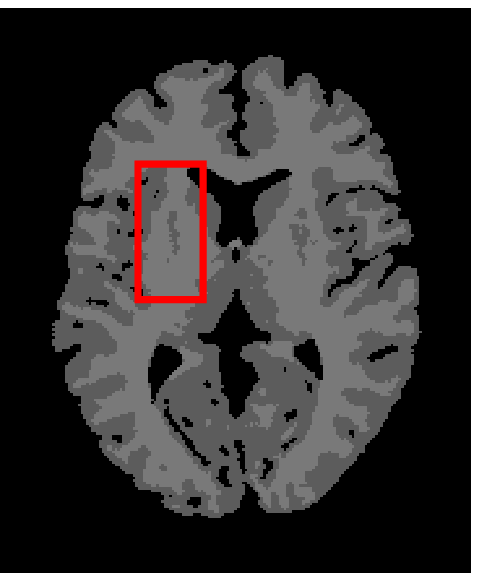}
%\centerline{(b)}
\end{minipage}
\begin{minipage}[t]{0.187\linewidth}
\centering
\includegraphics[width=1\textwidth]{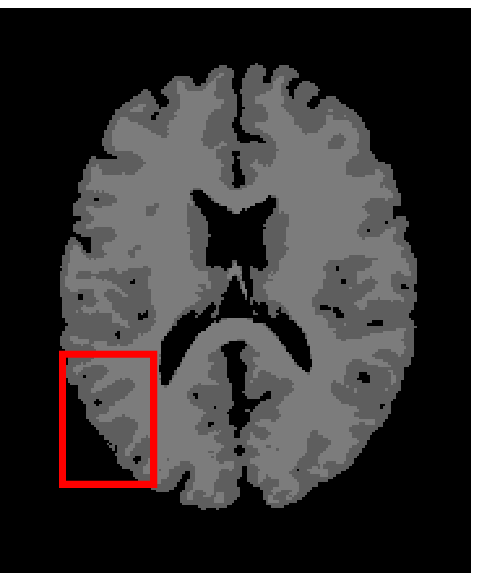}
%\centerline{(c)}
\end{minipage}
\begin{minipage}[t]{0.187\linewidth}
\centering
\includegraphics[width=1\textwidth]{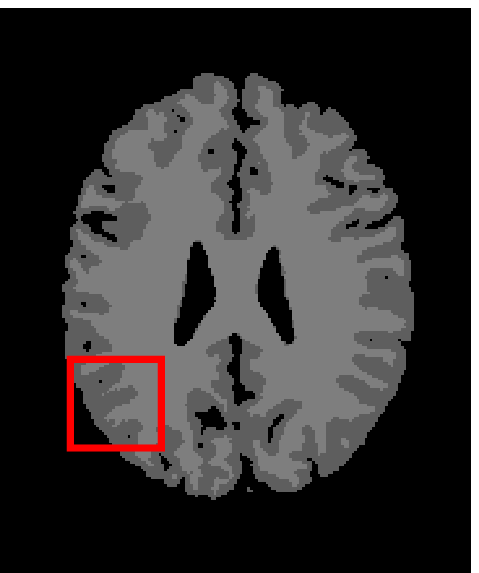}
%\centerline{(d)}
\end{minipage}
\begin{minipage}[t]{0.187\linewidth}
\centering
\includegraphics[width=1\textwidth]{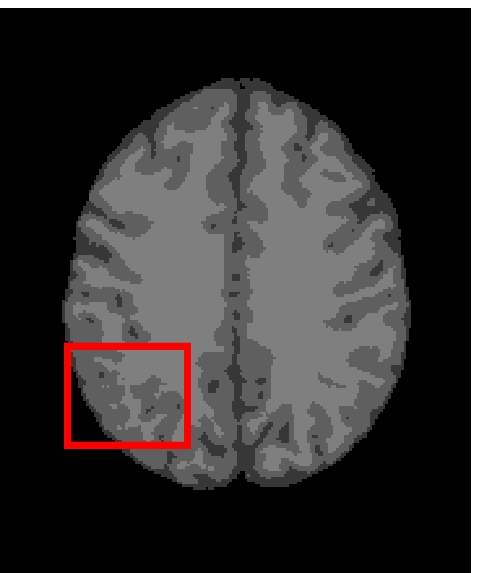}
%\centerline{(e)}
\end{minipage}\\
\begin{minipage}[t]{0.187\linewidth}
\centering
\includegraphics[width=1\textwidth]{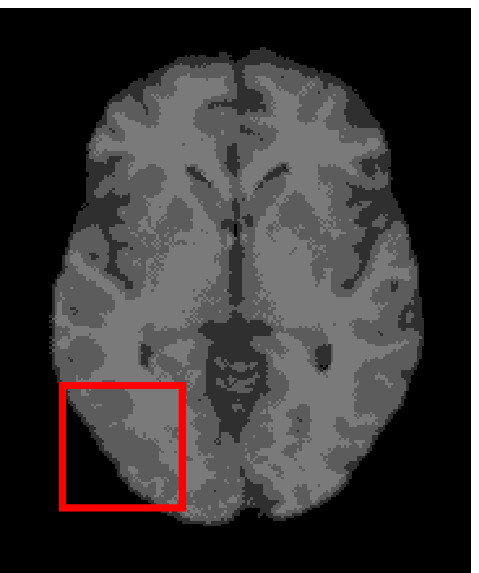}
%\centerline{(a)}
\end{minipage}
\begin{minipage}[t]{0.187\linewidth}
\centering
\includegraphics[width=1\textwidth]{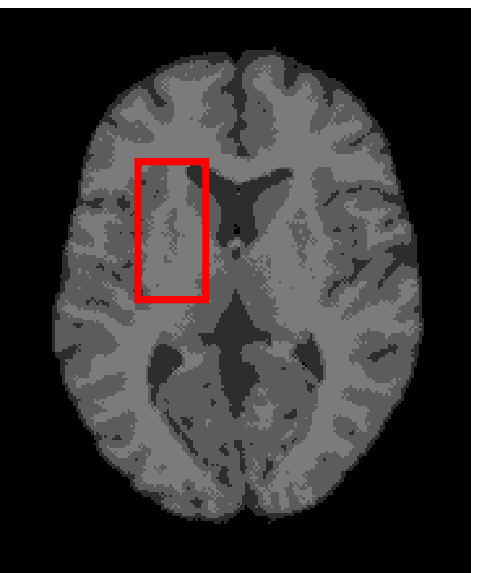}
%\centerline{(b)}
\end{minipage}
\begin{minipage}[t]{0.187\linewidth}
\centering
\includegraphics[width=1\textwidth]{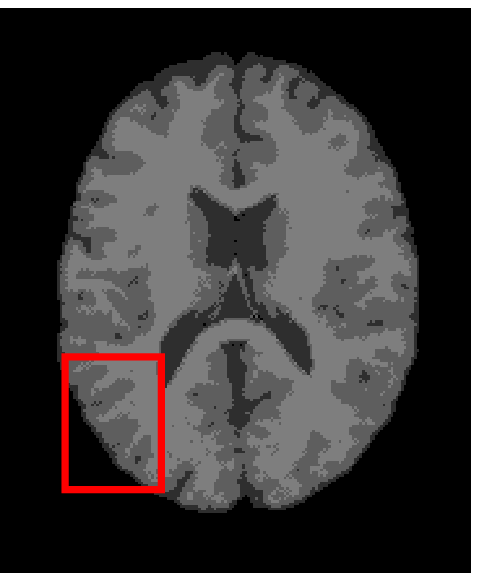}
%\centerline{(c)}
\end{minipage}
\begin{minipage}[t]{0.187\linewidth}
\centering
\includegraphics[width=1\textwidth]{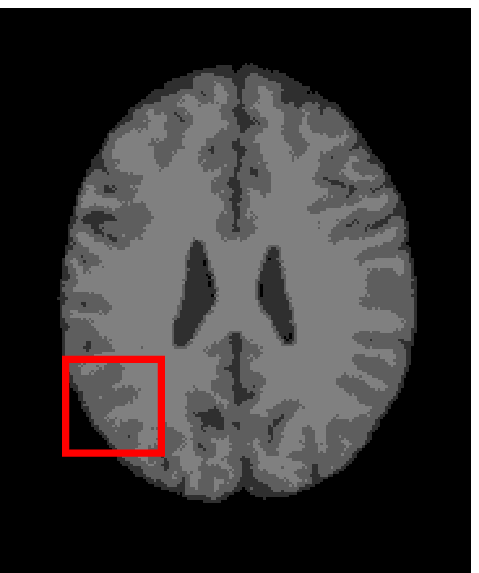}
%\centerline{(d)}
\end{minipage}
\begin{minipage}[t]{0.187\linewidth}
\centering
\includegraphics[width=1\textwidth]{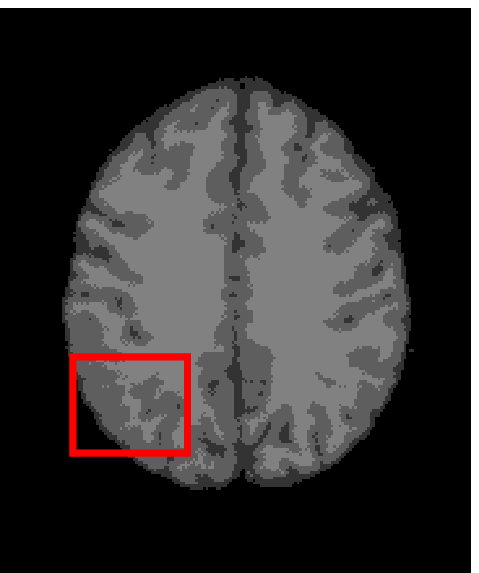}
%\centerline{(e)}
\end{minipage}\\
\begin{minipage}[t]{0.187\linewidth}
\centering
\includegraphics[width=1\textwidth]{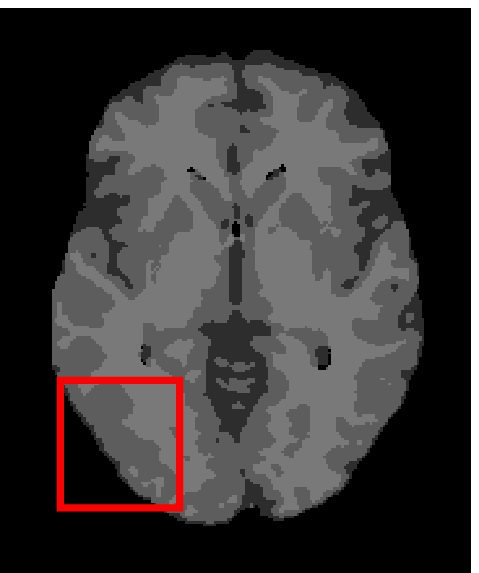}
%\centerline{(a)}
\end{minipage}
\begin{minipage}[t]{0.187\linewidth}
\centering
\includegraphics[width=1\textwidth]{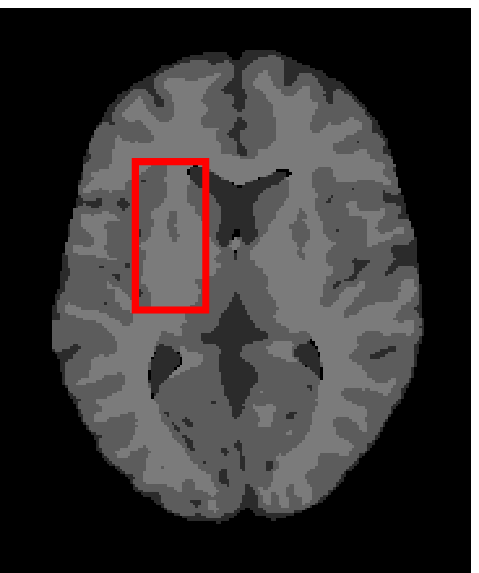}
%\centerline{(b)}
\end{minipage}
\begin{minipage}[t]{0.187\linewidth}
\centering
\includegraphics[width=1\textwidth]{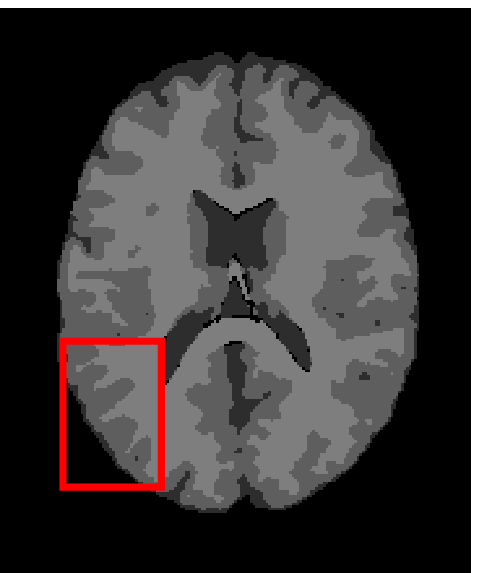}
%\centerline{(c)}
\end{minipage}
\begin{minipage}[t]{0.187\linewidth}
\centering
\includegraphics[width=1\textwidth]{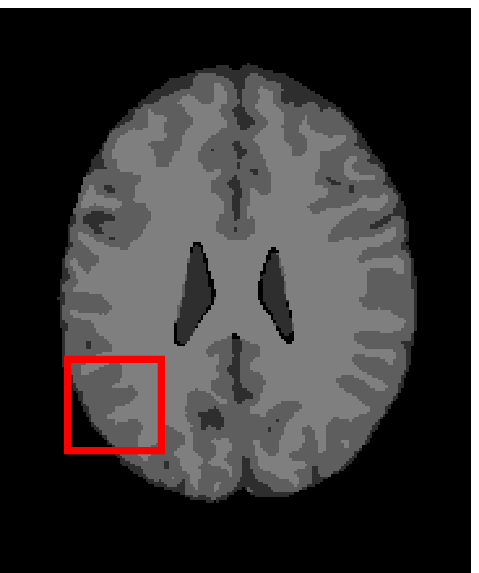}
%\centerline{(d)}
\end{minipage}
\begin{minipage}[t]{0.187\linewidth}
\centering
\includegraphics[width=1\textwidth]{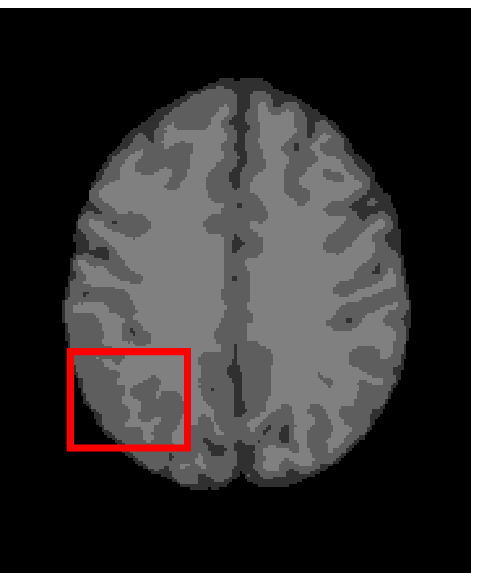}
%\centerline{(e)}
\end{minipage}\\
\begin{minipage}[t]{0.187\linewidth}
\centering
\includegraphics[width=1\textwidth]{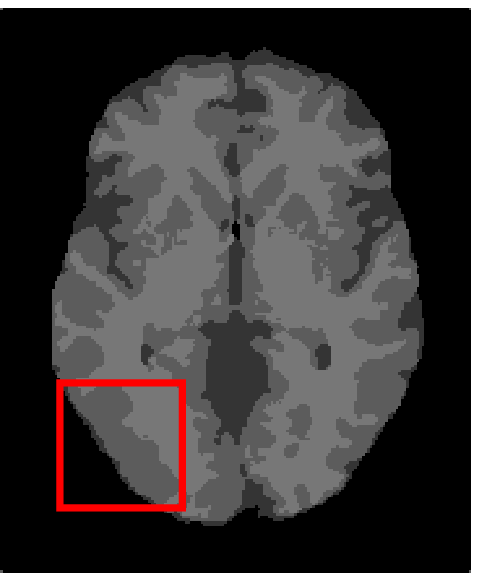}
%\centerline{(a)}
\end{minipage}
\begin{minipage}[t]{0.187\linewidth}
\centering
\includegraphics[width=1\textwidth]{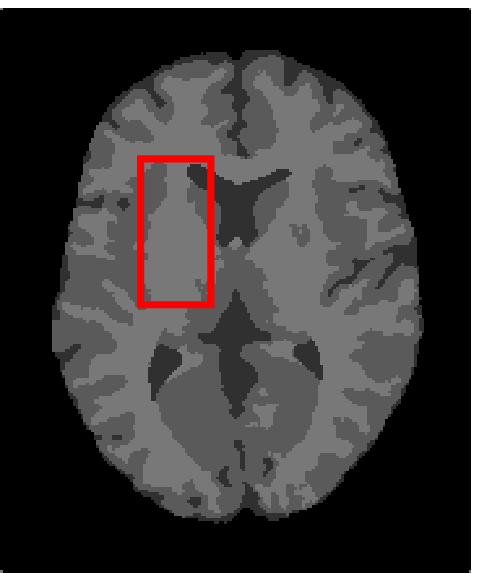}
%\centerline{(b)}
\end{minipage}
\begin{minipage}[t]{0.187\linewidth}
\centering
\includegraphics[width=1\textwidth]{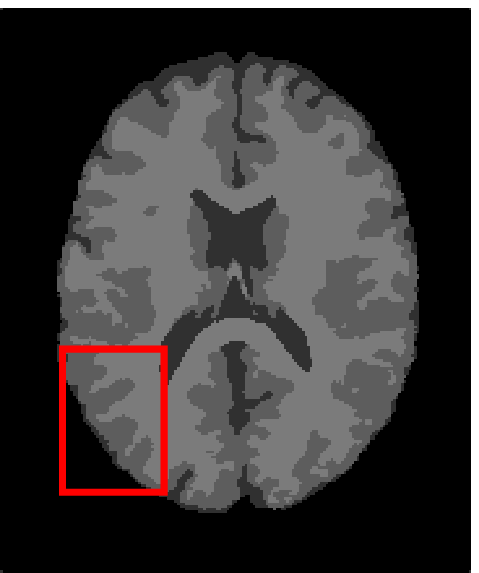}
%\centerline{(c)}
\end{minipage}
\begin{minipage}[t]{0.187\linewidth}
\centering
\includegraphics[width=1\textwidth]{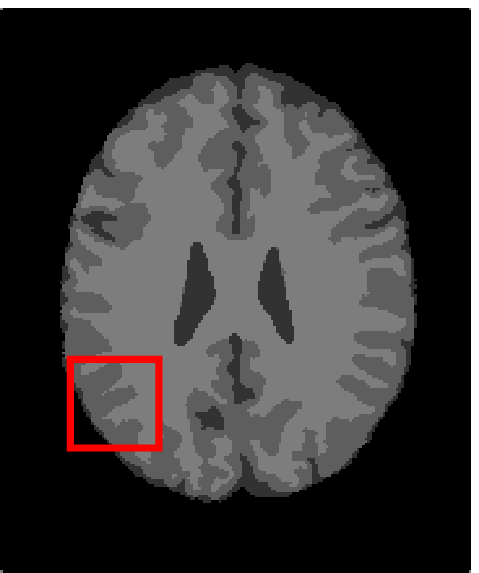}
%\centerline{(d)}
\end{minipage}
\begin{minipage}[t]{0.187\linewidth}
\centering
\includegraphics[width=1\textwidth]{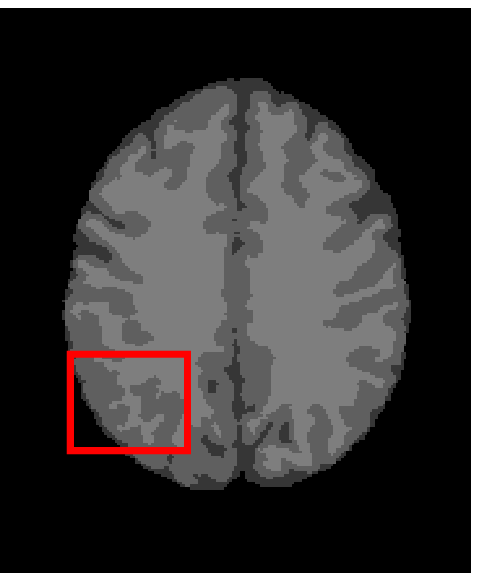}
%\centerline{(e)}
\end{minipage}\\
\begin{minipage}[t]{0.187\linewidth}
\centering
\includegraphics[width=1\textwidth]{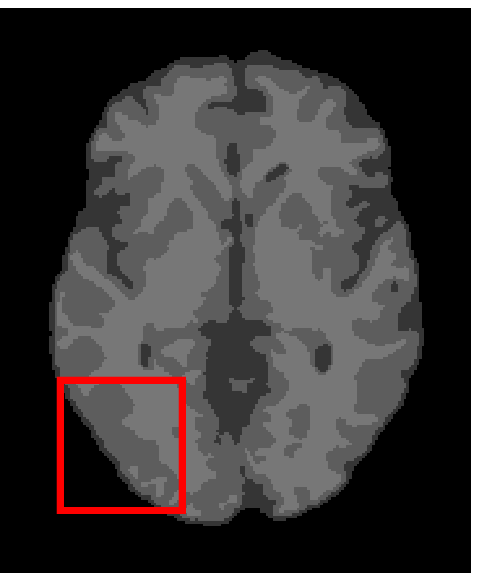}
%\centerline{(a)}
\end{minipage}
\begin{minipage}[t]{0.187\linewidth}
\centering
\includegraphics[width=1\textwidth]{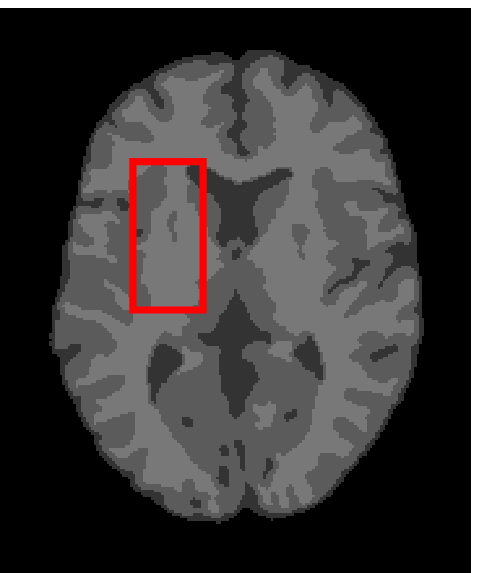}
%\centerline{(b)}
\end{minipage}
\begin{minipage}[t]{0.187\linewidth}
\centering
\includegraphics[width=1\textwidth]{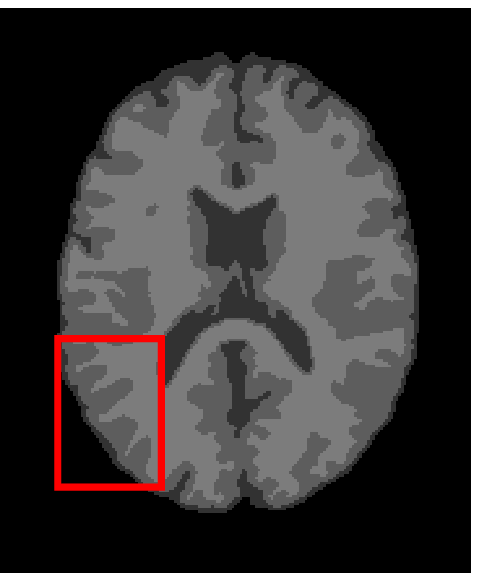}
%\centerline{(c)}
\end{minipage}
\begin{minipage}[t]{0.187\linewidth}
\centering
\includegraphics[width=1\textwidth]{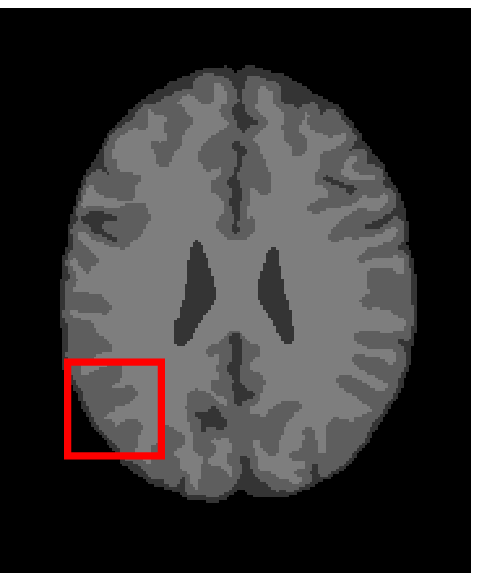}
%\centerline{(d)}
\end{minipage}
\begin{minipage}[t]{0.187\linewidth}
\centering
\includegraphics[width=1\textwidth]{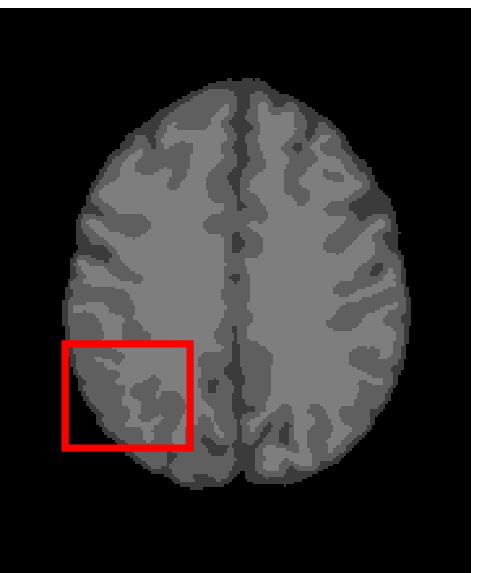}
%\centerline{(e)}
\end{minipage}\\
\begin{minipage}[t]{0.187\linewidth}
\centering
\includegraphics[width=1\textwidth]{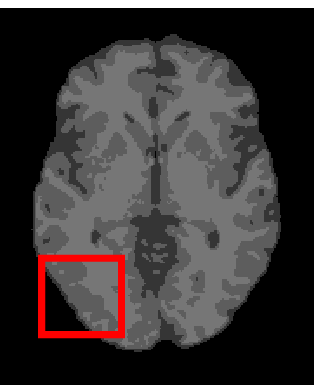}
%\centerline{(a)}
\end{minipage}
\begin{minipage}[t]{0.187\linewidth}
\centering
\includegraphics[width=1\textwidth]{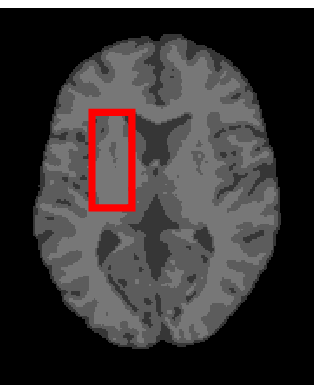}
%\centerline{(b)}
\end{minipage}
\begin{minipage}[t]{0.187\linewidth}
\centering
\includegraphics[width=1\textwidth]{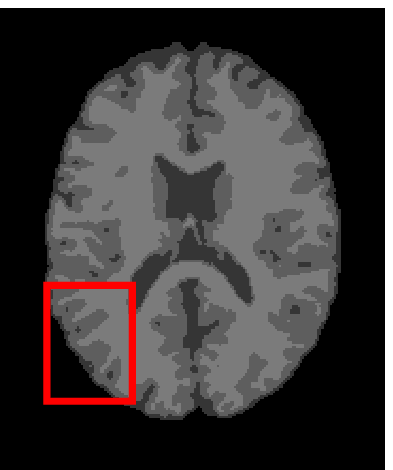}
%\centerline{(c)}
\end{minipage}
\begin{minipage}[t]{0.187\linewidth}
\centering
\includegraphics[width=1\textwidth]{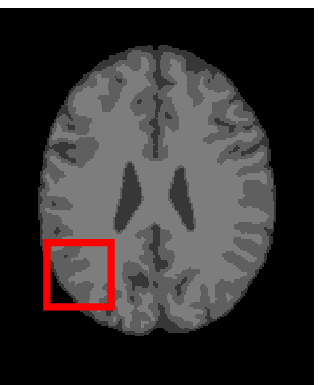}
%\centerline{(d)}
\end{minipage}
\begin{minipage}[t]{0.187\linewidth}
\centering
\includegraphics[width=1\textwidth]{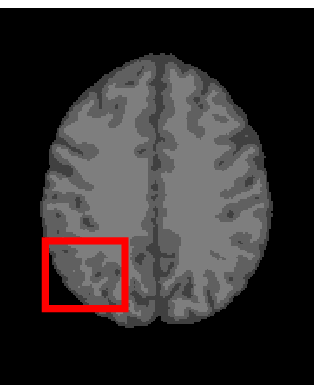}
%\centerline{(e)}
\end{minipage}\\
\begin{minipage}[t]{0.187\linewidth}
\centering
\includegraphics[width=1\textwidth]{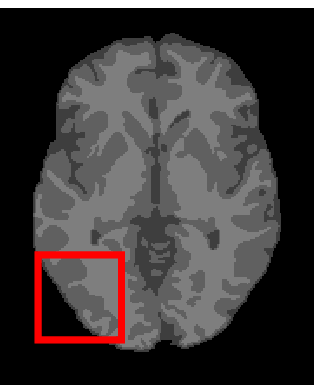}
%\centerline{(a)}
\end{minipage}
\begin{minipage}[t]{0.187\linewidth}
\centering
\includegraphics[width=1\textwidth]{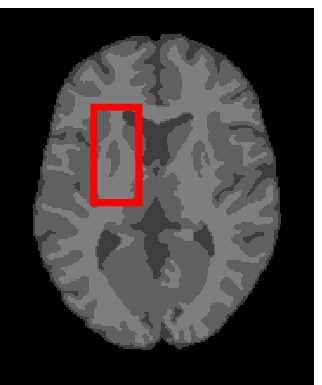}
%\centerline{(b)}
\end{minipage}
\begin{minipage}[t]{0.187\linewidth}
\centering
\includegraphics[width=1\textwidth]{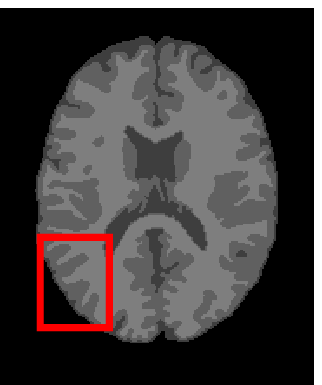}
%\centerline{(c)}
\end{minipage}
\begin{minipage}[t]{0.187\linewidth}
\centering
\includegraphics[width=1\textwidth]{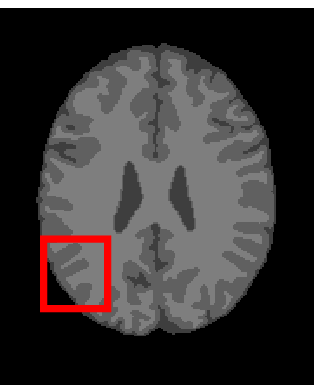}
%\centerline{(d)}
\end{minipage}
\begin{minipage}[t]{0.187\linewidth}
\centering
\includegraphics[width=1\textwidth]{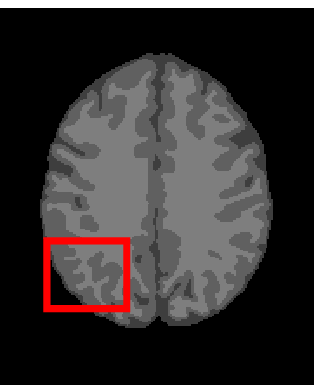}
%\centerline{(e)}
\end{minipage}
\caption{Segmentation results on five medical images. The parameter: $\phi=5.35$. From top to bottom: noisy images, ground truth, and results of FCM\_S1, FCM\_S2, FLICM, KWFLICM, FRFCM, WFCM, DSFCM\_N, and WRFCM.}
\label{fig:med}
\end{figure}

\begin{table*}[htbp]
    \setlength{\abovecaptionskip}{0pt}
\setlength{\belowcaptionskip}{0pt}
%\tabcolsep 0.008\textwidth
  \centering
  \caption{Segmentation performance (\%) on medical images in BrianWeb}
  \scriptsize
    \begin{tabular}{c|ccc|ccc|ccc|ccc|ccc}
    \toprule
    \multirow{2}[3]{*}{Algorithm} &
      \multicolumn{3}{c|}{Fig. 9 column 1} &
      \multicolumn{3}{c|}{Fig. 9 column 2} &
      \multicolumn{3}{c|}{Fig. 9 column 3} &
      \multicolumn{3}{c|}{Fig. 9 column 4} &
      \multicolumn{3}{c}{Fig. 9 column 5}
      \\
\cmidrule{2-16}     &
      SA &
      SDS &
      MCC &
      SA &
      SDS &
      MCC &
      SA &
      SDS &
      MCC &
      SA &
      SDS &
      MCC &
      SA &
      SDS &
      MCC
      \\
      \midrule
    FCM\_S1 &
      75.756 &
      97.852 &
      96.225 &
      75.026 &
      98.109 &
      96.656 &
      79.792 &
      98.452 &
      97.334 &
      81.887 &
      98.614 &
      97.680 &
      81.869 &
      94.254 &
      90.947
      \\
    FCM\_S2 &
      75.769 &
      98.119 &
      96.664 &
      74.970 &
      98.176 &
      96.765 &
      79.886 &
      98.458 &
      97.338 &
      82.073 &
      98.625 &
      97.695 &
      81.788 &
      98.223 &
      97.195
      \\
    FLICM &
      74.998 &
      98.070 &
      96.568 &
      74.185 &
      98.122 &
      96.660 &
      79.099 &
      98.515 &
      97.432 &
      81.447 &
      98.627 &
      97.691 &
      81.668 &
      98.273 &
      97.260
      \\
    KWFLICM &
      74.840 &
      98.259 &
      96.878 &
      73.839 &
      97.860 &
      96.190 &
      79.560 &
      98.453 &
      97.316 &
      81.887 &
      98.482 &
      97.443 &
      81.370 &
      98.297 &
      97.286
      \\
    FRFCM &
      75.853 &
      97.620 &
      95.775 &
      75.514 &
      97.660 &
      95.830 &
      80.283 &
      98.278 &
      97.013 &
      81.852 &
      98.319 &
      97.171 &
      81.666 &
      98.079 &
      96.945
      \\
    WFCM &
      75.507 &
      97.124 &
      94.957 &
      74.471 &
      97.213 &
      95.045 &
      79.316 &
      97.845 &
      96.283 &
      81.358 &
      97.546 &
      95.211 &
      81.452 &
      95.247 &
      92.501
      \\
    DSFCM\_N &
      76.400 &
      92.325 &
      86.262 &
      75.288 &
      91.574 &
      85.095 &
      79.861 &
      97.678 &
      95.996 &
      81.831 &
      93.304 &
      88.829 &
      81.750 &
      94.302 &
      91.024
      \\
    WRFCM &
      \textbf{82.317} &
      \textbf{98.966} &
      \textbf{98.147} &
      \textbf{82.141} &
      \textbf{98.298} &
      \textbf{96.970} &
      \textbf{83.914} &
      \textbf{98.963} &
      \textbf{98.202} &
      \textbf{83.533} &
      \textbf{99.170} &
      \textbf{98.603} &
      \textbf{84.615} &
      \textbf{98.429} &
      \textbf{97.511}
      \\
    \bottomrule
    \end{tabular}%
  \label{tab:med}%
\end{table*}%

By a view of the marked red square in Fig. \ref{fig:med}, we find that FCM\_S1, FCM\_S2, FLICM, KWFLICM and DSFCM\_N are vulnerable to noise and intensity non-uniformity. They give rise to the change of topological shapes to some extent. Unlike them, FRFCM and WFCM achieve sufficient noise removal. However, they produce overly smooth contours. Compared with its seven peers, WRFCM can not only suppress noise adequately but also acquire accurate contours. Moreover, it yields the visual result closer to ground truth than its peers. As Table \ref{tab:med} shows, WRFCM obtains optimal SA, SDS and MCC results for all five medical images. As a conclusion, it outperforms its peers visually and quantitatively.

\subsubsection{Results on Real-world Images}
In order to demonstrate the practicality of \mbox{WRFCM} for other image segmentation, we typically choose two sets of \mbox{real-world} images in the last experiment. The first set contains five representative images from BSDS and MSRC. There usually exist some outliers, noise or intensity inhomogeneity in each image. For all tested images, we set $c=2$. The segmentation results of all algorithms are shown in Fig. \ref{fig:BM} and Table \ref{tab:BM}.
\begin{table*}[htbp]
    \setlength{\abovecaptionskip}{0pt}
\setlength{\belowcaptionskip}{0pt}
%\tabcolsep 0.008\textwidth
  \centering
  \caption{Segmentation performance (\%) on real-world Images in BSDS and MSRC}
  \scriptsize
    \begin{tabular}{c|ccc|ccc|ccc|ccc|ccc}
    \toprule
    \multirow{2}[3]{*}{Algorithm} &
      \multicolumn{3}{c|}{Fig. 10 column 1} &
      \multicolumn{3}{c|}{Fig. 10 column 2} &
      \multicolumn{3}{c|}{Fig. 10 column 3} &
      \multicolumn{3}{c|}{Fig. 10 column 4} &
      \multicolumn{3}{c}{Fig. 10 column 5}
      \\
\cmidrule{2-16}     &
      SA &
      SDS &
      MCC &
      SA &
      SDS &
      MCC &
      SA &
      SDS &
      MCC &
      SA &
      SDS &
      MCC &
      SA &
      SDS &
      MCC
      \\
      \midrule
    FCM\_S1 &
      86.384 &
      89.687 &
      69.705 &
      50.997 &
      66.045 &
      2.724 &
      67.289 &
      72.570 &
      32.232 &
      80.688 &
      88.159 &
      49.369 &
      78.717 &
      47.696 &
      48.874
      \\
    FCM\_S2 &
      86.138 &
      79.701 &
      69.208 &
      51.433 &
      12.089 &
      2.951 &
      67.105 &
      59.523 &
      31.941 &
      80.657 &
      47.557 &
      49.256 &
      78.365 &
      86.449 &
      47.881
      \\
    FLICM &
      86.476 &
      89.771 &
      69.882 &
      55.292 &
      70.055 &
      2.403 &
      89.233 &
      91.167 &
      78.117 &
      80.771 &
      47.826 &
      49.729 &
      80.617 &
      54.490 &
      54.029
      \\
    KWFLICM &
      87.119 &
      90.278 &
      71.283 &
      48.252 &
      63.432 &
      1.554 &
      64.617 &
      66.081 &
      30.820 &
      80.484 &
      46.723 &
      48.777 &
      77.963 &
      44.791 &
      46.755
      \\
    FRFCM &
      97.701 &
      \textbf{98.235} &
      94.941 &
      99.690 &
      97.436 &
      97.273 &
      99.380 &
      99.467 &
      98.732 &
      83.974 &
      89.927 &
      58.683 &
      96.985 &
      97.861 &
      92.987
      \\
    WFCM &
      98.442 &
      97.755 &
      96.563 &
      99.688 &
      \textbf{99.834} &
      97.268 &
      99.295 &
      99.160 &
      98.555 &
      84.480 &
      62.664 &
      60.043 &
      96.445 &
      93.943 &
      91.719
      \\
    DSFCM\_N &
      93.116 &
      90.279 &
      84.987 &
      50.688 &
      11.093 &
      0.638 &
      92.101 &
      90.791 &
      83.922 &
      50.858 &
      60.181 &
      0.506 &
      95.412 &
      92.319 &
      89.179
      \\
    WRFCM &
      \textbf{98.732} &
      98.162 &
      \textbf{97.201} &
      \textbf{99.746} &
      97.906 &
      \textbf{97.771} &
      \textbf{99.442} &
      \textbf{99.520} &
      \textbf{98.857} &
      \textbf{99.826} &
      \textbf{99.888} &
      \textbf{99.074} &
      \textbf{99.869} &
      \textbf{99.789} &
      \textbf{99.694}
      \\
    \bottomrule
    \end{tabular}%
  \label{tab:BM}%
\end{table*}%

\begin{figure}[htb]
\centering
\begin{minipage}[t]{0.19\linewidth}
\centering
\includegraphics[width=1\textwidth]{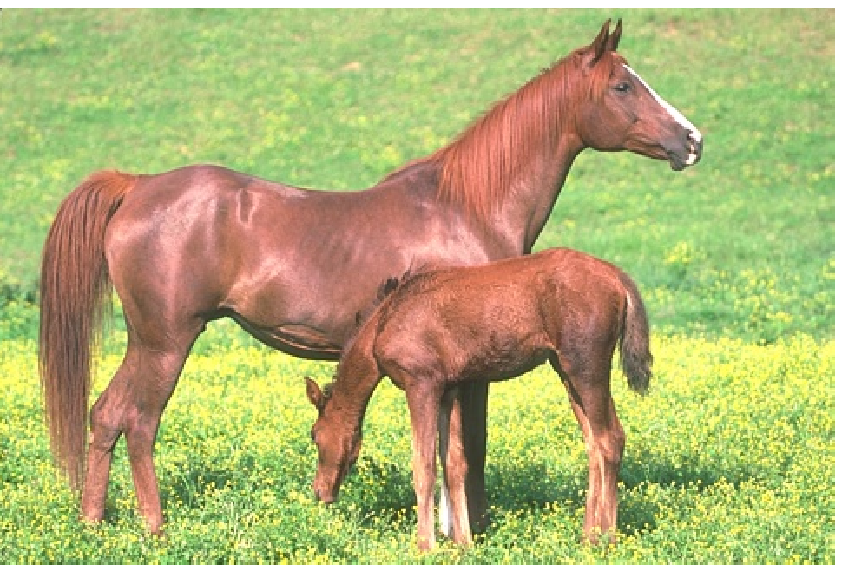}
%\centerline{Noisy}
\end{minipage}
\begin{minipage}[t]{0.19\linewidth}
\centering
\includegraphics[width=1\textwidth]{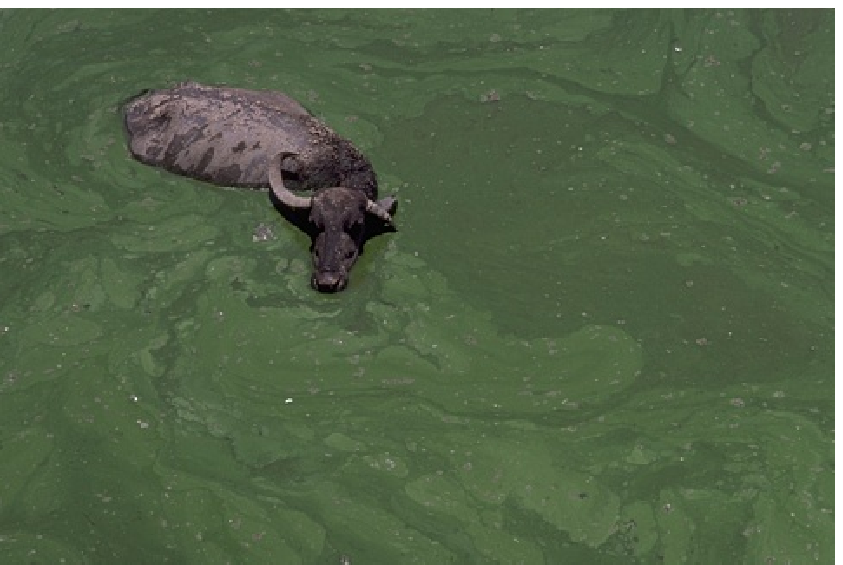}
%\centerline{Noisy}
\end{minipage}
\begin{minipage}[t]{0.19\linewidth}
\centering
\includegraphics[width=1\textwidth]{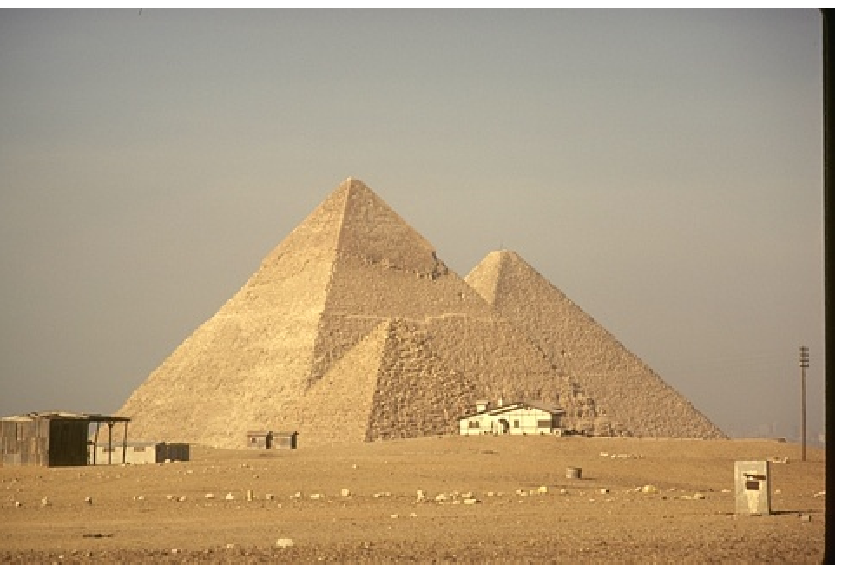}
%\centerline{Noisy}
\end{minipage}
\begin{minipage}[t]{0.19\linewidth}
\centering
\includegraphics[width=1\textwidth]{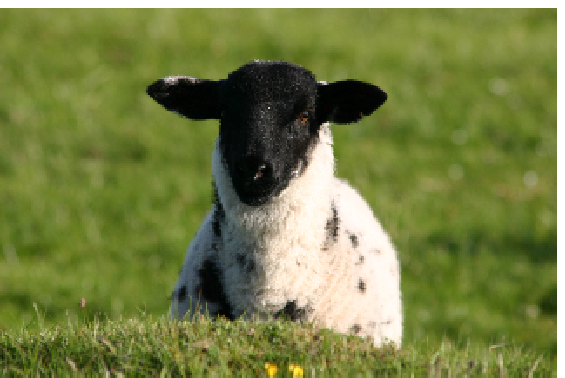}
%\centerline{Noisy}
\end{minipage}
\begin{minipage}[t]{0.19\linewidth}
\centering
\includegraphics[width=1\textwidth]{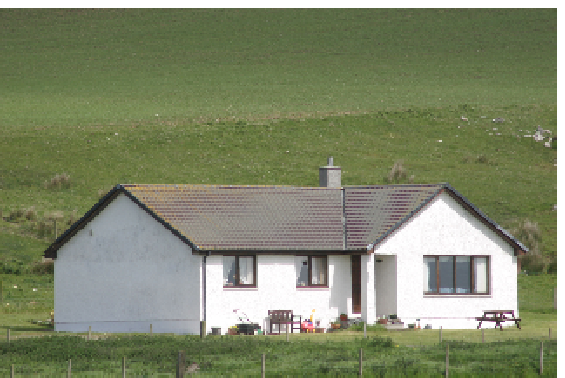}
%\centerline{Noisy}
\end{minipage}\\
\begin{minipage}[t]{0.19\linewidth}
\centering
\includegraphics[width=1\textwidth]{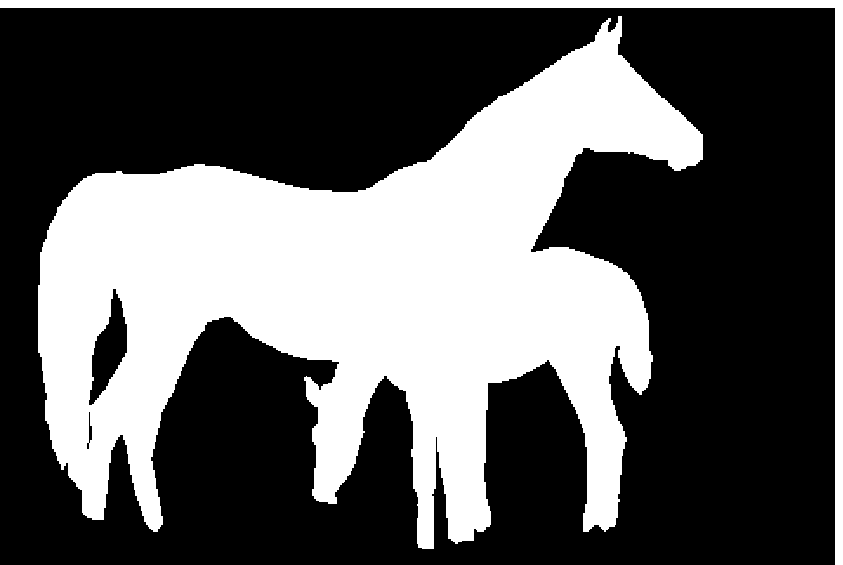}
%\centerline{Noisy}
\end{minipage}
\begin{minipage}[t]{0.19\linewidth}
\centering
\includegraphics[width=1\textwidth]{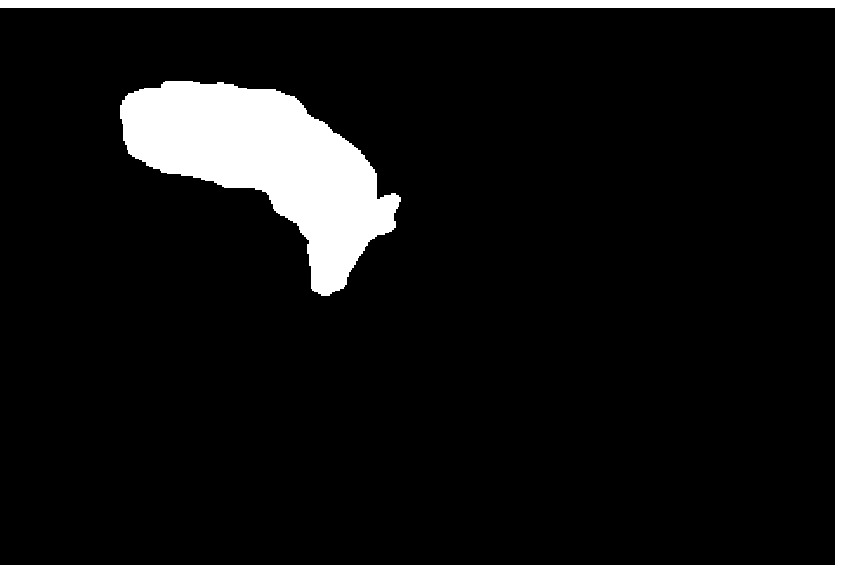}
%\centerline{Noisy}
\end{minipage}
\begin{minipage}[t]{0.19\linewidth}
\centering
\includegraphics[width=1\textwidth]{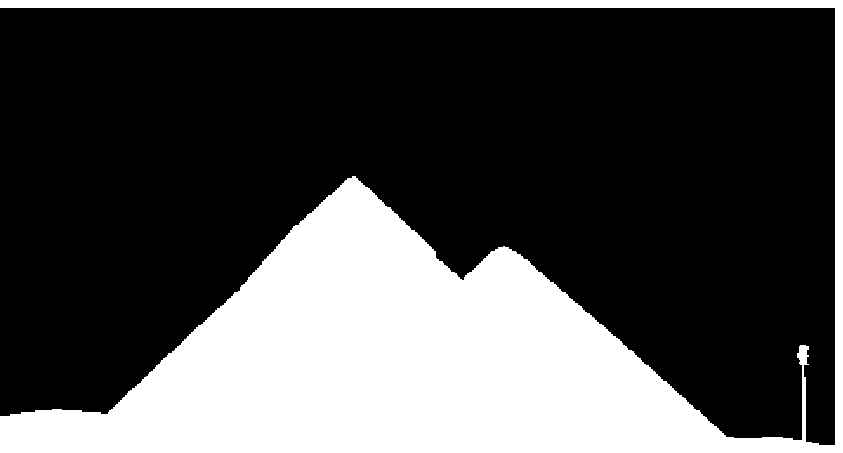}
%\centerline{Noisy}
\end{minipage}
\begin{minipage}[t]{0.19\linewidth}
\centering
\includegraphics[width=1\textwidth]{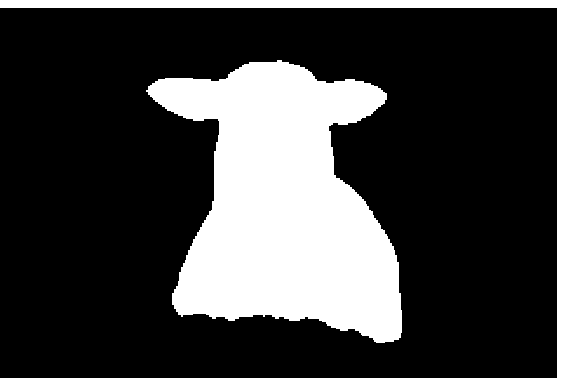}
%\centerline{Noisy}
\end{minipage}
\begin{minipage}[t]{0.19\linewidth}
\centering
\includegraphics[width=1\textwidth]{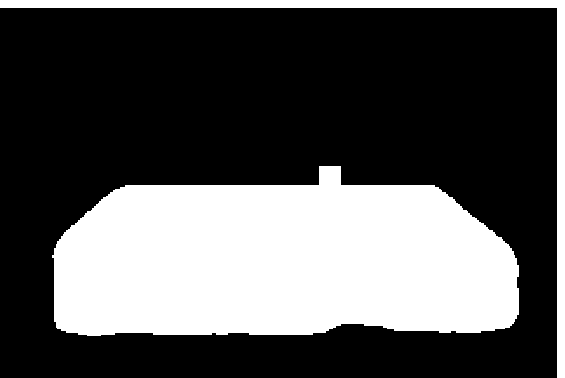}
%\centerline{Noisy}
\end{minipage}\\
\begin{minipage}[t]{0.19\linewidth}
\centering
\includegraphics[width=1\textwidth]{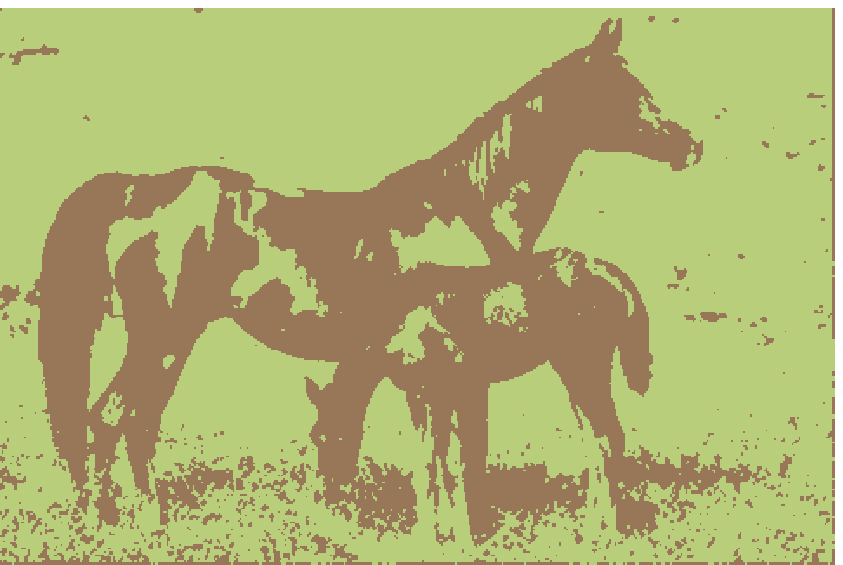}
%\centerline{Noisy}
\end{minipage}
\begin{minipage}[t]{0.19\linewidth}
\centering
\includegraphics[width=1\textwidth]{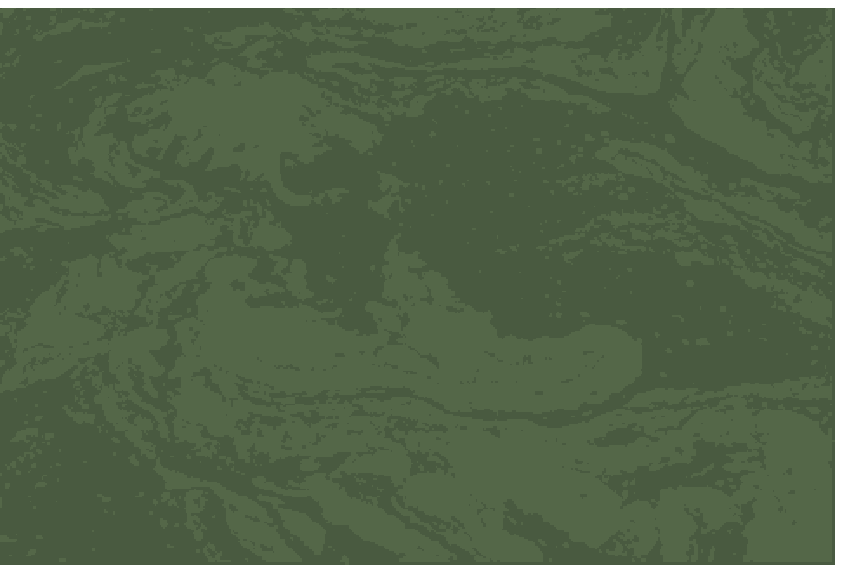}
%\centerline{Noisy}
\end{minipage}
\begin{minipage}[t]{0.19\linewidth}
\centering
\includegraphics[width=1\textwidth]{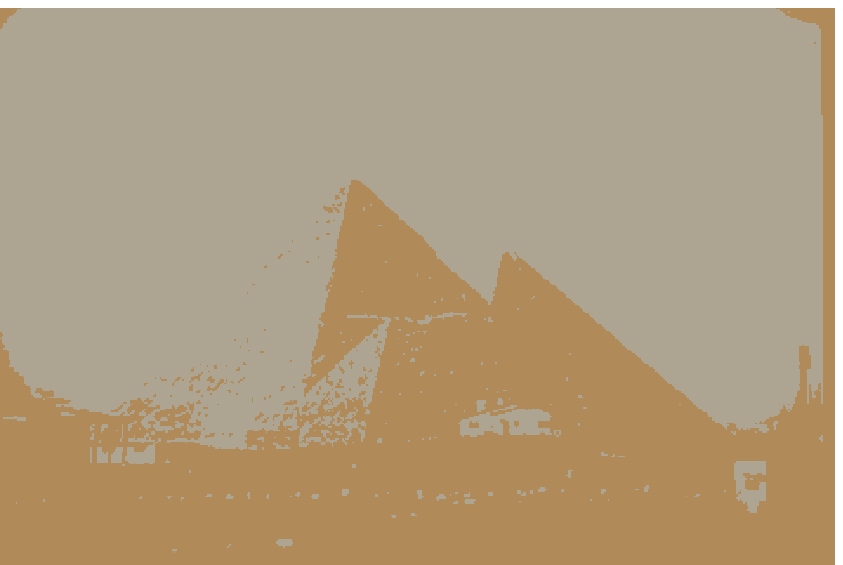}
%\centerline{Noisy}
\end{minipage}
\begin{minipage}[t]{0.19\linewidth}
\centering
\includegraphics[width=1\textwidth]{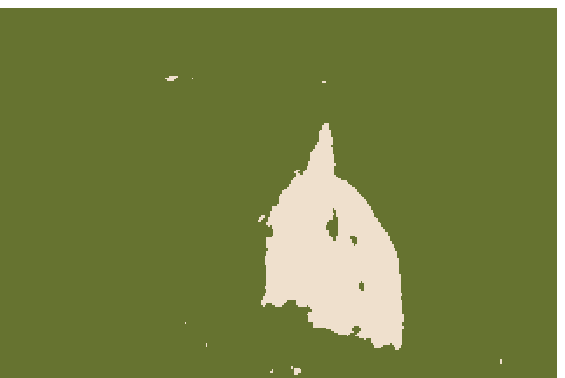}
%\centerline{Noisy}
\end{minipage}
\begin{minipage}[t]{0.19\linewidth}
\centering
\includegraphics[width=1\textwidth]{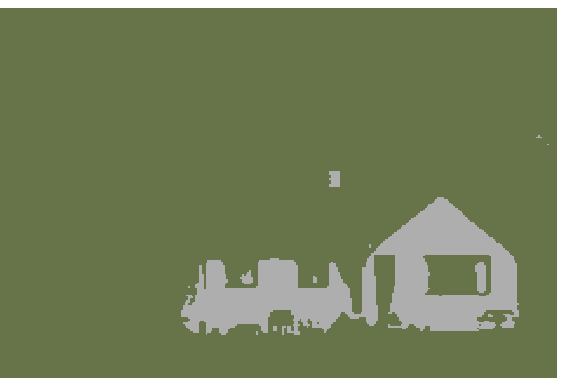}
%\centerline{Noisy}
\end{minipage}\\
\begin{minipage}[t]{0.19\linewidth}
\centering
\includegraphics[width=1\textwidth]{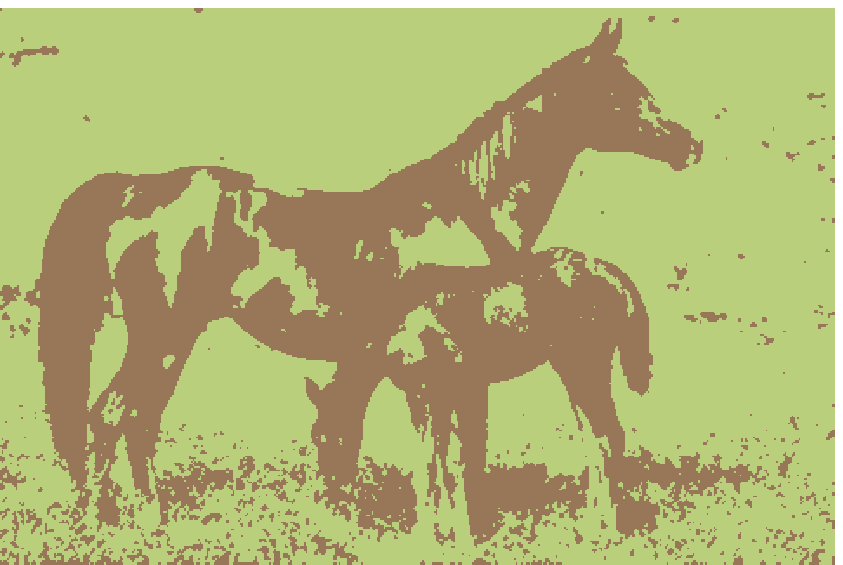}
%\centerline{Noisy}
\end{minipage}
\begin{minipage}[t]{0.19\linewidth}
\centering
\includegraphics[width=1\textwidth]{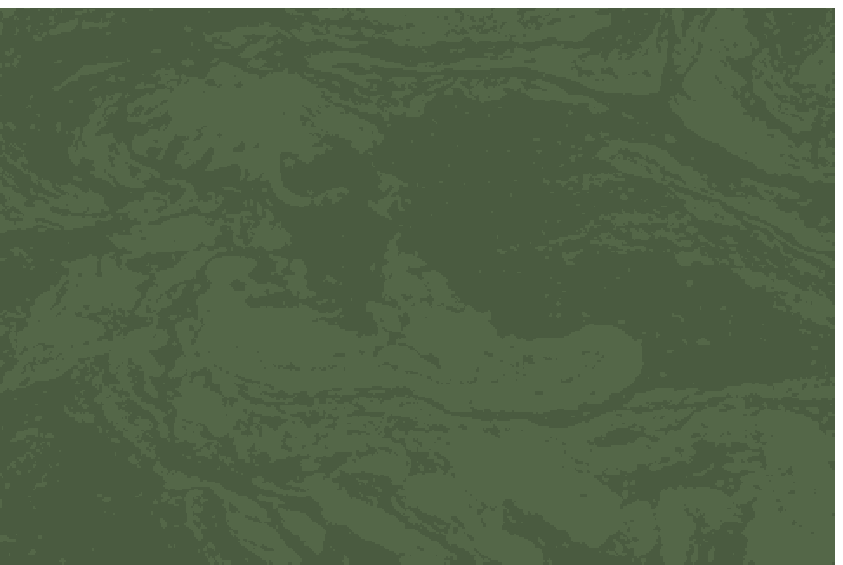}
%\centerline{Noisy}
\end{minipage}
\begin{minipage}[t]{0.19\linewidth}
\centering
\includegraphics[width=1\textwidth]{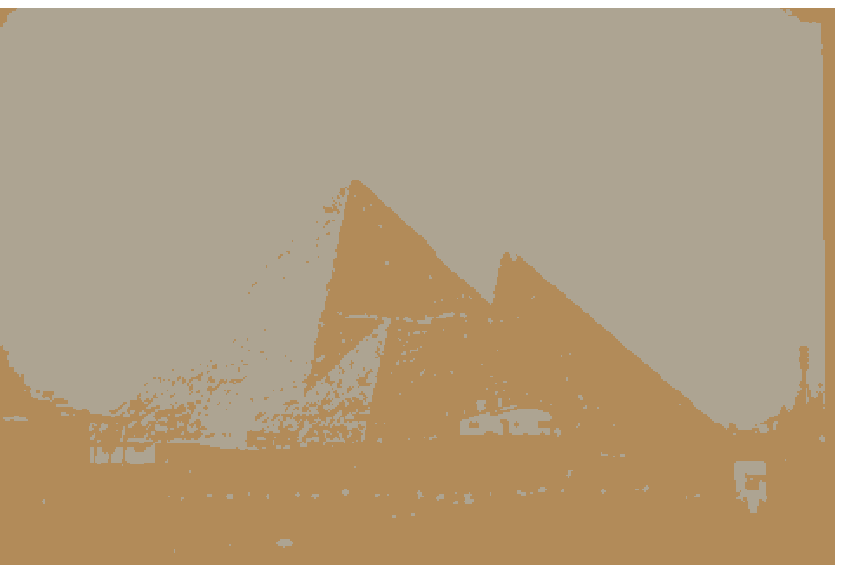}
%\centerline{Noisy}
\end{minipage}
\begin{minipage}[t]{0.19\linewidth}
\centering
\includegraphics[width=1\textwidth]{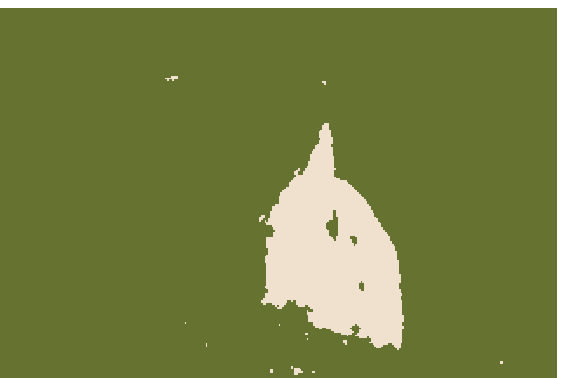}
%\centerline{Noisy}
\end{minipage}
\begin{minipage}[t]{0.19\linewidth}
\centering
\includegraphics[width=1\textwidth]{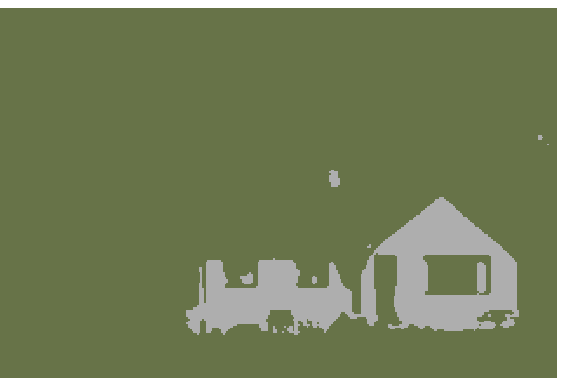}
%\centerline{Noisy}
\end{minipage}\\
\begin{minipage}[t]{0.19\linewidth}
\centering
\includegraphics[width=1\textwidth]{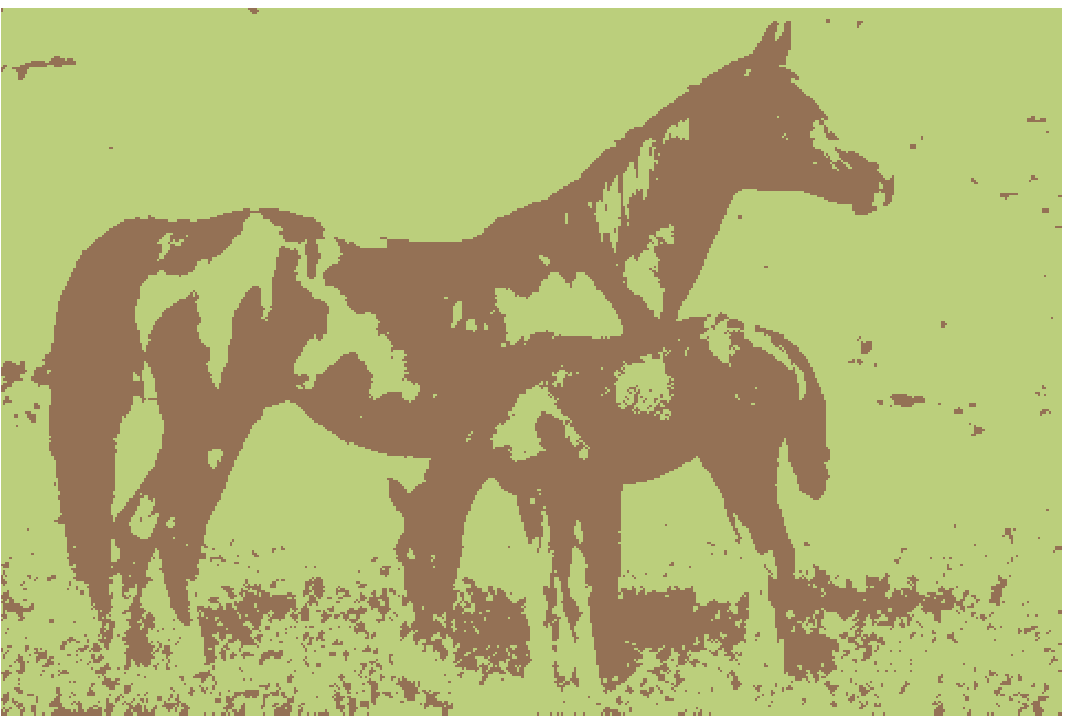}
%\centerline{Noisy}
\end{minipage}
\begin{minipage}[t]{0.19\linewidth}
\centering
\includegraphics[width=1\textwidth]{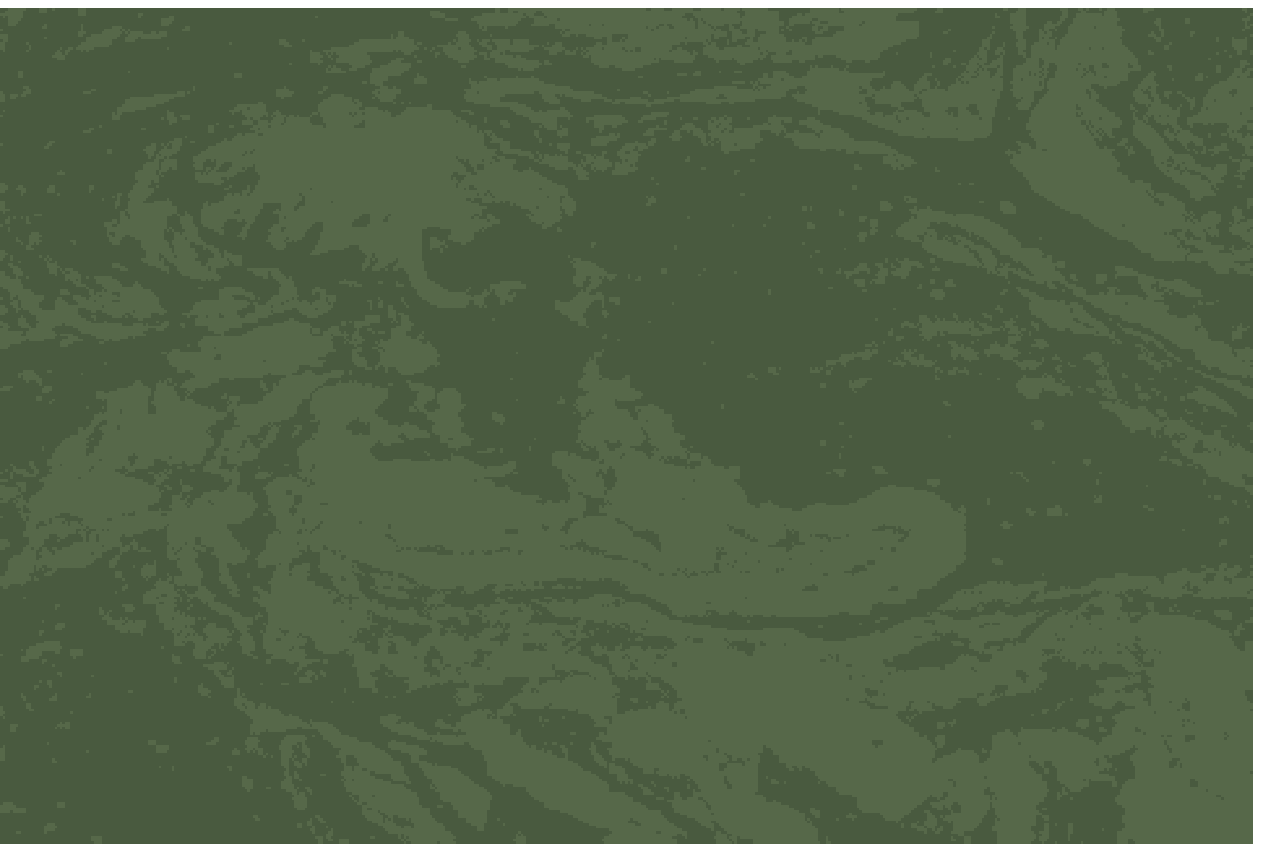}
%\centerline{Noisy}
\end{minipage}
\begin{minipage}[t]{0.19\linewidth}
\centering
\includegraphics[width=1\textwidth]{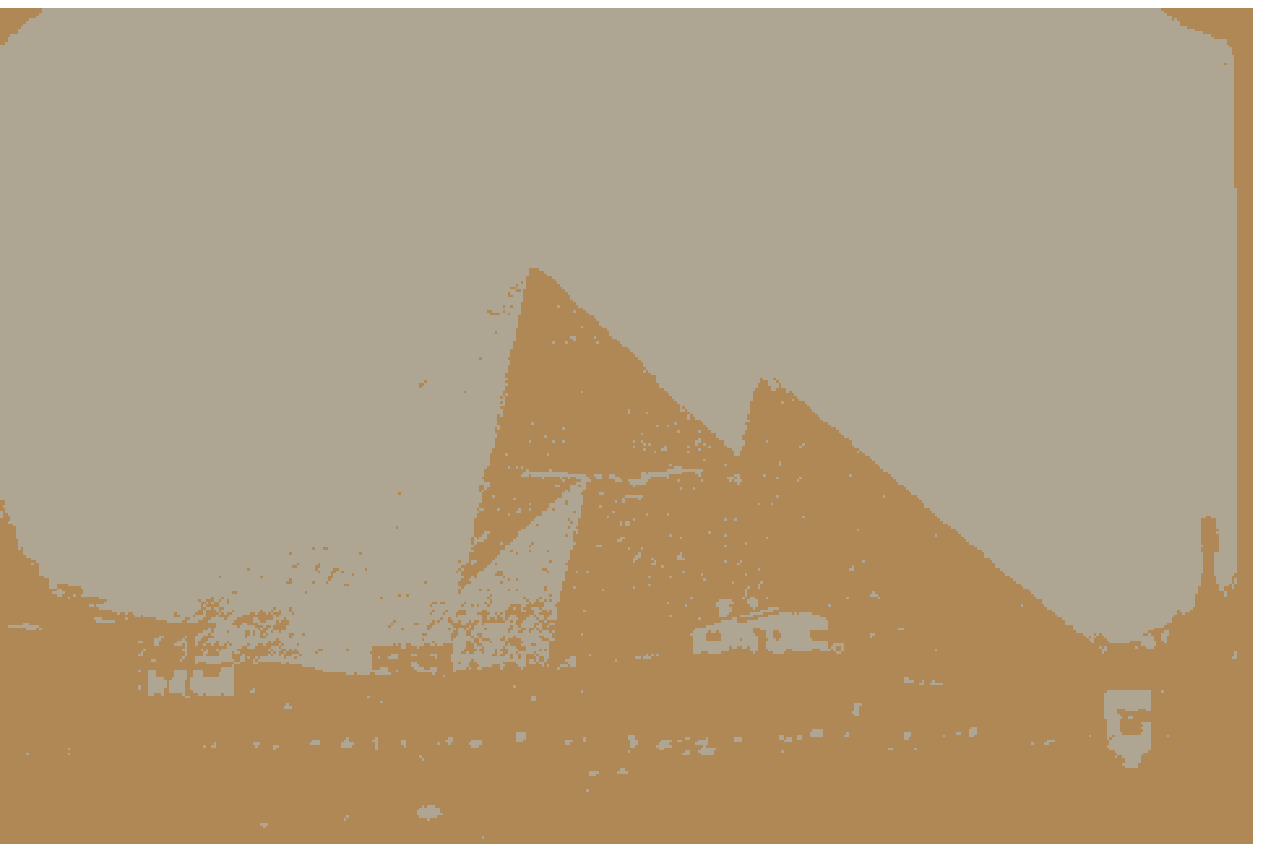}
%\centerline{Noisy}
\end{minipage}
\begin{minipage}[t]{0.19\linewidth}
\centering
\includegraphics[width=1\textwidth]{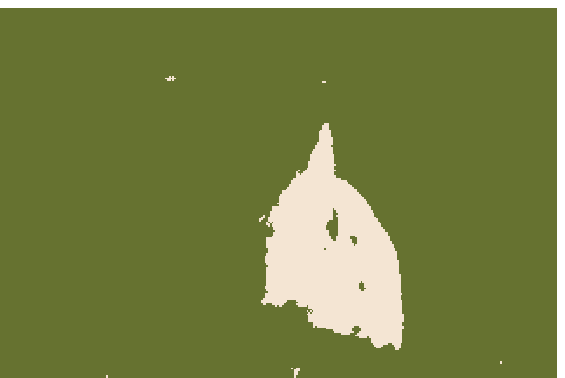}
%\centerline{Noisy}
\end{minipage}
\begin{minipage}[t]{0.19\linewidth}
\centering
\includegraphics[width=1\textwidth]{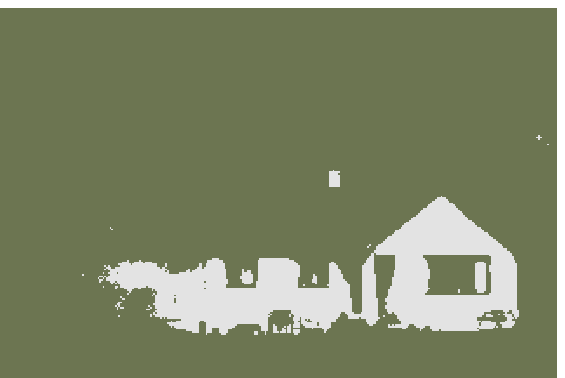}
%\centerline{Noisy}
\end{minipage}\\
\begin{minipage}[t]{0.19\linewidth}
\centering
\includegraphics[width=1\textwidth]{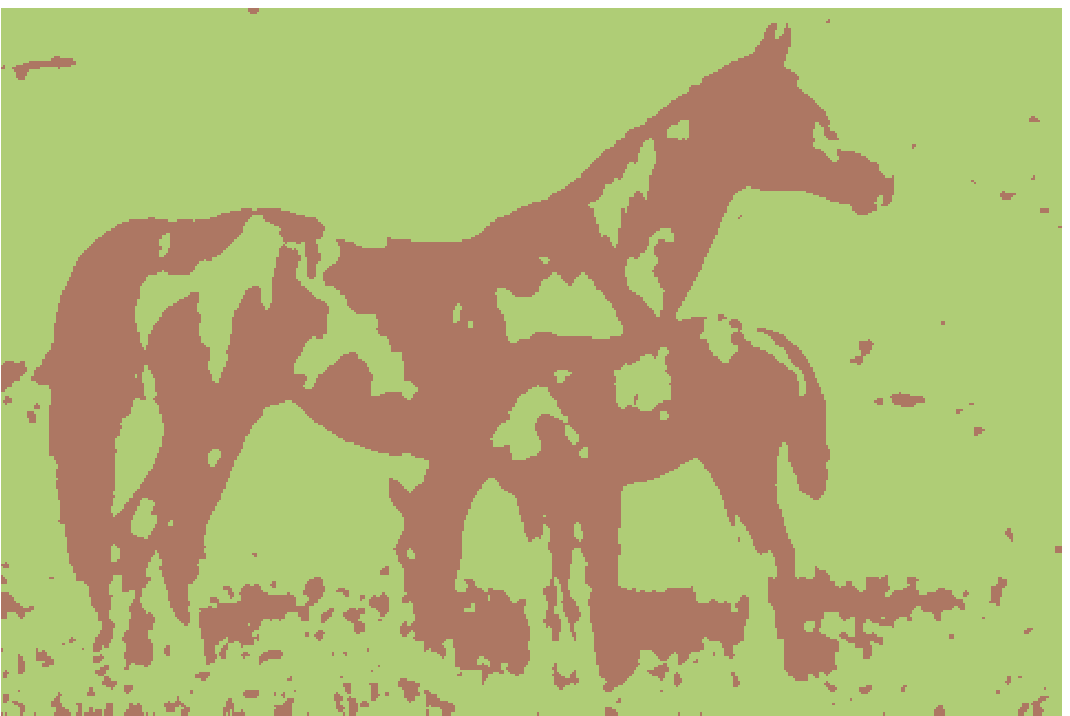}
%\centerline{Noisy}
\end{minipage}
\begin{minipage}[t]{0.19\linewidth}
\centering
\includegraphics[width=1\textwidth]{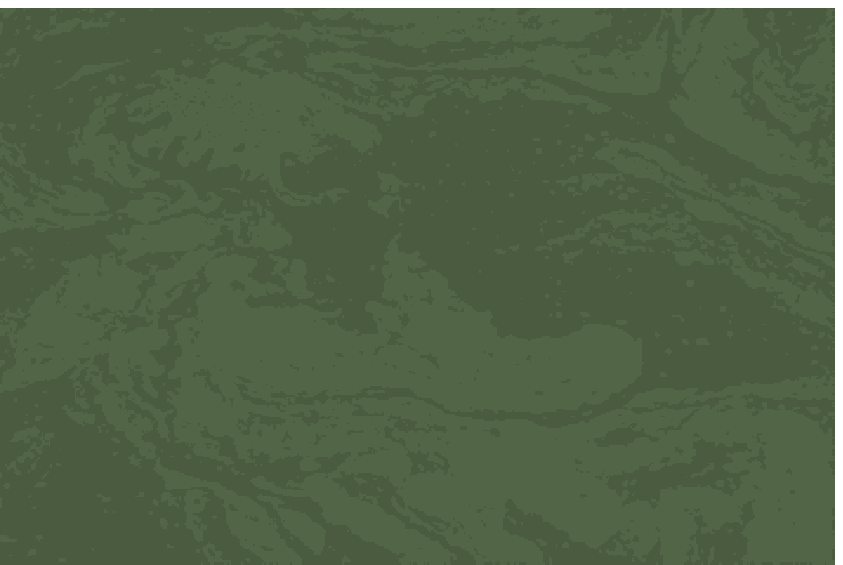}
%\centerline{Noisy}
\end{minipage}
\begin{minipage}[t]{0.19\linewidth}
\centering
\includegraphics[width=1\textwidth]{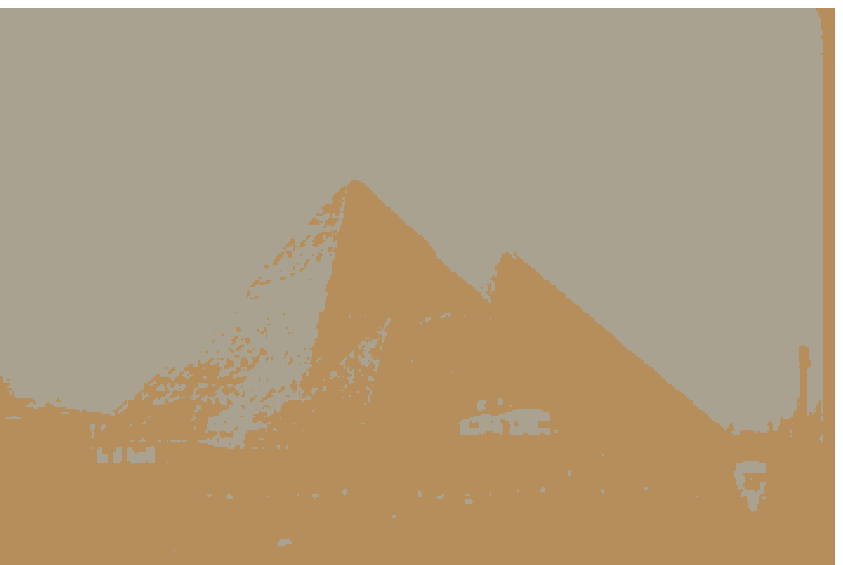}
%\centerline{Noisy}
\end{minipage}
\begin{minipage}[t]{0.19\linewidth}
\centering
\includegraphics[width=1\textwidth]{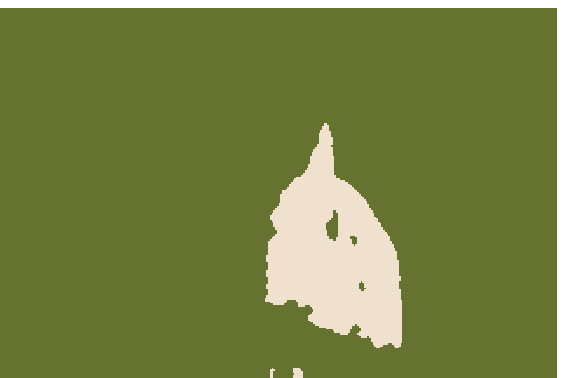}
%\centerline{Noisy}
\end{minipage}
\begin{minipage}[t]{0.19\linewidth}
\centering
\includegraphics[width=1\textwidth]{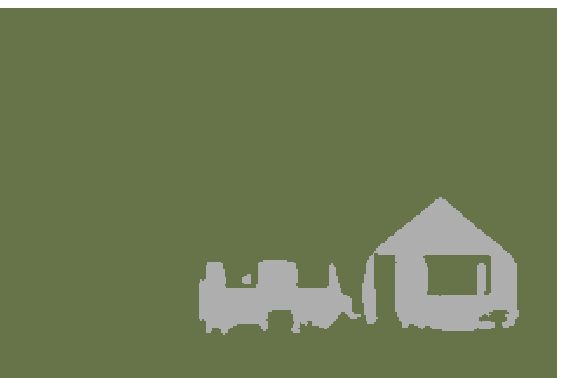}
%\centerline{Noisy}
\end{minipage}\\
\begin{minipage}[t]{0.19\linewidth}
\centering
\includegraphics[width=1\textwidth]{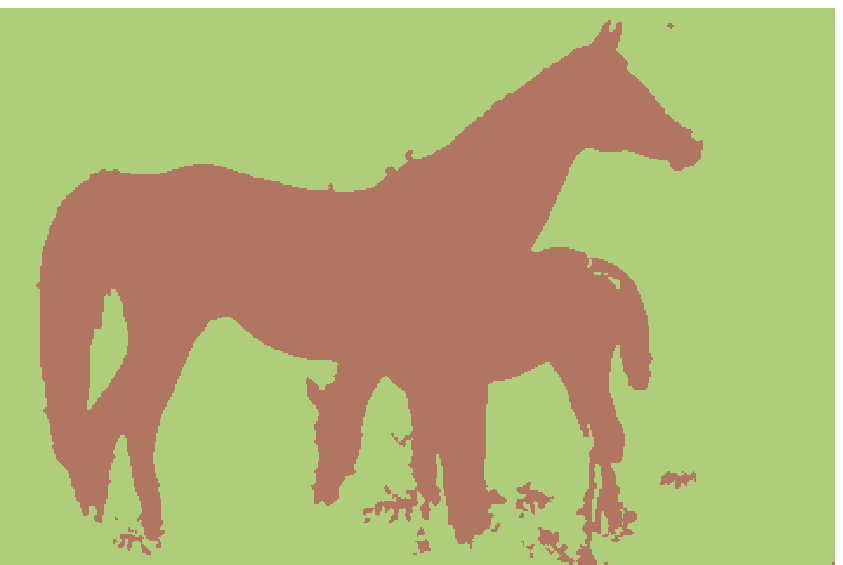}
%\centerline{Noisy}
\end{minipage}
\begin{minipage}[t]{0.19\linewidth}
\centering
\includegraphics[width=1\textwidth]{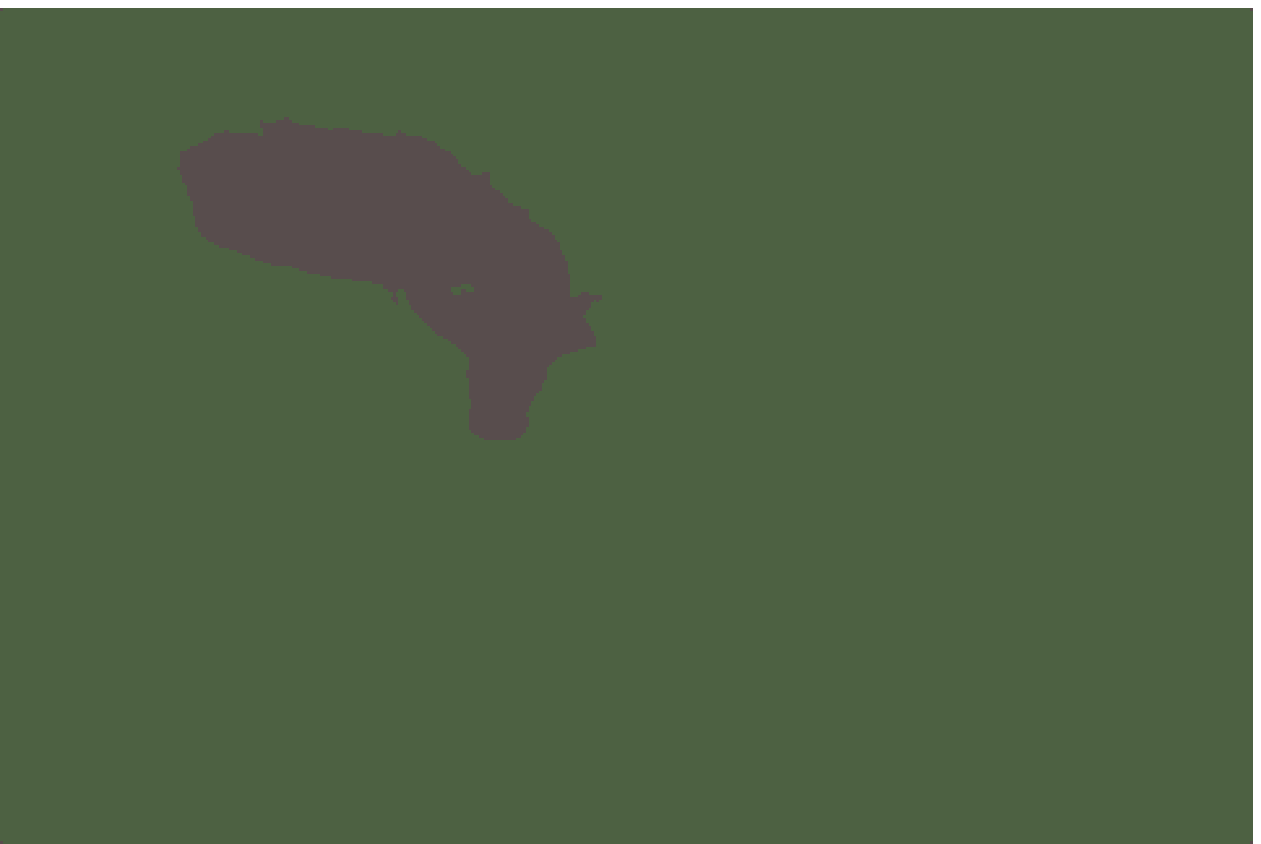}
%\centerline{Noisy}
\end{minipage}
\begin{minipage}[t]{0.19\linewidth}
\centering
\includegraphics[width=1\textwidth]{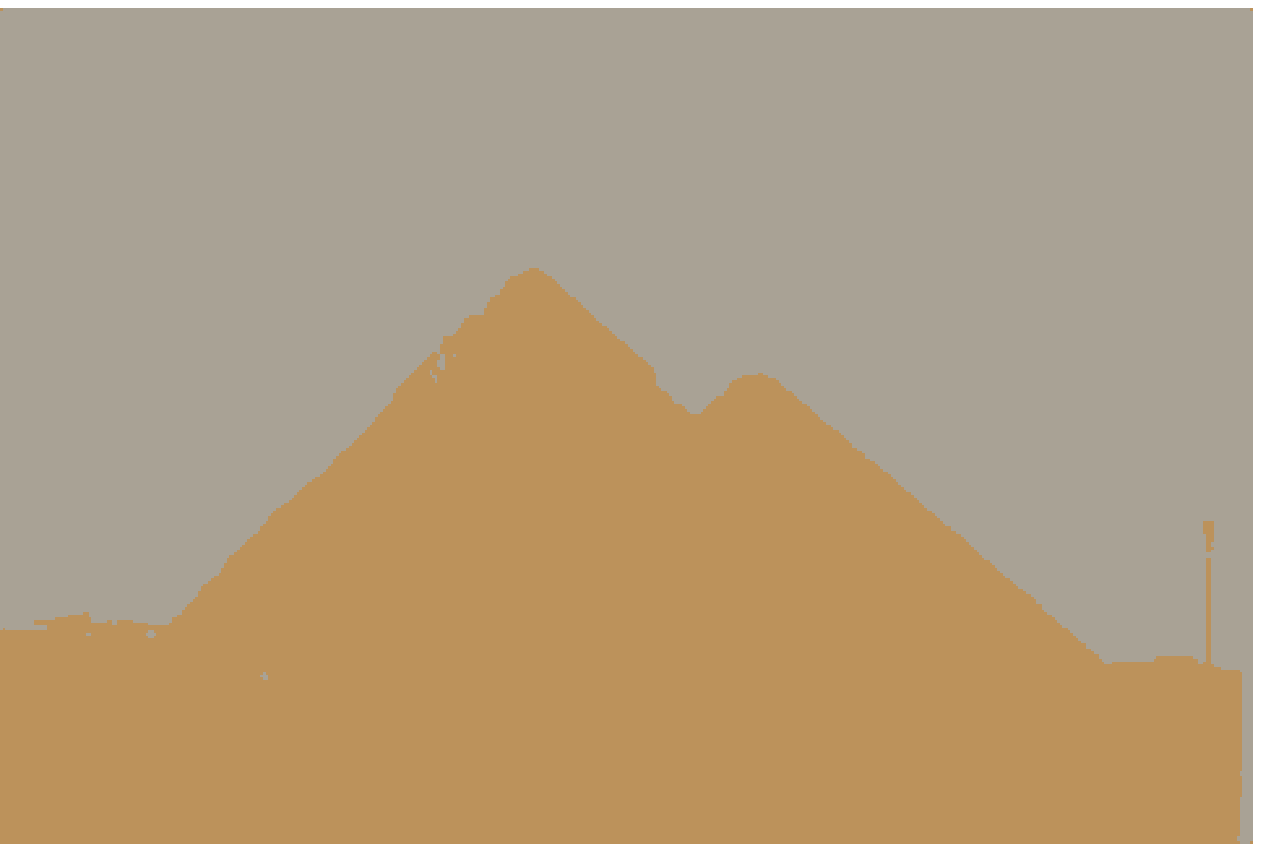}
%\centerline{Noisy}
\end{minipage}
\begin{minipage}[t]{0.19\linewidth}
\centering
\includegraphics[width=1\textwidth]{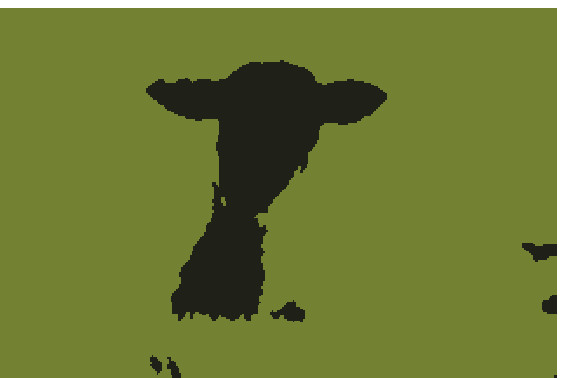}
%\centerline{Noisy}
\end{minipage}
\begin{minipage}[t]{0.19\linewidth}
\centering
\includegraphics[width=1\textwidth]{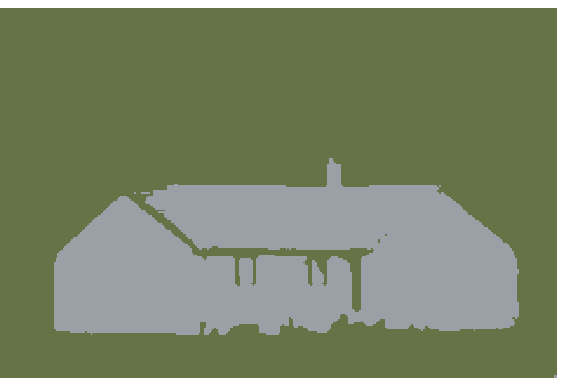}
%\centerline{Noisy}
\end{minipage}\\
\begin{minipage}[t]{0.19\linewidth}
\centering
\includegraphics[width=1\textwidth]{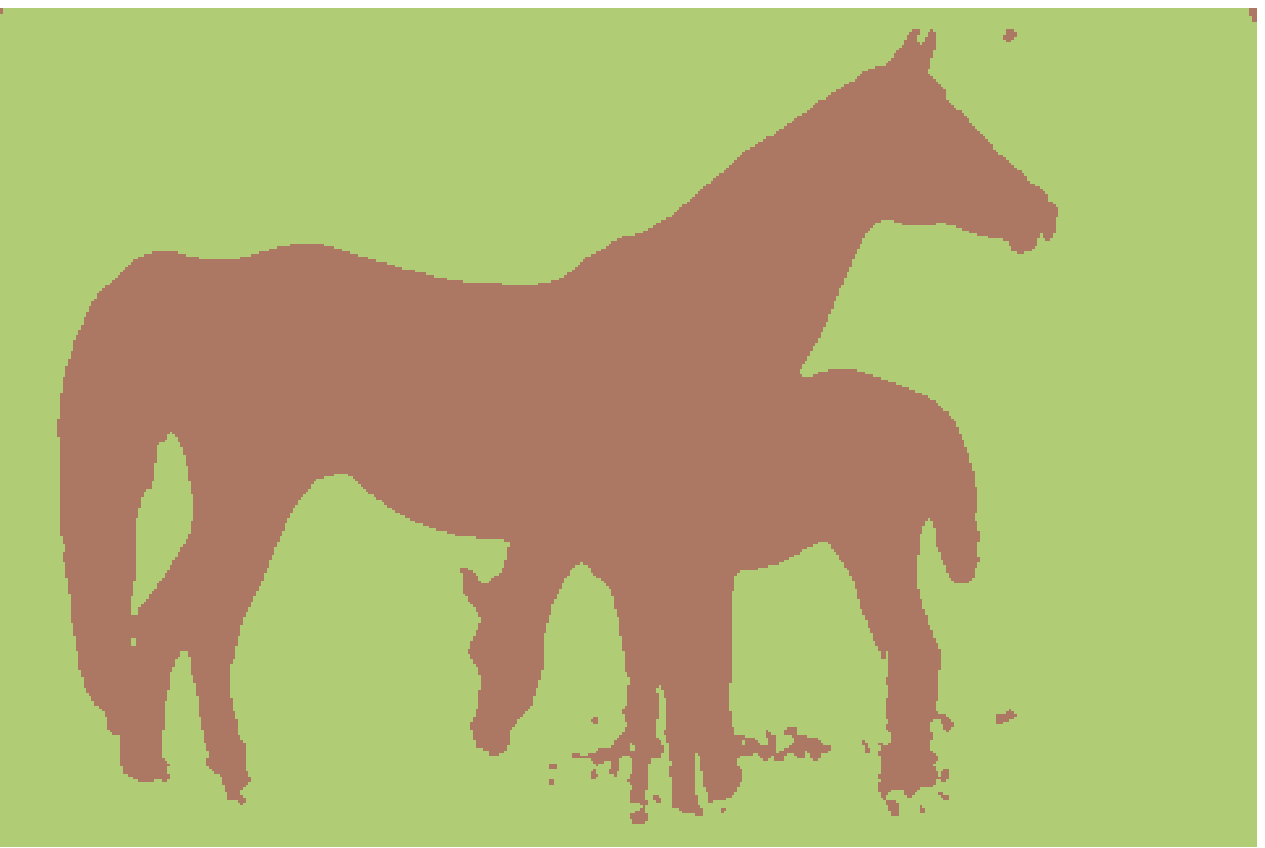}
%\centerline{Noisy}
\end{minipage}
\begin{minipage}[t]{0.19\linewidth}
\centering
\includegraphics[width=1\textwidth]{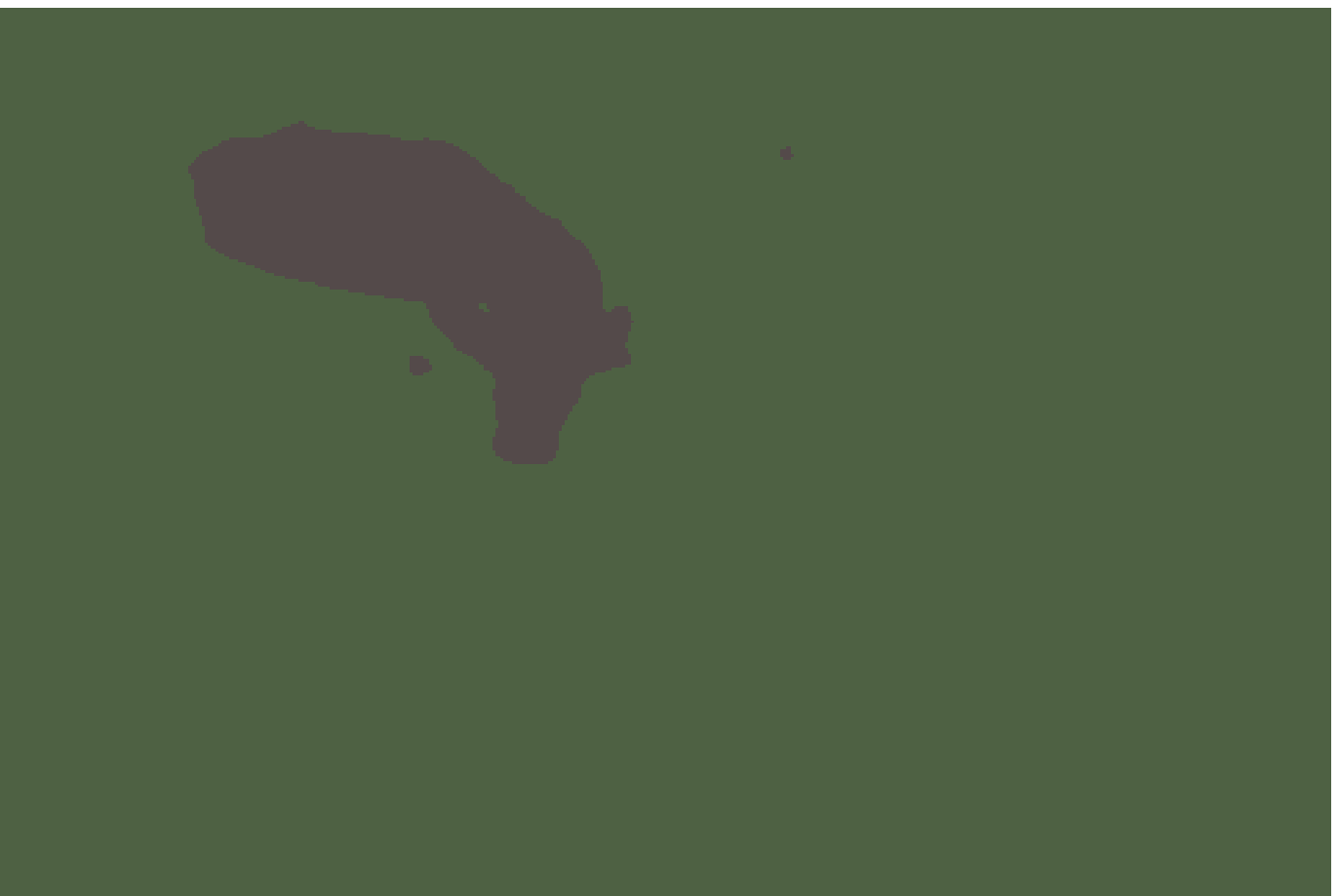}
%\centerline{Noisy}
\end{minipage}
\begin{minipage}[t]{0.19\linewidth}
\centering
\includegraphics[width=1\textwidth]{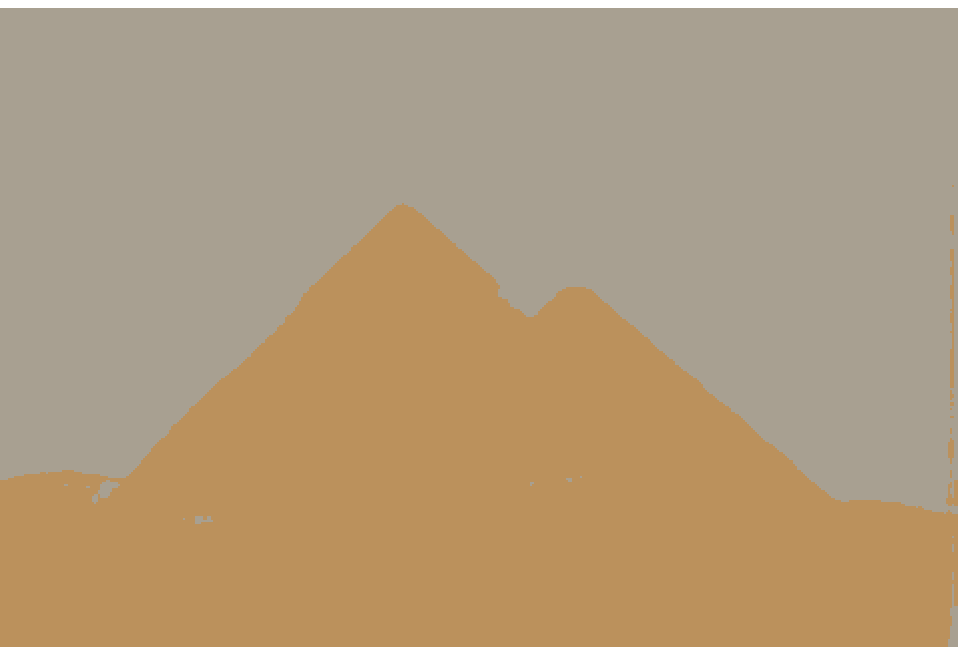}
%\centerline{Noisy}
\end{minipage}
\begin{minipage}[t]{0.19\linewidth}
\centering
\includegraphics[width=1\textwidth]{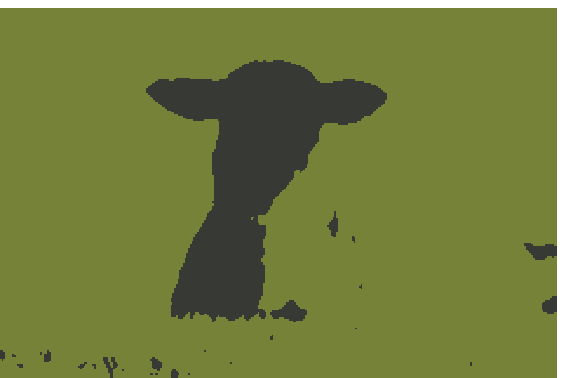}
%\centerline{Noisy}
\end{minipage}
\begin{minipage}[t]{0.19\linewidth}
\centering
\includegraphics[width=1\textwidth]{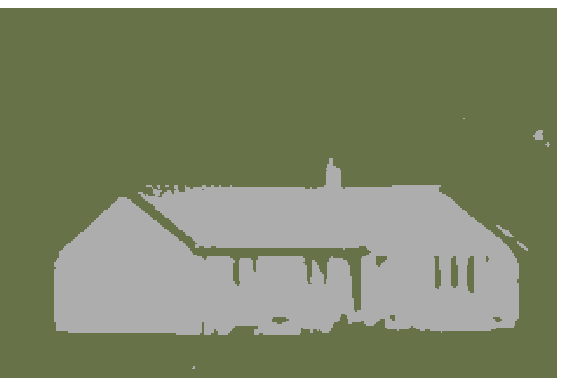}
%\centerline{Noisy}
\end{minipage}\\
\begin{minipage}[t]{0.19\linewidth}
\centering
\includegraphics[width=1\textwidth]{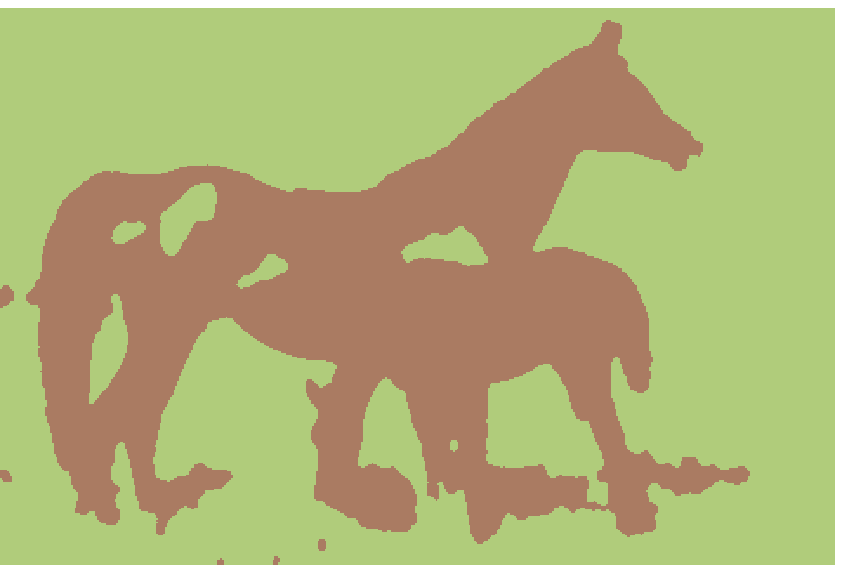}
%\centerline{Noisy}
\end{minipage}
\begin{minipage}[t]{0.19\linewidth}
\centering
\includegraphics[width=1\textwidth]{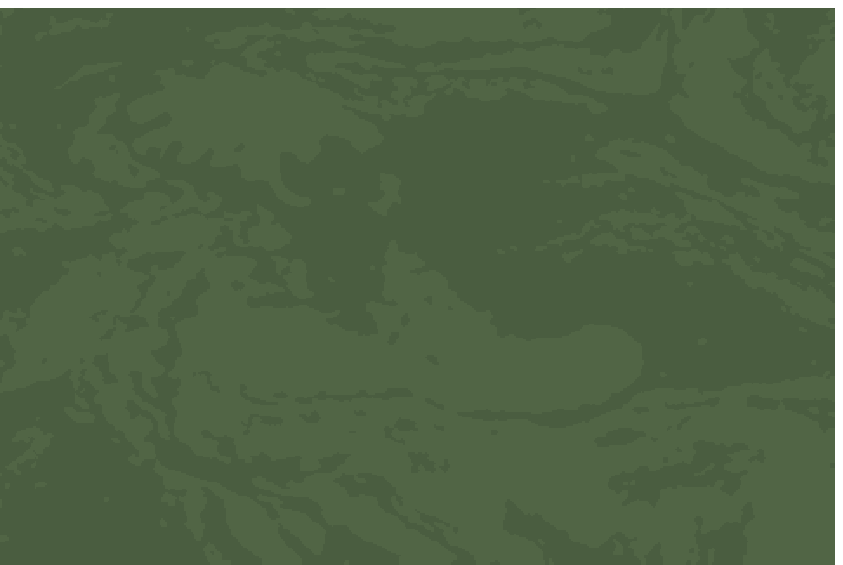}
%\centerline{Noisy}
\end{minipage}
\begin{minipage}[t]{0.19\linewidth}
\centering
\includegraphics[width=1\textwidth]{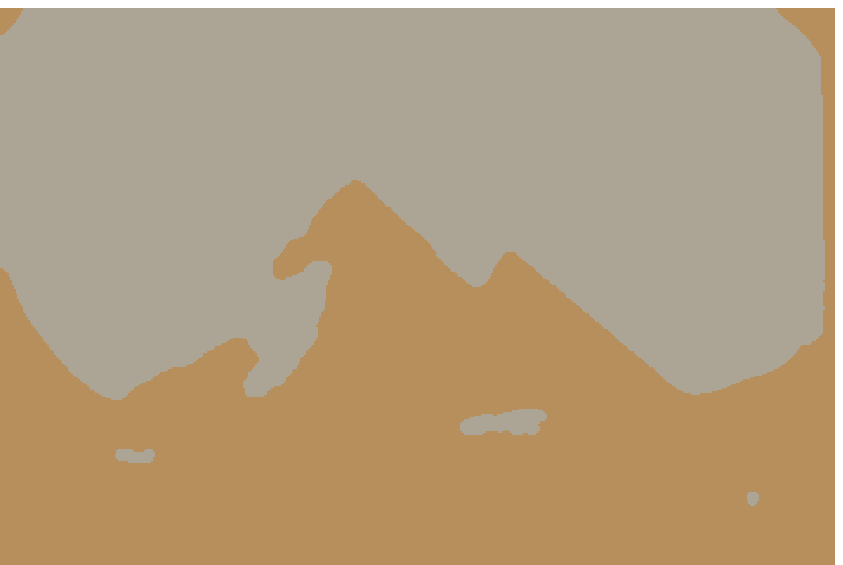}
%\centerline{Noisy}
\end{minipage}
\begin{minipage}[t]{0.19\linewidth}
\centering
\includegraphics[width=1\textwidth]{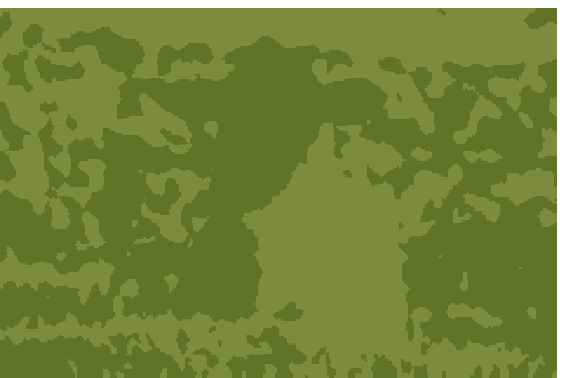}
%\centerline{Noisy}
\end{minipage}
\begin{minipage}[t]{0.19\linewidth}
\centering
\includegraphics[width=1\textwidth]{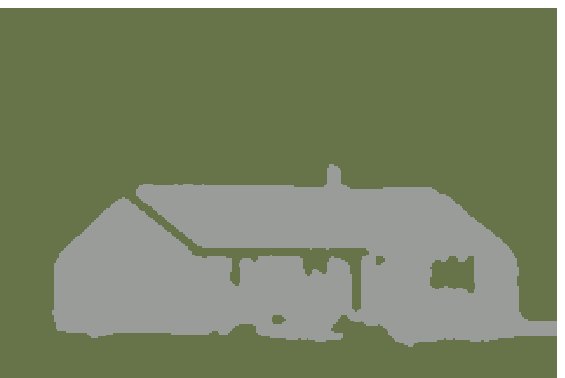}
%\centerline{Noisy}
\end{minipage}\\
\begin{minipage}[t]{0.19\linewidth}
\centering
\includegraphics[width=1\textwidth]{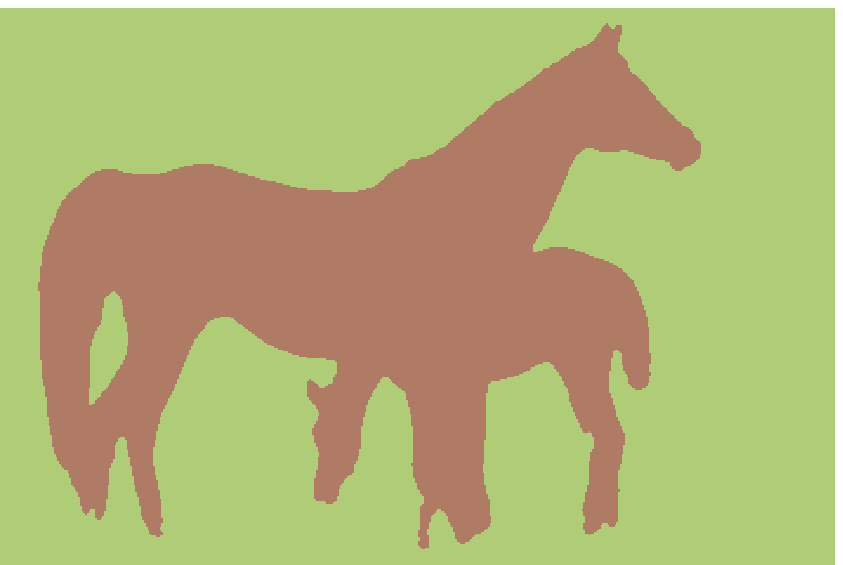}
%\centerline{Noisy}
\end{minipage}
\begin{minipage}[t]{0.19\linewidth}
\centering
\includegraphics[width=1\textwidth]{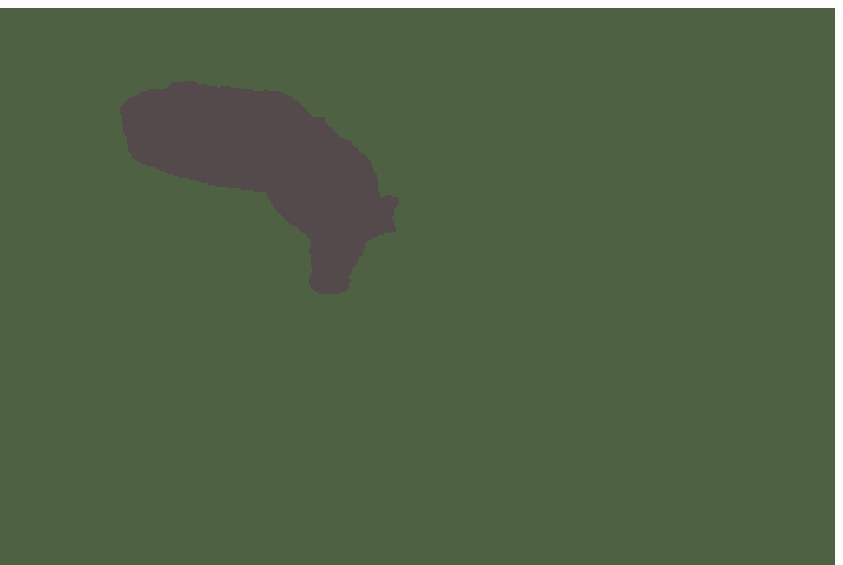}
%\centerline{Noisy}
\end{minipage}
\begin{minipage}[t]{0.19\linewidth}
\centering
\includegraphics[width=1\textwidth]{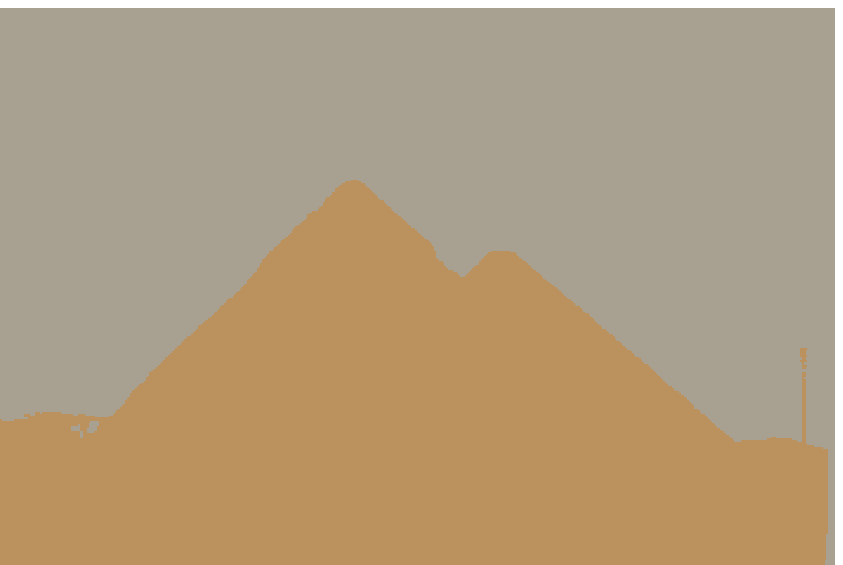}
%\centerline{Noisy}
\end{minipage}
\begin{minipage}[t]{0.19\linewidth}
\centering
\includegraphics[width=1\textwidth]{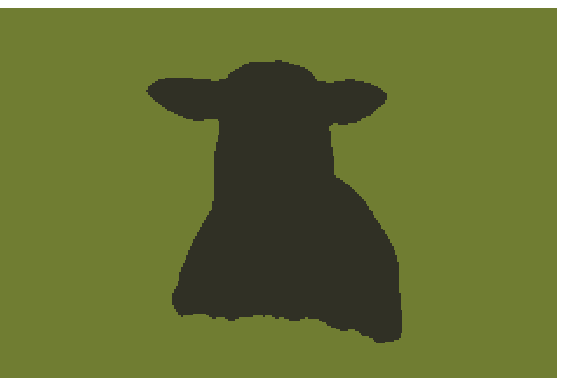}
%\centerline{Noisy}
\end{minipage}
\begin{minipage}[t]{0.19\linewidth}
\centering
\includegraphics[width=1\textwidth]{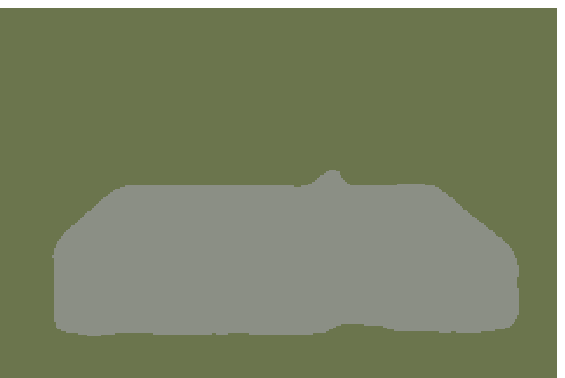}
%\centerline{Noisy}
\end{minipage}
\caption{Segmentation results on five real-world images in BSDS and MSRC. The parameters: $\phi_{1} =6.05$, $\phi_{2} =10.00$, $\phi_{3} =9.89$, $\phi_{4} =9.98$, and $\phi_{5} =9.50$. From top to bottom: observed images, ground truth, and results of FCM\_S1, FCM\_S2, FLICM, KWFLICM, FRFCM, WFCM, DSFCM\_N, and WRFCM.}
\label{fig:BM}
\end{figure}

Fig. \ref{fig:BM} visually shows the comparison between \mbox{WRFCM} and seven peers while Table \ref{tab:BM} gives the quantitative comparison. Apparently, WRFCM achieves better segmentation results than its peers. FCM\_S1, FCM\_S2, FLICM, KWFLICM and DSFCM\_N obtain unsatisfactory results on all tested images. Superior to them, FRFCM and WFCM preserve more contours and feature details. From a quantitative point of view, WRFCM acquires optimal SA, SDS, and MCC values much more than its peers. Note that it merely gets a slightly smaller SDS value than FRFCM and WFCM for the first and second images, respectively.

The second set contains images from NEO. Here, we select two typical images. Each of them represents an example for a specific scene. We produce the ground truth of each scene by randomly shooting it for 50 times within the time span 2000--2019. The visual results of all algorithms are shown in Figs. \ref{fig:NEO1} and \ref{fig:NEO2}. The corresponding SA, SDS, and MCC values are given in Table \ref{tab:NEO}.
\begin{figure}[htb]
\centering
\begin{minipage}[t]{0.32\linewidth}
\centering
\includegraphics[width=1\textwidth]{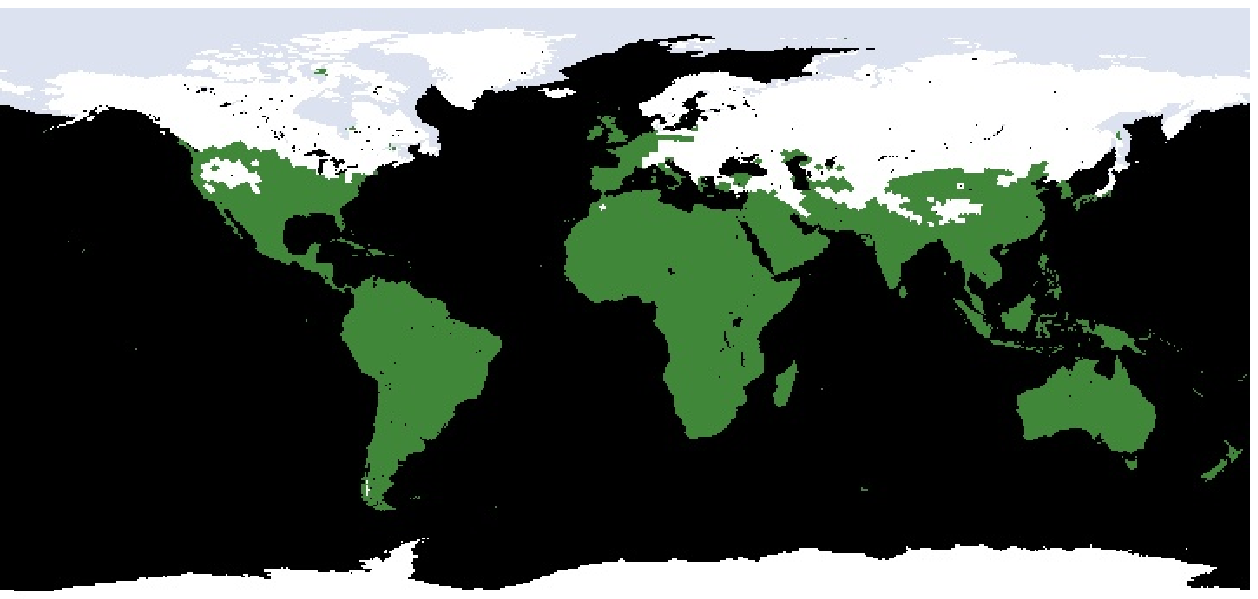}
\centerline{\footnotesize(a)}
\end{minipage}
\begin{minipage}[t]{0.32\linewidth}
\centering
\includegraphics[width=1\textwidth]{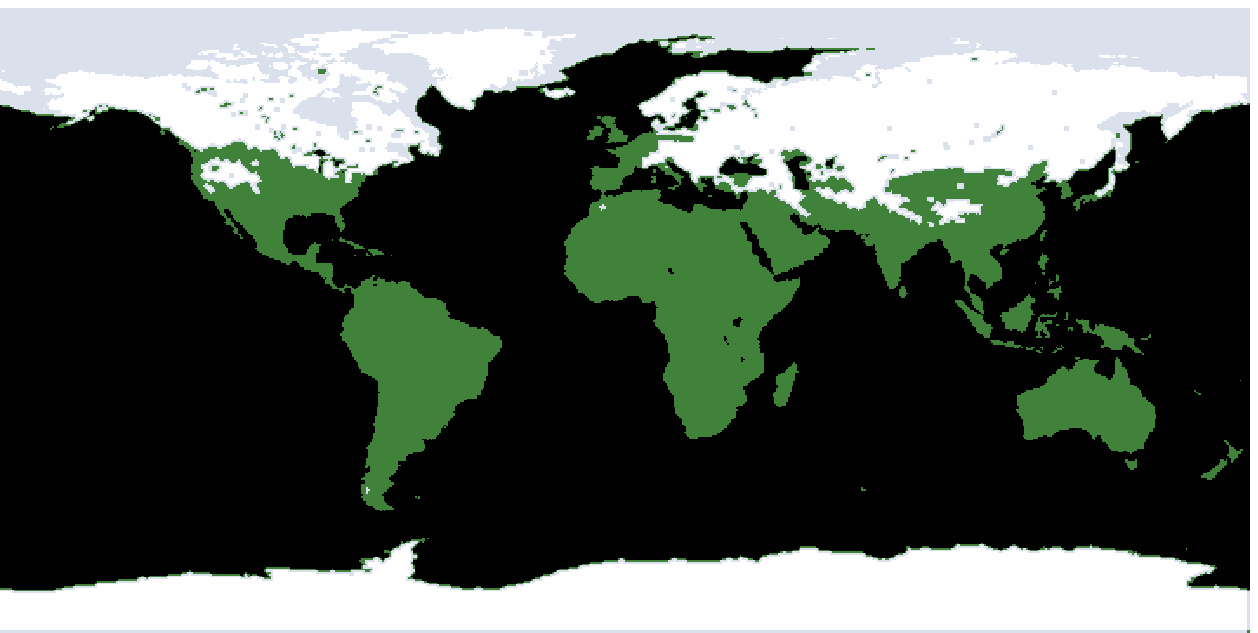}
\centerline{\footnotesize(b)}
\end{minipage}
\begin{minipage}[t]{0.32\linewidth}
\centering
\includegraphics[width=1\textwidth]{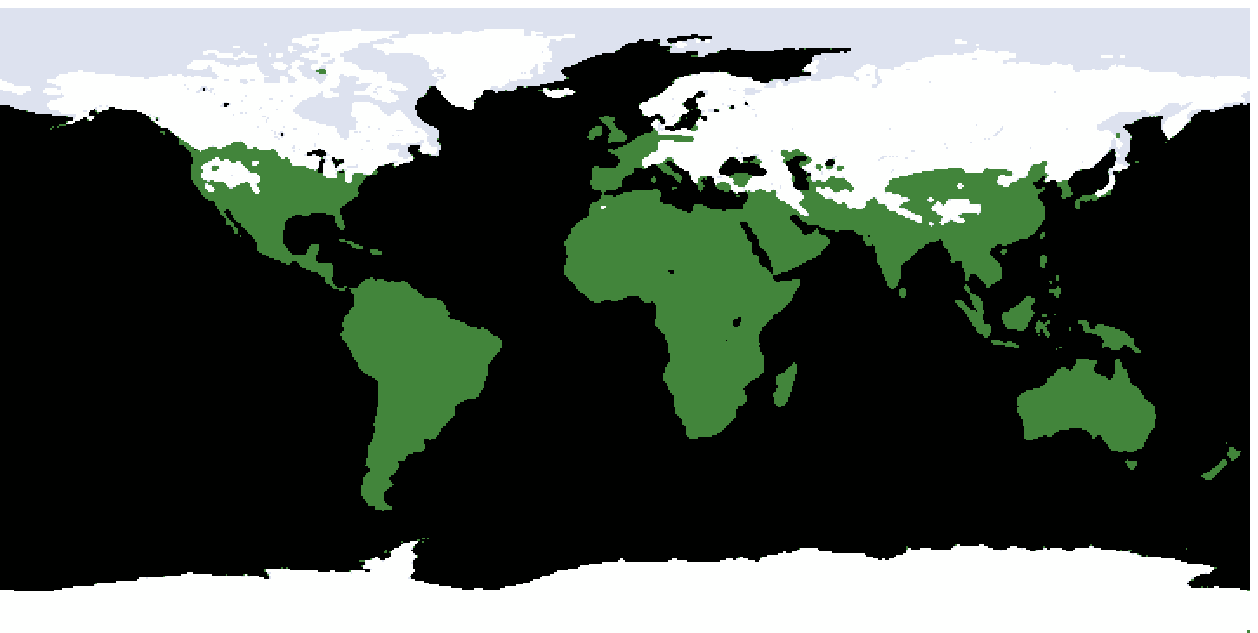}
\centerline{\footnotesize(c)}
\end{minipage}
\begin{minipage}[t]{0.32\linewidth}
\centering
\includegraphics[width=1\textwidth]{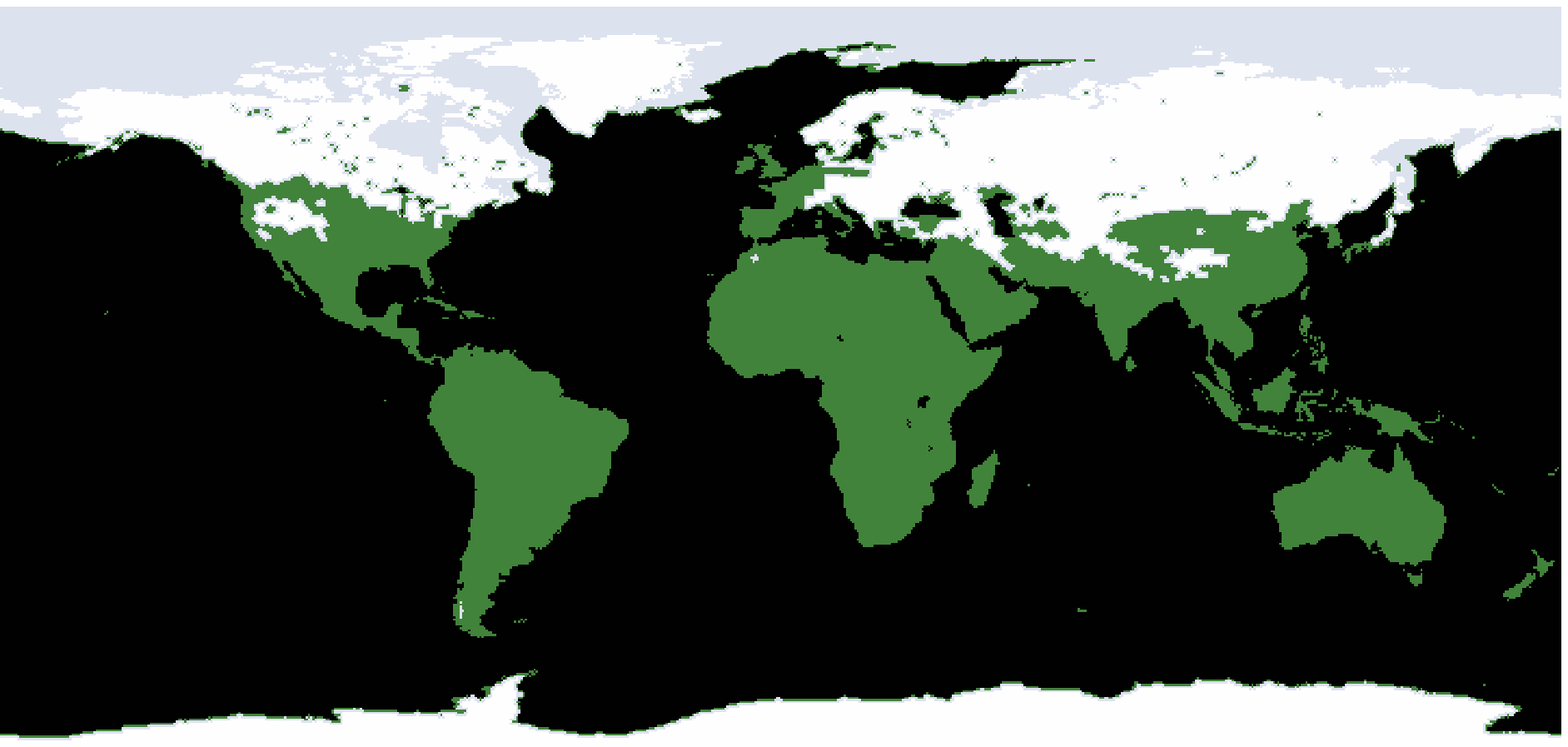}
\centerline{\footnotesize(d)}
\end{minipage}
\begin{minipage}[t]{0.32\linewidth}
\centering
\includegraphics[width=1\textwidth]{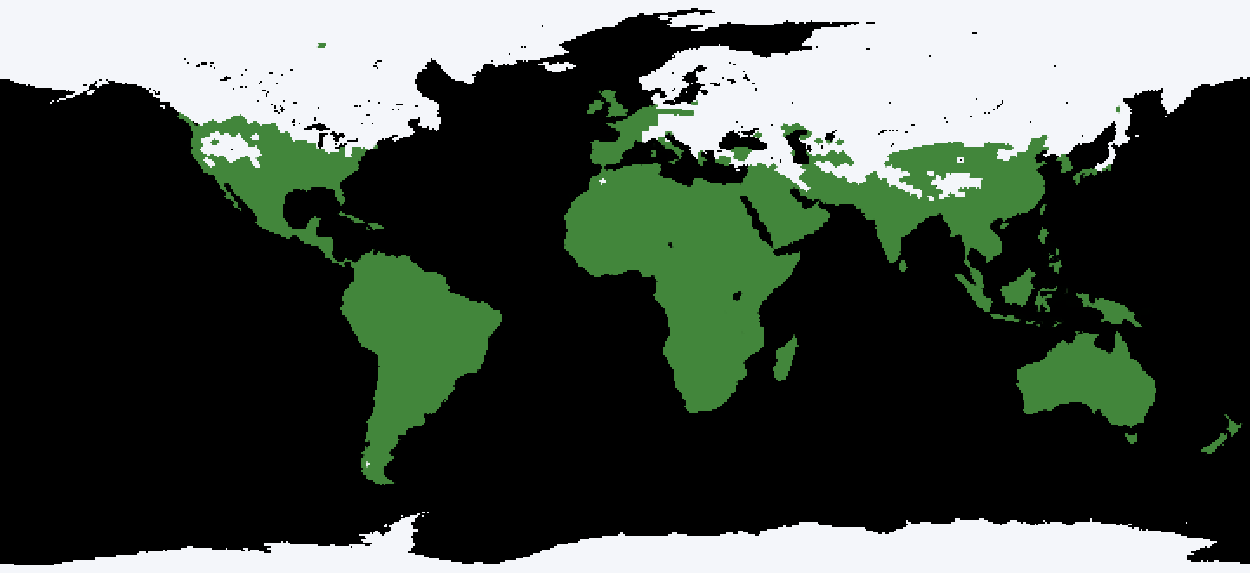}
\centerline{\footnotesize(e)}
\end{minipage}
\begin{minipage}[t]{0.32\linewidth}
\centering
\includegraphics[width=1\textwidth]{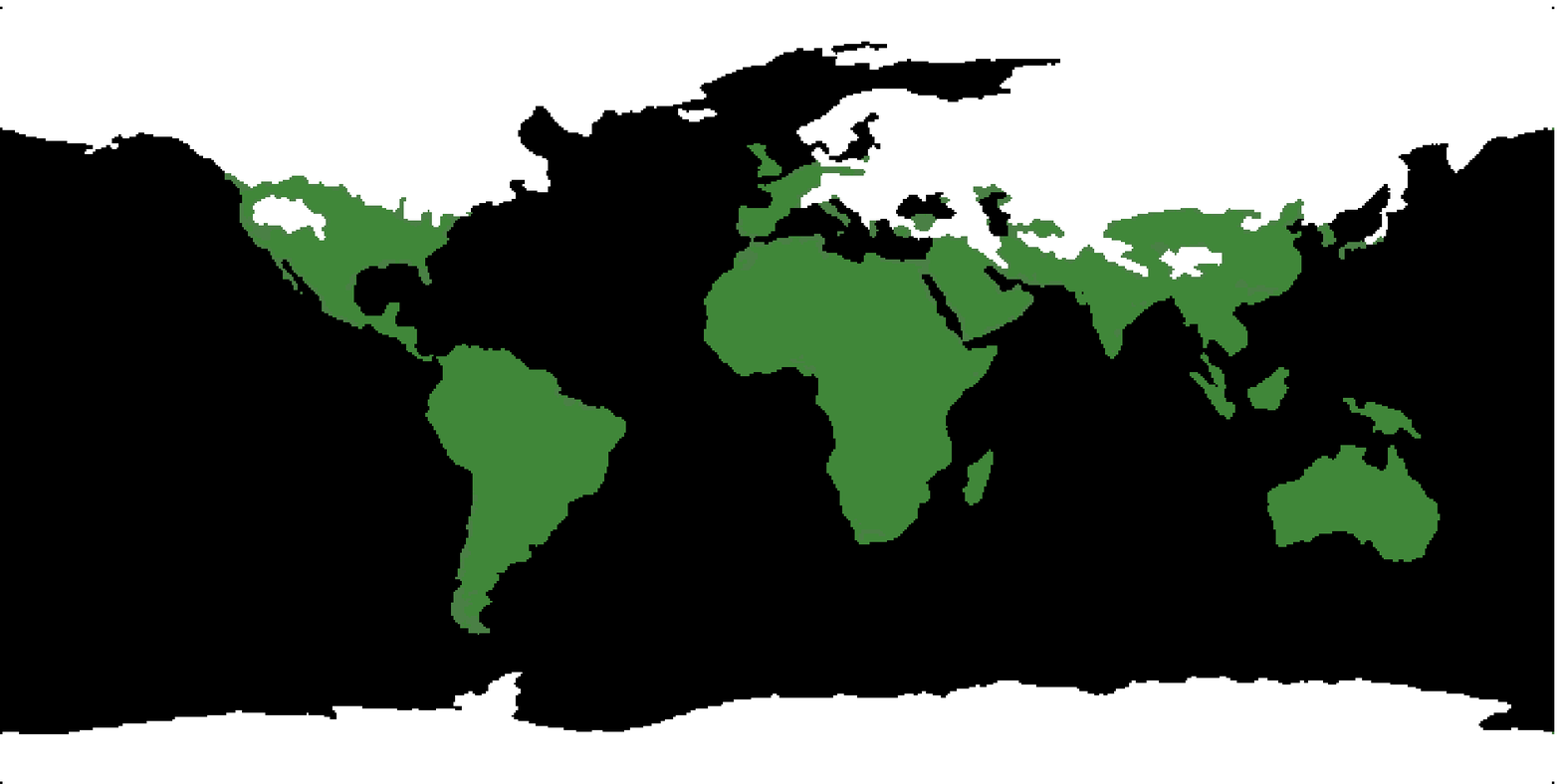}
\centerline{\footnotesize(f)}
\end{minipage}
\begin{minipage}[t]{0.32\linewidth}
\centering
\includegraphics[width=1\textwidth]{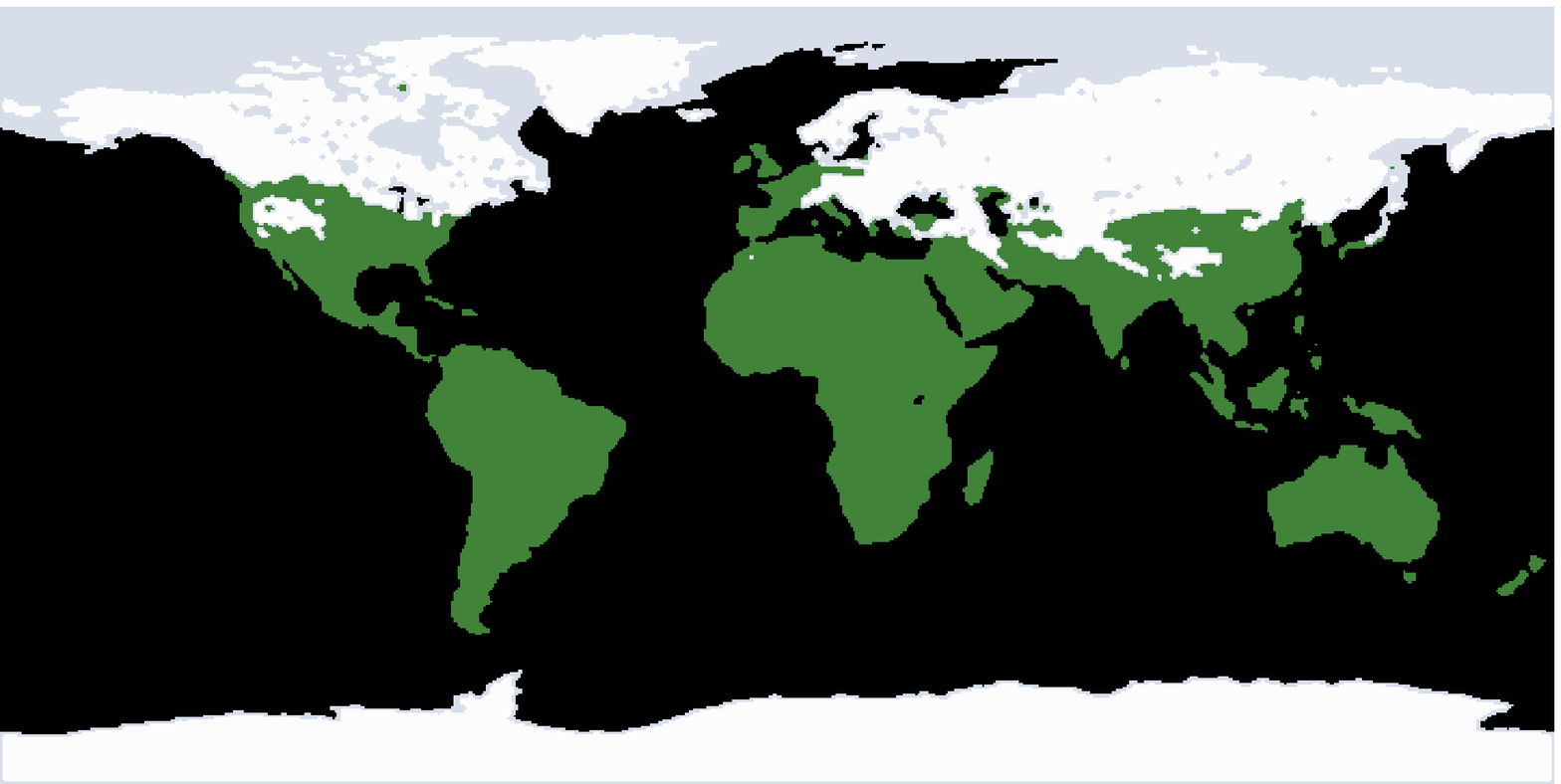}
\centerline{\footnotesize(g)}
\end{minipage}
\begin{minipage}[t]{0.32\linewidth}
\centering
\includegraphics[width=1\textwidth]{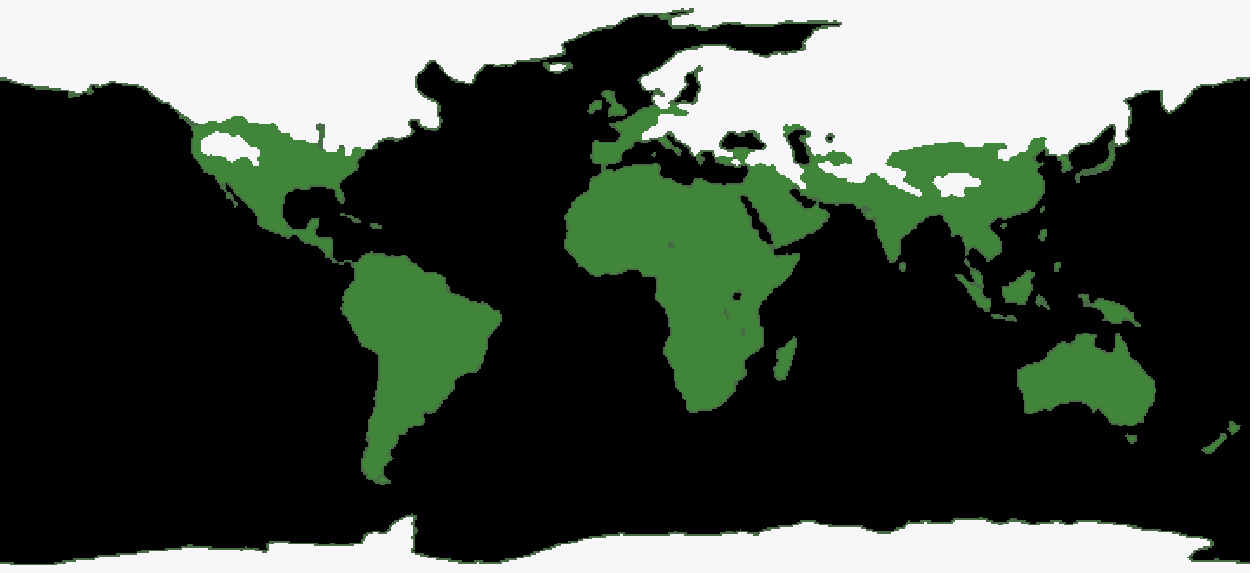}
\centerline{\footnotesize(h)}
\end{minipage}
\begin{minipage}[t]{0.32\linewidth}
\centering
\includegraphics[width=1\textwidth]{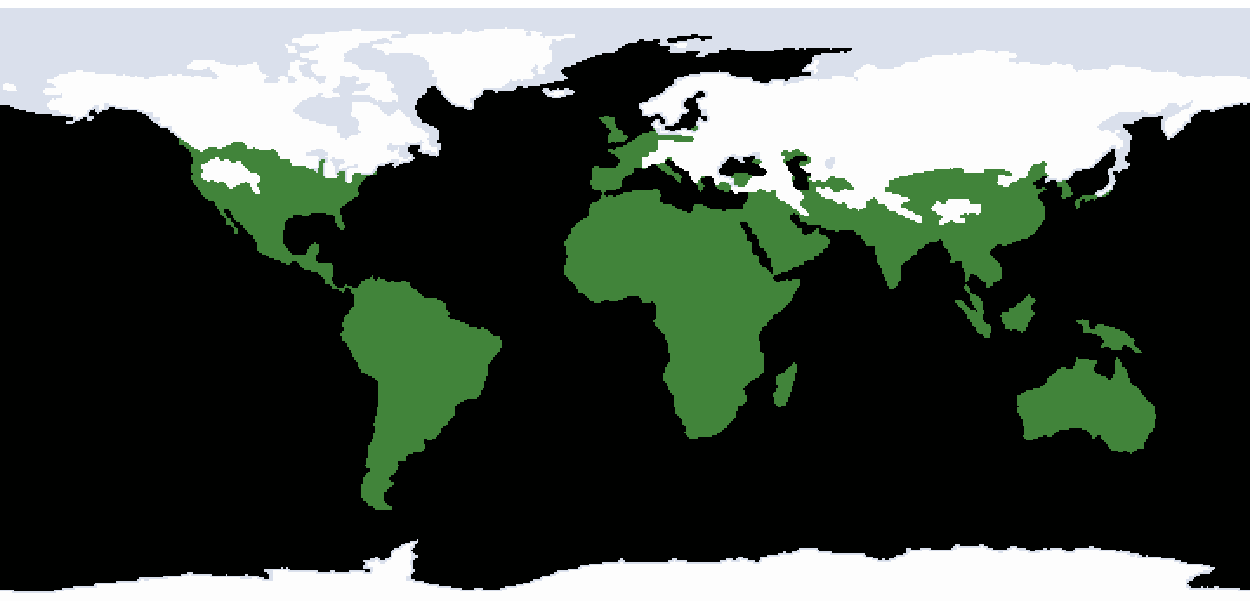}
\centerline{\footnotesize(i)}
\end{minipage}
\caption{Segmentation results on the first real-world image in NEO. The parameter: $\phi =6.10$. From (a) to (i): observed image and results of FCM\_S1, FCM\_S2, FLICM, KWFLICM, FRFCM, WFCM, DSFCM\_N, and WRFCM.}
\label{fig:NEO1}%
\end{figure}

\begin{figure}[htb]
\centering
\begin{minipage}[t]{0.32\linewidth}
\centering
\includegraphics[width=1\textwidth]{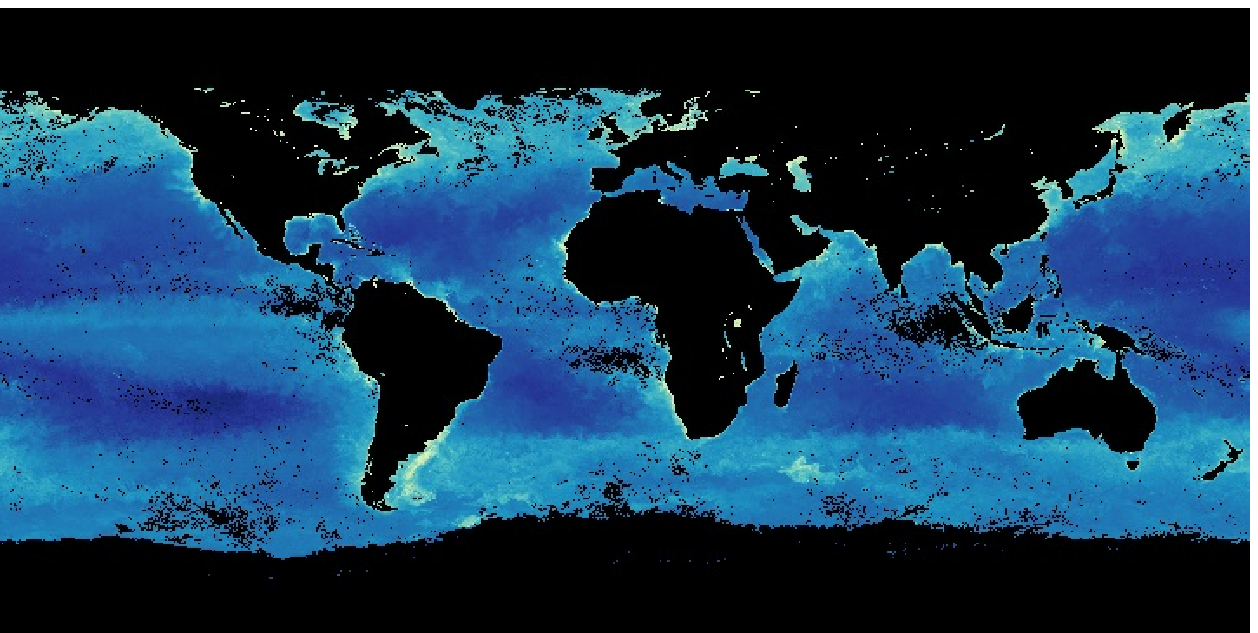}
\centerline{\footnotesize(a)}
\end{minipage}
\begin{minipage}[t]{0.32\linewidth}
\centering
\includegraphics[width=1\textwidth]{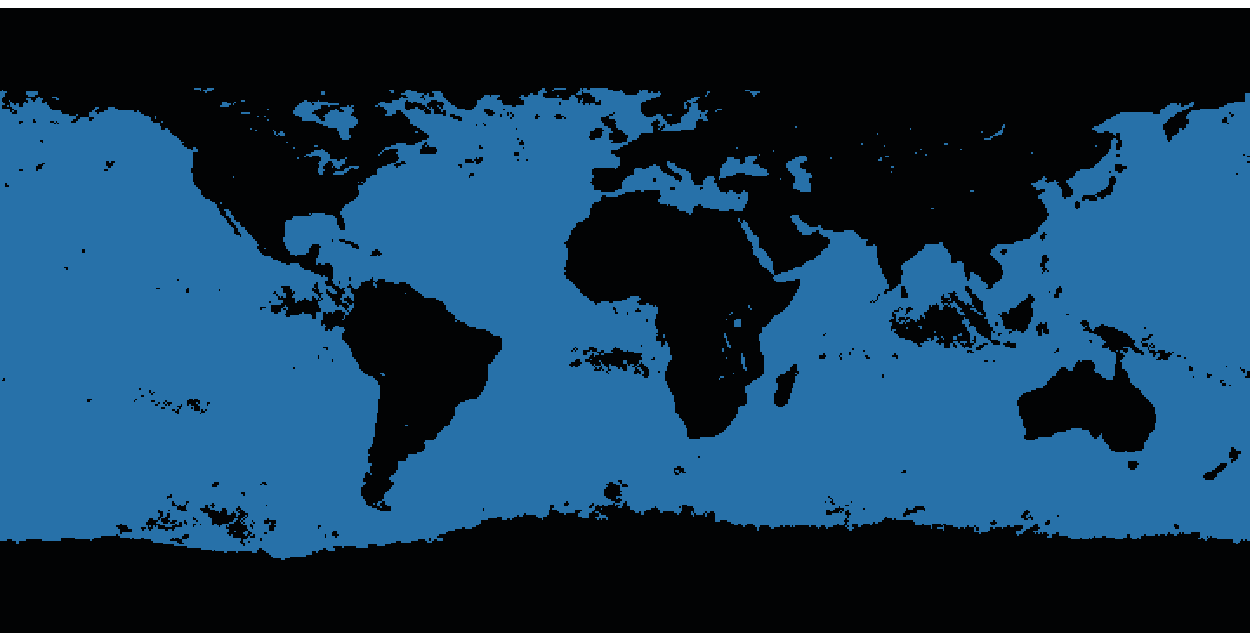}
\centerline{\footnotesize(b)}
\end{minipage}
\begin{minipage}[t]{0.32\linewidth}
\centering
\includegraphics[width=1\textwidth]{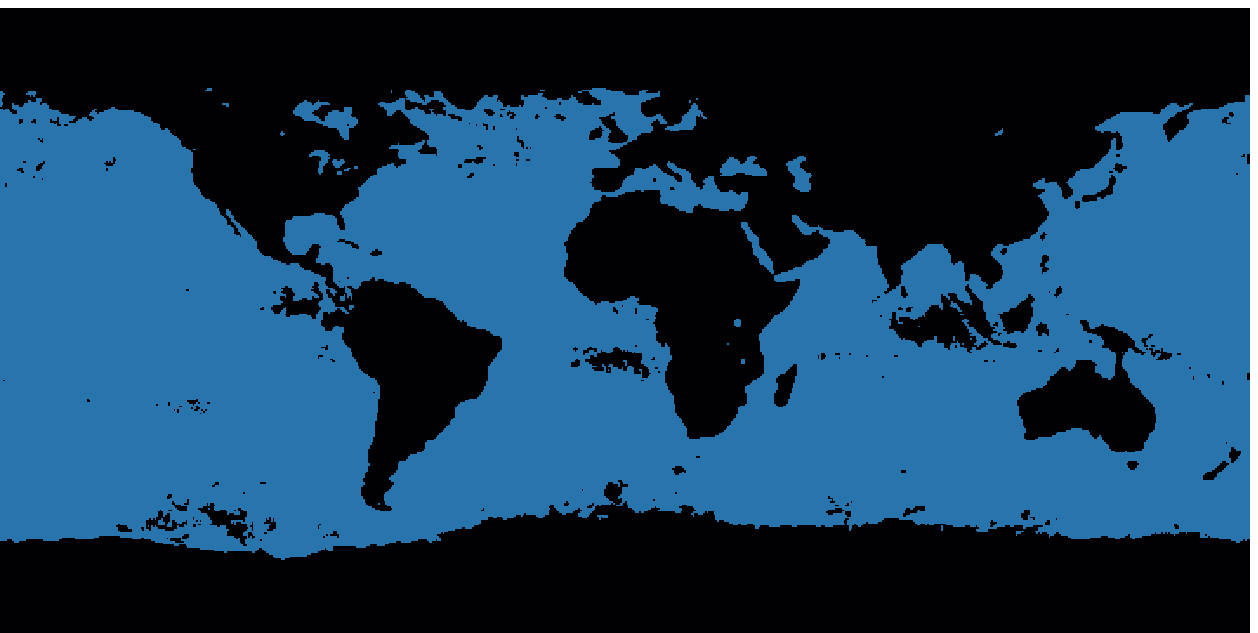}
\centerline{\footnotesize(c)}
\end{minipage}
\begin{minipage}[t]{0.32\linewidth}
\centering
\includegraphics[width=1\textwidth]{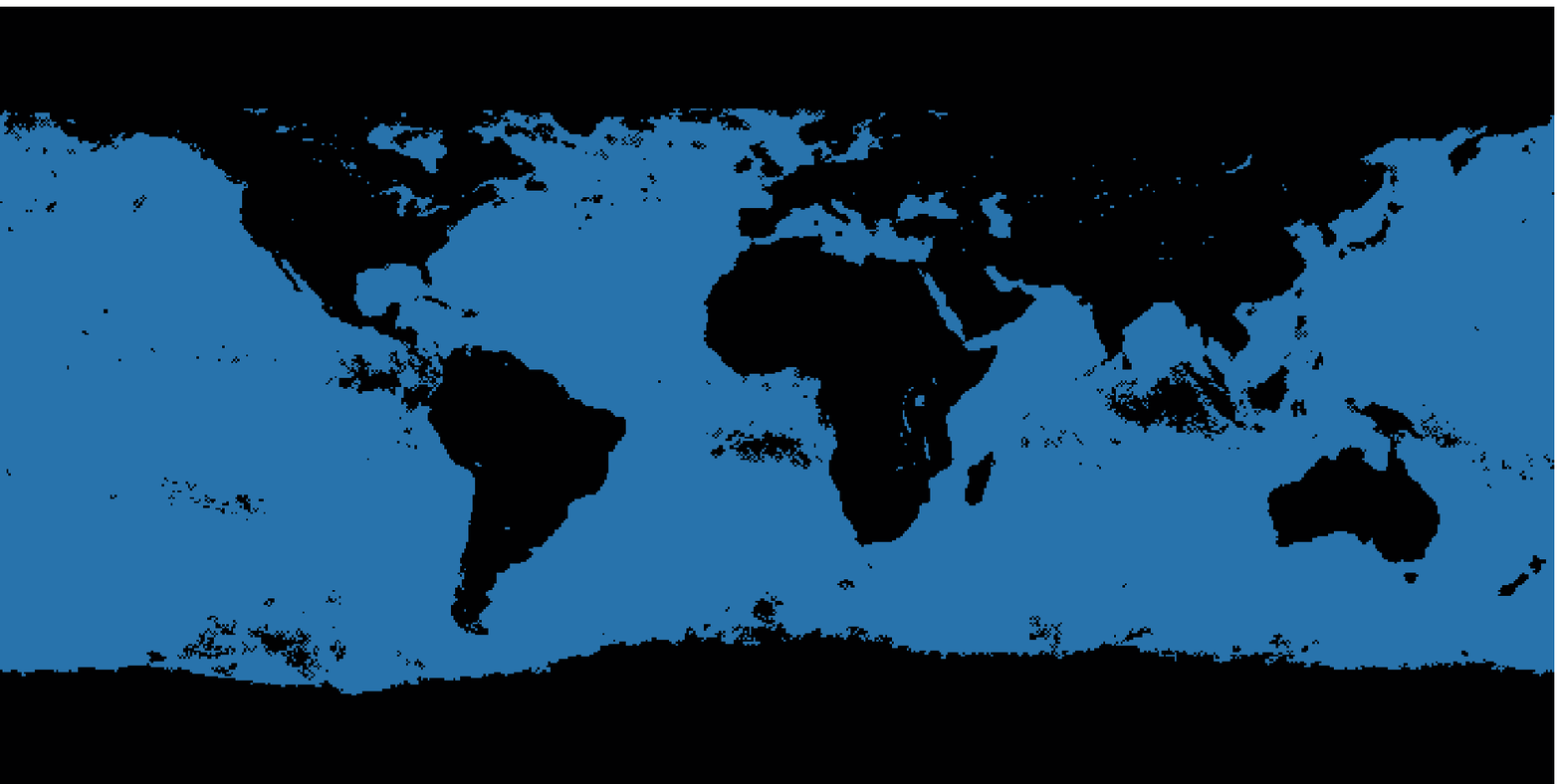}
\centerline{\footnotesize(d)}
\end{minipage}
\begin{minipage}[t]{0.32\linewidth}
\centering
\includegraphics[width=1\textwidth]{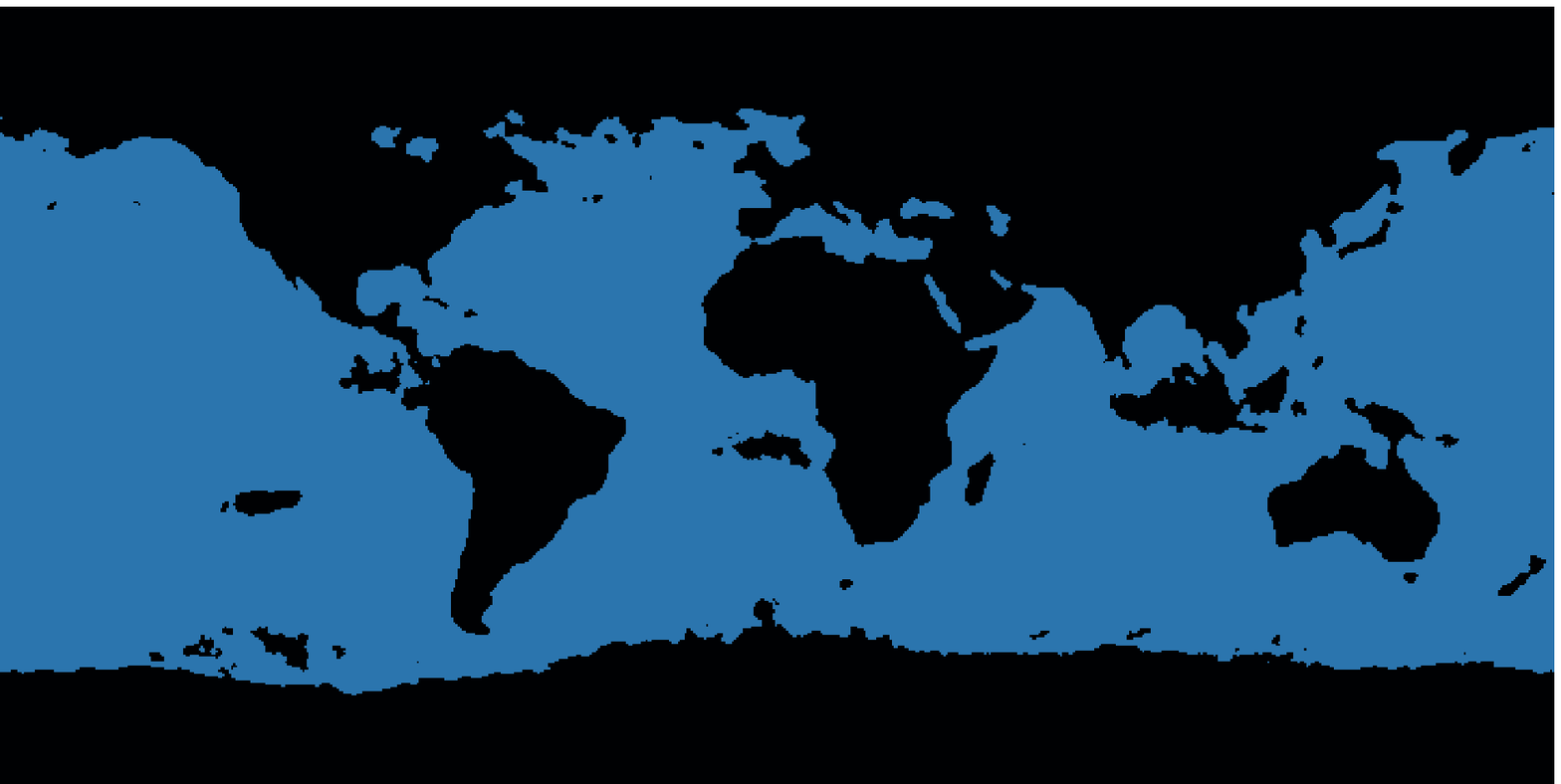}
\centerline{\footnotesize(e)}
\end{minipage}
\begin{minipage}[t]{0.32\linewidth}
\centering
\includegraphics[width=1\textwidth]{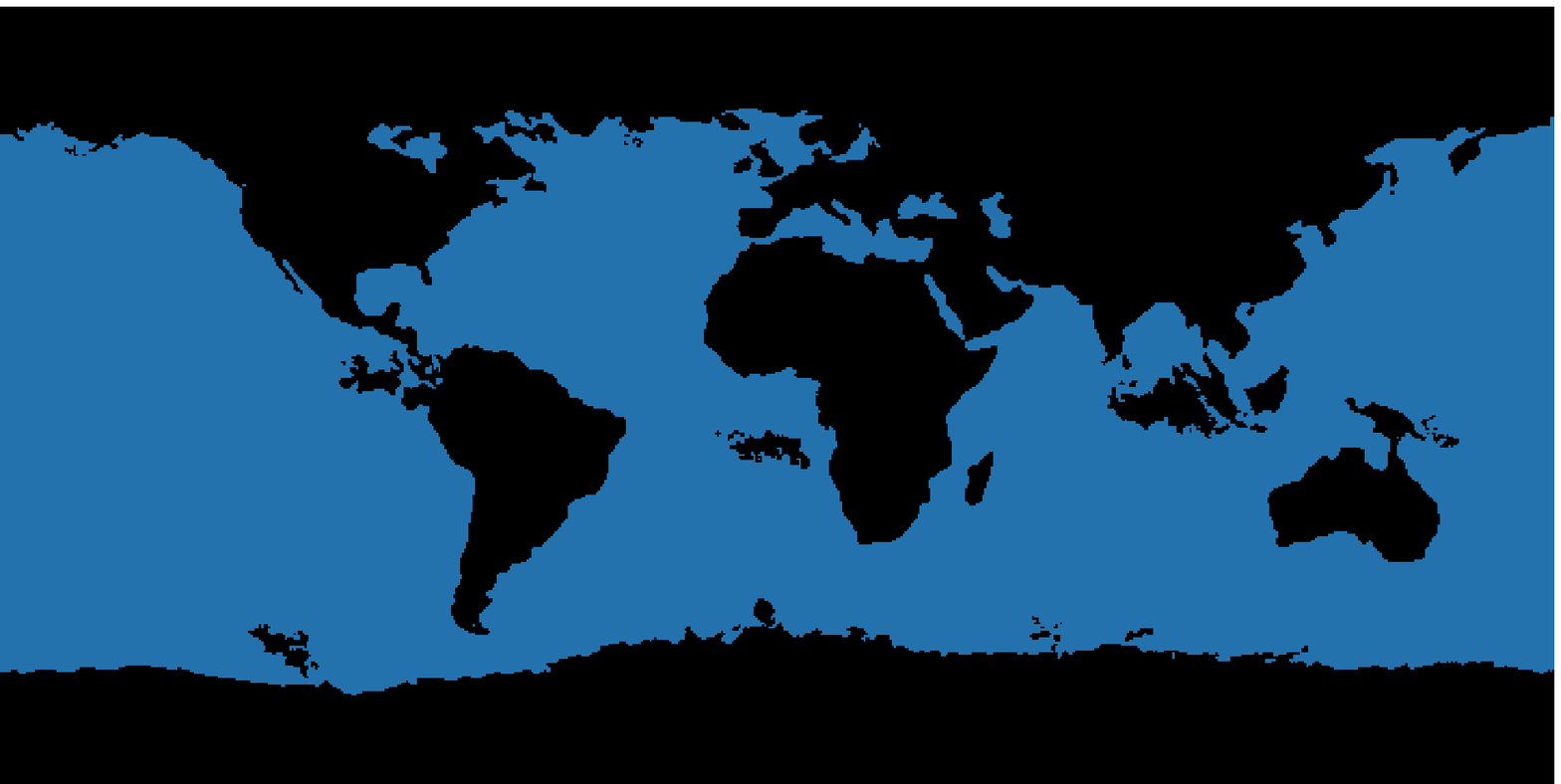}
\centerline{\footnotesize(f)}
\end{minipage}
\begin{minipage}[t]{0.32\linewidth}
\centering
\includegraphics[width=1\textwidth]{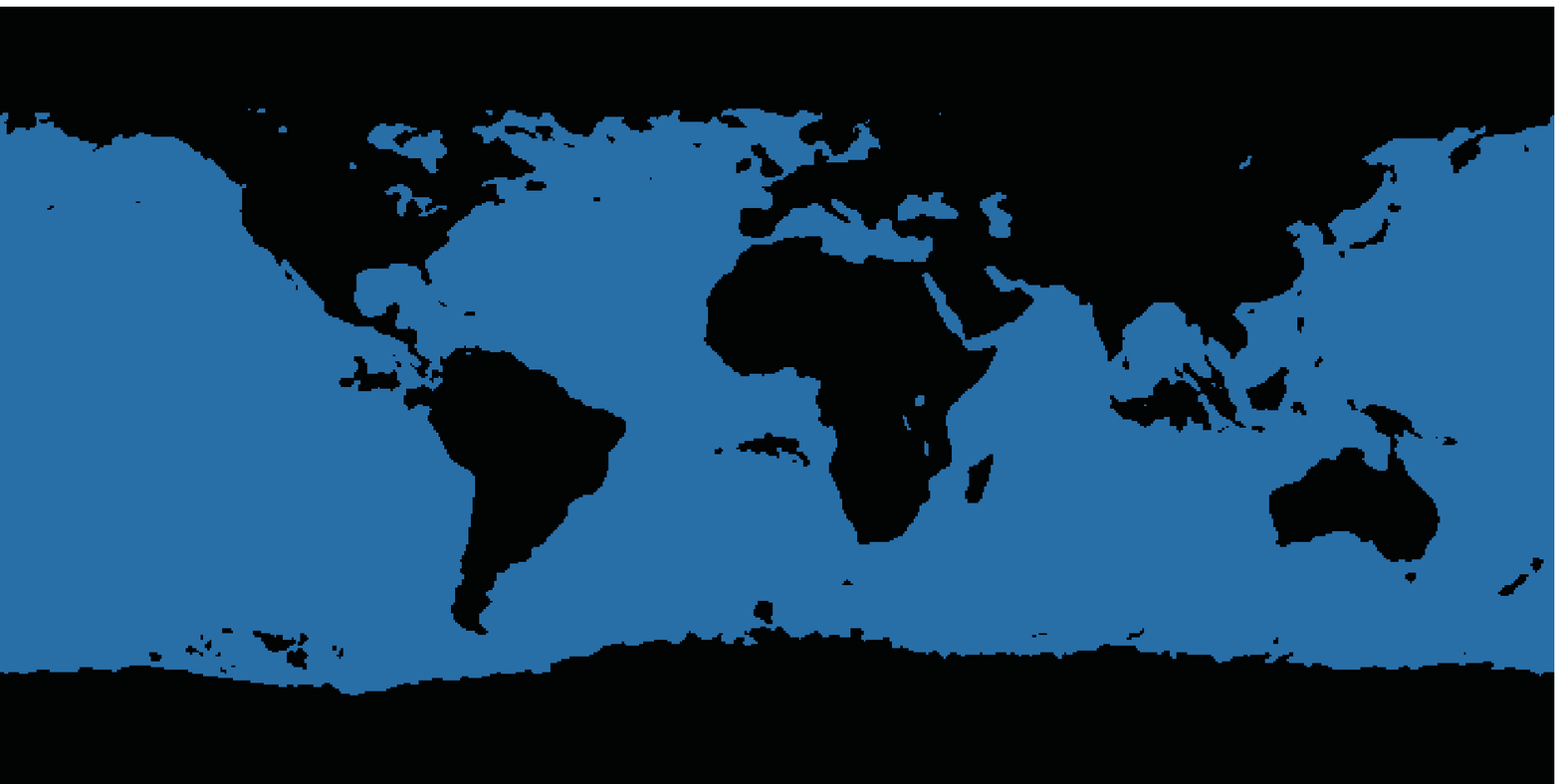}
\centerline{\footnotesize(g)}
\end{minipage}
\begin{minipage}[t]{0.32\linewidth}
\centering
\includegraphics[width=1\textwidth]{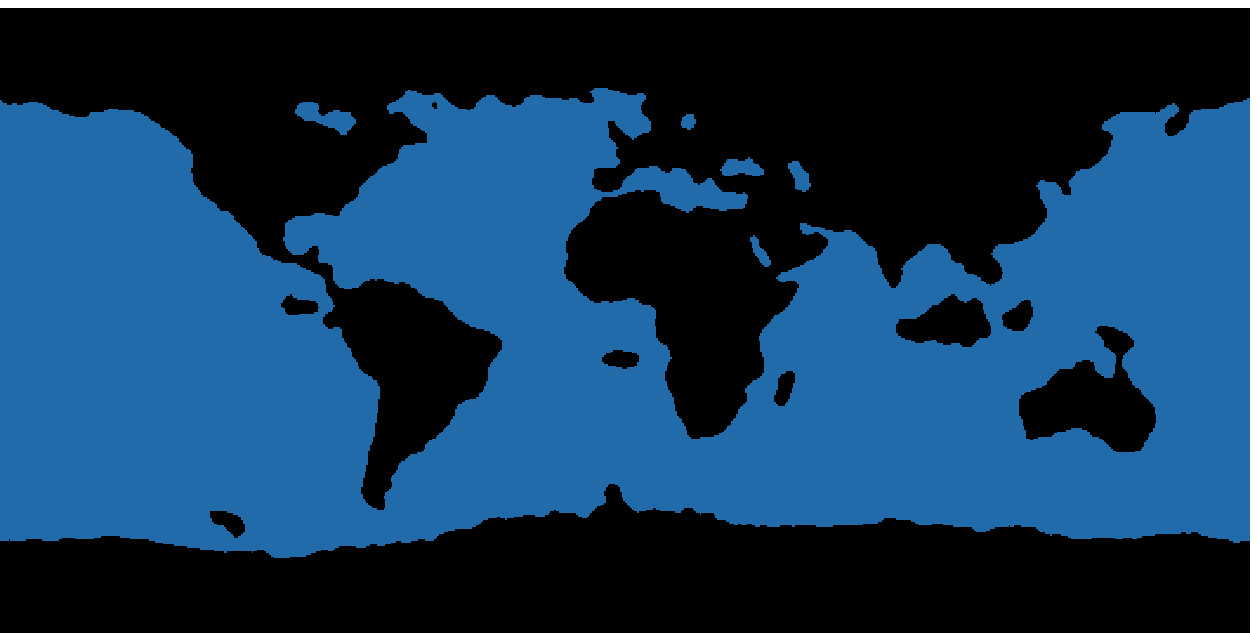}
\centerline{\footnotesize(h)}
\end{minipage}
\begin{minipage}[t]{0.32\linewidth}
\centering
\includegraphics[width=1\textwidth]{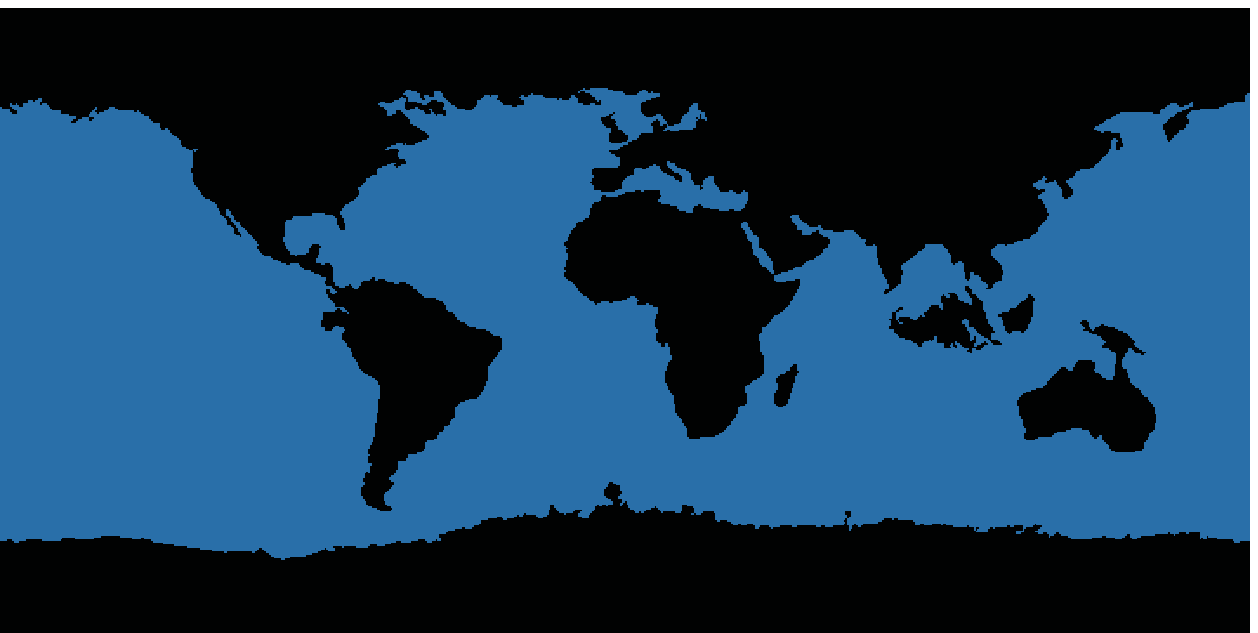}
\centerline{\footnotesize(i)}
\end{minipage}
\caption{Segmentation results on the second real-world image in NEO. The parameter: $\phi =9.98$. From (a) to (i): observed image and results of FCM\_S1, FCM\_S2, FLICM, KWFLICM, FRFCM, WFCM, DSFCM\_N, and WRFCM.}
\label{fig:NEO2}%
\end{figure}

\begin{table}[htb]
    \setlength{\abovecaptionskip}{0pt}
\setlength{\belowcaptionskip}{0pt}
%\tabcolsep 0.01\textwidth
  \centering
  \caption{Segmentation performance (\%) on real-world images in NEO}
  \scriptsize
    \begin{tabular}{c|ccc|ccc}
    \toprule
    \multirow{2}[3]{*}{Algorithm} &
      \multicolumn{3}{c|}{Fig. 11} &
      \multicolumn{3}{c}{Fig. 12}
      \\
\cmidrule{2-7}     &
      SA &
      SDS &
      MCC &
      SA &
      SDS &
      MCC
      \\
\midrule
    FCM\_S1 &
      90.065 &
      97.060 &
      95.106 &
      80.214 &
      92.590 &
      90.329
      \\
    FCM\_S2 &
      93.801 &
      97.723 &
      95.563 &
      81.054 &
      92.066 &
      90.023
      \\
    FLICM &
      90.234 &
      97.056 &
      95.781 &
      81.582 &
      92.352 &
      90.236
      \\
    KWFLICM &
      85.902 &
      80.109 &
      76.329 &
      95.001 &
      96.364 &
      95.633
      \\
    FRFCM &
      81.319 &
      80.616 &
      78.220 &
      96.369 &
      97.309 &
      96.215
      \\
    WFCM &
      95.882 &
      98.854 &
      97.293 &
      97.342 &
      97.430 &
      97.178
      \\
    DSFCM\_N &
      80.131 &
      81.618 &
      79.597 &
      96.639 &
      97.936 &
      96.436
      \\
    WRFCM &
      \textbf{99.080} &
      \textbf{99.149} &
      \textbf{98.512} &
      \textbf{98.881} &
      \textbf{98.797} &
      \textbf{97.582}
      \\
    \bottomrule
    \end{tabular}%
  \label{tab:NEO}%
\end{table}%

Fig. \ref{fig:NEO1} shows the segmentation results on sea ice and snow extent. The colors represent the land and ocean covered by snow and ice per week (here is February 7--14, 2015). We set $c=4$. Fig. \ref{fig:NEO2} gives the segmentation results on chlorophyll concentration. The colors represent where and how much phytoplankton are growing over a span of days. We choose $c=2$. As a whole, by seeing Figs. \ref{fig:NEO1} and \ref{fig:NEO2}, as well as Table \ref{tab:NEO}, FCM\_S1, FCM\_S2, FLICM, KWFLICM, and WFCM are sensitive to unknown noise. FRFCM and DSFCM\_N produce overly smooth results. Especially, they generate incorrect clusters when segmenting the first image in NEO. Superior to its seven peers, WRFCM cannot only suppress unknown noise well but also retain image contours well. In particular, it makes up the shortcoming that other peers forge several topology changes in the form of black patches when coping with the second image in NEO.

\subsection{Performance Improvement}
Besides segmentation results reported for all algorithms, we also present the performance improvement of WRFCM over seven comparative algorithms in Table \ref{tab:improve}. Clearly, for all types of images, the average SA, SDS and MCC improvements of WRFCM over other peers are within the value span 0.238\%--27.836\%, 0.039\%--41.989\%, and 0.047\%--58.681\%, respectively.

\begin{table*}[htbp]
\setlength{\abovecaptionskip}{0pt}
\setlength{\belowcaptionskip}{0pt}
%\tabcolsep 0.017\textwidth
  \centering
  \caption{Average performance improvements (\%) of WRFCM over comparative algorithms}
  \scriptsize
    \begin{tabular}{c|ccc|ccc|ccc|ccc}
    \toprule
    \multirow{2}[4]{*}{Algorithm} &
      \multicolumn{3}{c|}{Synthetic images} &
      \multicolumn{3}{c|}{Medical images} &
      \multicolumn{3}{c|}{\tabincell{c}{Real-world images\\in BSDS and MSRC}} &
      \multicolumn{3}{c}{\tabincell{c}{Real-world images\\ in NEO}}
      \\
\cmidrule{2-13} &
      SA &
      SDS &
      MCC &
      SA &
      SDS &
      MCC &
      SA &
      SDS &
      MCC &
      SA &
      SDS &
      MCC
      \\
    \midrule
    FCM\_S1 &
      8.322  &
      1.782  &
      3.677  &
      4.438  &
      1.309  &
      2.118  &
      26.708  &
      26.221  &
      57.938  &
      13.841  &
      4.148  &
      5.329
      \\
    FCM\_S2 &
      4.863  &
      3.767  &
      7.286  &
      4.407  &
      0.445  &
      0.755  &
      26.783  &
      41.989  &
      58.272  &
      11.553  &
      4.078  &
      5.254
      \\
    FLICM &
      15.275  &
      9.177  &
      16.950  &
      5.024  &
      0.444  &
      0.764  &
      21.045  &
      28.390  &
      47.687  &
      13.072  &
      4.268  &
      5.038
      \\
    KWFLICM &
      0.238  &
      0.038  &
      0.047  &
      5.004  &
      0.494  &
      0.864  &
      27.835  &
      36.791  &
      58.681  &
      8.528  &
      10.736  &
      12.066
      \\
    FRFCM &
      0.293  &
      2.988  &
      5.627  &
      4.270  &
      0.774  &
      1.339  &
      3.976  &
      2.467  &
      9.995  &
      10.136  &
      10.010  &
      10.829
      \\
    WFCM &
      2.100  &
      0.716  &
      1.484  &
      4.883  &
      1.769  &
      3.087  &
      3.852  &
      8.381  &
      9.689  &
      2.368  &
      0.830  &
      0.811
      \\
    DSFCM\_N &
      0.702  &
      3.130  &
      6.071  &
      4.278  &
      4.928  &
      8.445  &
      23.087  &
      30.119  &
      46.672  &
      10.595  &
      9.195  &
      10.030
      \\
    \bottomrule
    \end{tabular}%
  \label{tab:improve}%
\end{table*}%

\subsection{Overhead Analysis}
In the previous subsections, the segmentation performance of WRFCM is presented. Next, we provide the comparison of computing overheads between WRFCM and seven comparative algorithms in order to show its practicality. For a fair comparison, all experiments are implemented in Matlab on a laptop with Intel(R) Core(TM) i5-8250U CPU of (1.60 GHz) and 8.0 GB RAM. The execution time of all algorithms for segmenting synthetic, medical, real-world images is presented in Table \ref{tab:time}. The mean values are in bold. Moreover, we portray them in Fig. \ref{fig:time}.
% Table generated by Excel2LaTeX from sheet 'Sheet1'
\begin{table}[htb]
    \setlength{\abovecaptionskip}{0pt}
\setlength{\belowcaptionskip}{0pt}
\tabcolsep 0.005\textwidth
  \centering
  \caption{Comparison of execution time (in seconds) of all algorithms}
  %\tiny
  \scriptsize
    \begin{tabular}{ccccccccc}
    \toprule
    Image &
      \tiny{FCM\_S1} &
      \tiny{FCM\_S2} &
      \tiny{FLICM} &
      \tiny{KWFLICM} &
      \tiny{FRFCM} &
      \tiny{WFCM} &
      \tiny{DSFCM\_N} &
      \tiny{WRFCM}
      \\
    \midrule
    Fig. 8 column 1 &
      40.453 &
      32.367 &
      4.131 &
      63.069 &
      0.255 &
      2.387 &
      8.245 &
      4.313
      \\
    Fig. 8 column 2 &
      44.116 &
      38.982 &
      4.567 &
      72.607 &
      0.270 &
      5.157 &
      7.271 &
      4.598
      \\
    Fig. 8 column 3 &
      67.155 &
      49.889 &
      3.817 &
      102.019 &
      0.263 &
      3.214 &
      7.501 &
      4.877
      \\
    Fig. 8 column 4 &
      41.030 &
      31.835 &
      3.560 &
      71.339 &
      0.236 &
      2.536 &
      4.561 &
      3.873
      \\
    Fig. 8 column 5 &
      37.364 &
      32.343 &
      3.655 &
      120.872 &
      0.323 &
      2.542 &
      4.245 &
      3.570
      \\
    \textbf{Mean} &
      \textbf{46.024} &
      \textbf{37.083} &
      \textbf{3.946} &
      \textbf{85.981} &
      \textbf{0.270} &
      \textbf{3.167} &
      \textbf{6.365} &
      \textbf{4.246}
      \\
    \midrule
    Fig. 9 column 1 &
      35.518 &
      29.373 &
      3.421 &
      146.758 &
      0.221 &
      2.674 &
      6.565 &
      1.615
      \\
    Fig. 9 column 2 &
      42.423 &
      37.351 &
      3.508 &
      111.133 &
      0.272 &
      2.695 &
      5.007 &
      2.024
      \\
    Fig. 9 column 3 &
      23.213 &
      26.378 &
      3.341 &
      109.381 &
      0.265 &
      2.473 &
      4.242 &
      1.813
      \\
    Fig. 9 column 4 &
      30.322 &
      30.718 &
      6.073 &
      99.687 &
      0.227 &
      2.823 &
      4.689 &
      2.142
      \\
    Fig. 9 column 5 &
      53.060 &
      38.189 &
      5.382 &
      125.533 &
      0.223 &
      3.884 &
      8.439 &
      2.565
      \\
    \textbf{Mean} &
      \textbf{36.907} &
      \textbf{32.402} &
      \textbf{4.345} &
      \textbf{118.498} &
      \textbf{0.242} &
      \textbf{2.910} &
      \textbf{5.789} &
      \textbf{2.032}
      \\
    \midrule
    Fig. 10 column 1 &
      46.739 &
      52.238 &
      4.243 &
      229.498 &
      1.085 &
      5.452 &
      10.283 &
      9.567
      \\
    Fig. 10 column 2 &
      51.998 &
      52.459 &
      4.557 &
      336.104 &
      1.866 &
      9.395 &
      10.715 &
      5.037
      \\
    Fig. 10 column 3 &
      89.086 &
      90.724 &
      4.049 &
      1039.269 &
      1.134 &
      4.960 &
      12.337 &
      2.000
      \\
    Fig. 10 column 4 &
      37.832 &
      38.430 &
      3.203 &
      72.786 &
      0.892 &
      4.011 &
      4.986 &
      3.943
      \\
    Fig. 10 column 5 &
      29.722 &
      27.879 &
      4.436 &
      180.304 &
      0.836 &
      3.277 &
      6.614 &
      3.917
      \\
    \textbf{Mean} &
      \textbf{51.075} &
      \textbf{52.346} &
      \textbf{4.098} &
      \textbf{371.592} &
      \textbf{1.162} &
      \textbf{5.419} &
      \textbf{8.987} &
      \textbf{4.893}
      \\
    \midrule
    Fig. 11 &
      82.535 &
      82.880 &
      7.509 &
      298.926 &
      5.815 &
      6.644 &
      36.648 &
      6.977
      \\
    Fig. 12 &
      44.644 &
      41.817 &
      5.164 &
      54.303 &
      1.786 &
      7.104 &
      18.897 &
      2.761
      \\
    \textbf{Mean} &
      \textbf{63.589} &
      \textbf{62.348} &
      \textbf{6.336} &
      \textbf{176.614} &
      \textbf{3.800} &
      \textbf{6.874} &
      \textbf{27.772} &
      \textbf{4.869}
      \\
    \bottomrule
    \end{tabular}%
  \label{tab:time}%
\end{table}%

\begin{figure}[htb]
\setlength{\abovecaptionskip}{0pt}
\setlength{\belowcaptionskip}{0pt}
\centering
\begin{minipage}[t]{1\linewidth}
\centering
\includegraphics[width=1\textwidth]{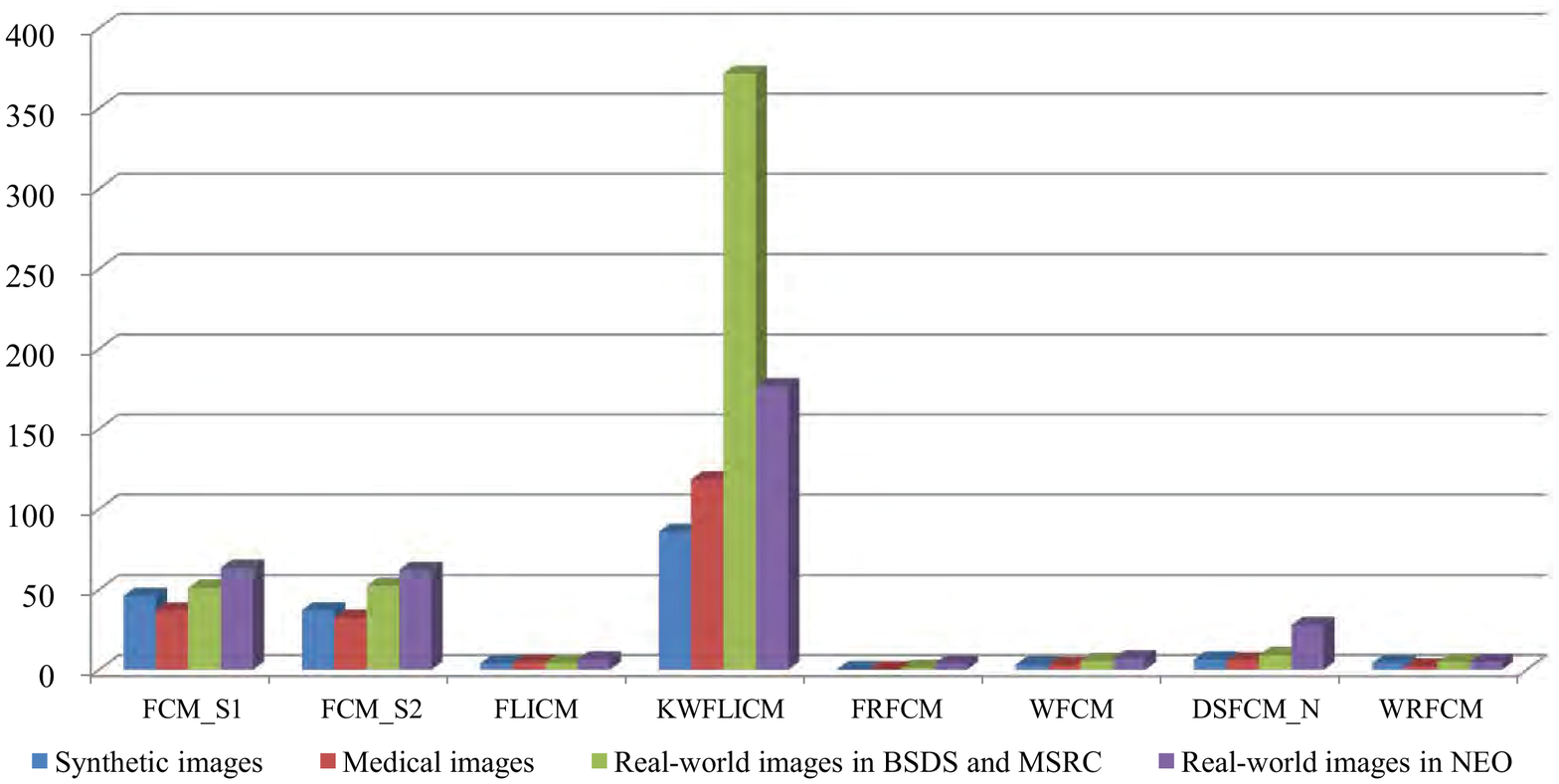}
%\caption{Original model}
\end{minipage}
\vspace*{-1em}
\caption{Average execution time (in seconds) on different types of images.}
\label{fig:time}
\end{figure}

As Table \ref{tab:time} and Fig. \ref{fig:time} show, for gray and color image segmentation, the computational efficiency of KWFLICM is far lower than the others. In contrast, since gray level histograms are considered, FRFCM takes the least execution time among all algorithms. Due to the computation of a neighbor term in each iteration, FCM\_S1 and FCM\_S2 are more time-consuming than the others except KWFLICM. Even though FLICM, WFCM and DSFCM\_N need more computing overheads than FRFCM, they are still very efficient. For color image segmentation, the execution time of DSFCM\_N increases dramatically. Compared with most of seven comparative algorithms, WRFCM shows higher computational efficiency. In most cases, it only runs slower than FRFCM. However, the shortcoming can be offset by its better segmentation performance. In a quantitative study, for each image, WRFCM takes 2.642 seconds longer than FRFCM. However, it saves 45.389, 42.035, 0.671, 184.161, 0.583, and 8.218 seconds over FCM\_S1, FCM\_S2, FLICM, KWFLICM, FRFCM, WFCM, and DSFCM\_N, respectively.

\section{Conclusions and Future Work}\label{sec:conclusions}

For the first time, a residual-driven FCM (RFCM) framework is proposed for image segmentation, which advances FCM research. It realizes favorable noise estimation in virtue of a residual-related fidelity term coming with an analysis of noise distribution. On the basis of the framework, \mbox{RFCM} with weighted $\ell_{2}$-norm fidelity (WRFCM) is presented for coping with image segmentation with mixed or unknown noise. Spatial information is also considered in WRFCM for making residual estimation more reliable. A two-step iterative algorithm is presented to implement WRFCM. Experiments reported for four benchmark databases demonstrate that it outperforms existing FCM variants. Moreover, differing from popular residual-learning methods, it is unsupervised and exhibits a high speed of clustering.

There are some open issues worth pursuing. First, since a tight wavelet frame transform \cite{Wang2018,Yang2017,Wang2020} provides redundant representations of images, it can be used to manipulate and analyze image features and noise well. Therefore, it can be taken as a kernel function so as to produce an improved FCM algorithm, i.e., wavelet kernel-based FCM. Second, can the proposed algorithm be applied to a wide range of non-flat domains such as remote sensing \cite{Xu2017}, ecological systems \cite{Wang2019ecological}, and transportation networks \cite{Lv2017}? How can the number of clusters be selected automatically? Answering them needs more research efforts.

% if have a single appendix:
%\appendix[Proof of the Zonklar Equations]
% or
%\appendix  % for no appendix heading
% do not use \section anymore after \appendix, only \section*
% is possibly needed

% use appendices with more than one appendix
% then use \section to start each appendix
% you must declare a \section before using any
% \subsection or using \label (\appendices by itself
% starts a section numbered zero.)
%

\appendix[Proof of Theorem \ref{lem1}]
Consider the first two subproblems of \eqref{GrindEQ__12_}. The Lagrangian function \eqref{GrindEQ__10_} is reformulated as
\begin{equation} \label{GrindEQ__19_}
{\rm {\mathcal L}}_{\Lambda } ({\bm U},{\bm V})=\sum _{i=1}^{c}\sum _{j=1}^{K}u_{ij}^{m} D_{ij}+\sum _{j=1}^{K}\lambda _{j} \left(\sum _{i=1}^{c}u_{ij}-1\right) ,
\end{equation}
where $D_{ij} ={\kern 1pt} {\kern 1pt} \sum\limits_{n\in {\rm {\mathcal N}}_{j} }\frac{\|{\bm x}_{n} -{\bm r}_{n} -{\bm v}_{i} \|^{2} }{1+d_{nj} }  $.

By fixing ${\bm V}$, we minimize \eqref{GrindEQ__19_} in terms of ${\bm U}$. By zeroing the gradient of \eqref{GrindEQ__19_} in terms of ${\bm U}$, one has
\[\frac{\partial {\rm {\mathcal L}}_{\lambda } }{\partial u_{ij} } =mD_{ij} u_{ij}^{m-1} +\lambda _{j} =0.\]
Thus, $u_{ij} $ is expressed as:
\begin{equation} \label{GrindEQ__20_}
u_{ij} =\left(\frac{-\lambda _{j} }{m} \right)^{1/(m-1)} D_{ij}^{-1/(m-1)} .
\end{equation}
Due to the constraint $\sum _{i=1}^{c}u_{ij}  =1$, one has
\begin{equation*}
\begin{array}{l} {1=\sum\limits_{q=1}^{c}u_{qj}  =\sum\limits_{q=1}^{c}\left(\left(\frac{-\lambda_{j} }{m} \right)^{1/(m-1)} D_{qj}^{-1/(m-1)} \right)}\\{{\kern 1pt} {\kern 1pt} {\kern 1pt} {\kern 1pt} {\kern 1pt} {\kern 1pt} {\kern 1pt} {\kern 1pt} {\kern 1pt} {\kern 1pt} {\kern 1pt} {\kern 1pt} {\kern 1pt} {\kern 1pt} {\kern 1pt} {\kern 1pt} {\kern 1pt} {\kern 1pt} {\kern 1pt} {\kern 1pt} {\kern 1pt} {\kern 1pt} {\kern 1pt} {\kern 1pt} {\kern 1pt} {\kern 1pt} {\kern 1pt} {\kern 1pt} {\kern 1pt} {\kern 1pt} {\kern 1pt} {\kern 1pt} {\kern 1pt} {\kern 1pt} {\kern 1pt} {\kern 1pt} {\kern 1pt} {\kern 1pt} {\kern 1pt} {\kern 1pt} {\kern 1pt} {\kern 1pt} {\kern 1pt} {\kern 1pt} {\kern 1pt} {\kern 1pt} {\kern 1pt} {\kern 1pt}=\left(\frac{-\lambda _{j} }{m} \right)^{1/(m-1)} \sum\limits_{q=1}^{c}D_{qj}^{-1/(m-1)}} \end{array}.
\end{equation*}
In the sequel, one can get
\begin{equation} \label{GrindEQ__21_}
\left(\frac{-\lambda_{j} }{m} \right)^{1/(m-1)} =1/\sum\limits_{q=1}^{c}D_{qj}^{-1/(m-1)}  .
\end{equation}
Substituting \eqref{GrindEQ__21_} into \eqref{GrindEQ__20_}, the optimal $u_{ij} $ is acquired:
\[u_{ij} =\frac{D_{ij}^{-1/(m-1)} }{\sum\limits_{q=1}^{c}D_{qj}^{-1/(m-1)}  } .\]

By fixing ${\bm U}$, we minimize \eqref{GrindEQ__19_} in terms of ${\bm V}$. By zeroing the gradient of \eqref{GrindEQ__19_} in terms of ${\bm V}$, one has
\[\frac{\partial {\rm {\mathcal L}}_{\lambda } }{\partial {\bm v}_{i} } =-2\cdot \sum _{j=1}^{K}\left(u_{ij}^{m} \sum_{n\in {\it {\mathcal N}}_{j} }\frac{({\bm x}_{n} -{\bm r}_{n} -{\bm v}_{i} )}{1+d_{nj} }  \right) =0.\]
The intermediate process is presented as:
\[\sum\limits_{j=1}^{K}u_{ij}^{m} \left(\sum\limits_{n\in {\rm {\mathcal N}}_{j} }\frac{({\bm x}_{n} -{\bm r}_{n} )}{1+d_{nj} }  \right) =\sum\limits_{j=1}^{K}u_{ij}^{m} \left(\sum\limits_{n\in {\rm {\mathcal N}}_{j} }\frac{{\bm v}_{i} }{1+d_{nj} }  \right) .\]
The optimal ${\bm v}_{i} $ is computed:
\[{\bm v}_{i} =\frac{\sum\limits_{j=1}^{K}\left(u_{ij}^{m} \sum\limits_{n\in {\rm {\mathcal N}}_{j} }\frac{{\bm x}_{n} -{\bm r}_{n} }{1+d_{nj} }  \right) }{\sum\limits_{j=1}^{K}\left(u_{ij}^{m} \sum\limits_{n\in {\rm {\mathcal N}}_{j} }\frac{1}{1+d_{nj} }  \right) } .\]

% use section* for acknowledgment
\ifCLASSOPTIONcompsoc
  % The Computer Society usually uses the plural form
\section*{Acknowledgments}
This work is supported in part by the Doctoral Students' Short Term Study Abroad Scholarship Fund of Xidian University, in part by the National Natural Science Foundation of China under Grant Nos. 61873342, 61672400, in part by the Recruitment Program of Global Experts, and in part by the Science and Technology Development Fund, MSAR, under Grant No. 0012/2019/A1.

% Can use something like this to put references on a page
% by themselves when using endfloat and the captionsoff option.
\ifCLASSOPTIONcaptionsoff
  \newpage
\fi

% trigger a \newpage just before the given reference
% number - used to balance the columns on the last page
% adjust value as needed - may need to be readjusted if
% the document is modified later
%\IEEEtriggeratref{8}
% The "triggered" command can be changed if desired:
%\IEEEtriggercmd{\enlargethispage{-5in}}

% references section

% can use a bibliography generated by BibTeX as a .bbl file
% BibTeX documentation can be easily obtained at:
% http://mirror.ctan.org/biblio/bibtex/contrib/doc/
% The IEEEtran BibTeX style support page is at:
% http://www.michaelshell.org/tex/ieeetran/bibtex/
%\bibliographystyle{IEEEtran}
% argument is your BibTeX string definitions and bibliography database(s)
%\bibliography{IEEEabrv,../bib/paper}
%
% <OR> manually copy in the resultant .bbl file
% set second argument of \begin to the number of references
% (used to reserve space for the reference number labels box)

% biography section
%
% If you have an EPS/PDF photo (graphicx package needed) extra braces are
% needed around the contents of the optional argument to biography to prevent
% the LaTeX parser from getting confused when it sees the complicated
% \includegraphics command within an optional argument. (You could create
% your own custom macro containing the \includegraphics command to make things
% simpler here.)
%\begin{IEEEbiography}[{\includegraphics[width=1in,height=1.25in,clip,keepaspectratio]{mshell}}]{Michael Shell}
% or if you just want to reserve a space for a photo:

\begin{IEEEbiography}[{\includegraphics[width=1in,height=1.25in,clip,keepaspectratio]{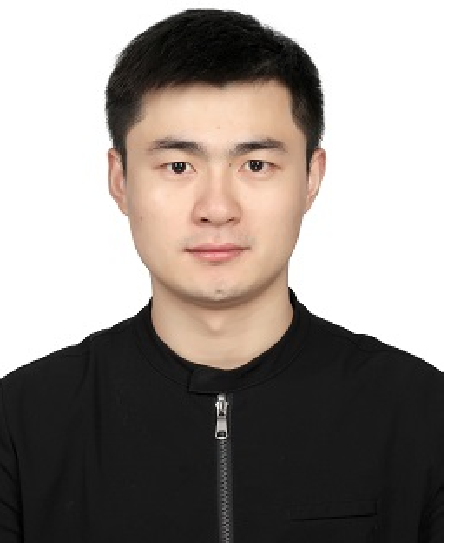}}]{Cong Wang}
received the B.S. degree in automation and the M.S. degree in mathematics from Hohai University, Nanjing, China, in 2014 and 2017, respectively. He is currently pursuing the Ph.D. degree in mechatronic engineering, Xidian University, Xi'an, China.

He was a Visiting Ph.D. Student with the Department of Electrical and Computer Engineering, University of Alberta, Edmonton, AB, Canada. He is currently a Research Assistant at the School of Computer Science and Engineering, Nanyang Technological University, Singapore. His current research interests include wavelet analysis and its applications, granular computing, and pattern recognition and image processing.
\end{IEEEbiography}

\begin{IEEEbiography}[{\includegraphics[width=1in,height=1.25in,clip,keepaspectratio]{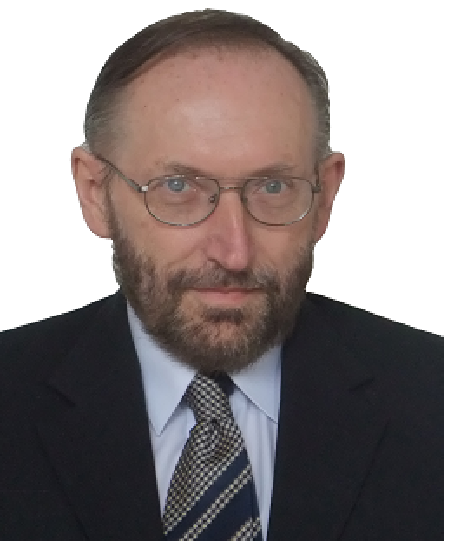}}]{Witold Pedrycz}
(M'88-SM'90-F'99) received the MS.c., Ph.D., and D.Sci., degrees from the Silesian University of Technology, Gliwice, Poland.

He is a Professor and the Canada Research Chair in Computational Intelligence with the Department of Electrical and Computer Engineering, University of Alberta, Edmonton, AB, Canada. He is also with the Systems Research Institute of the Polish Academy of Sciences, Warsaw, Poland. He is a foreign member of the Polish academy of Sciences. He has authored 15 research monographs covering various aspects of computational intelligence, data mining, and software engineering. His current research interests include computational intelligence, fuzzy modeling, and granular computing, knowledge discovery and data mining, fuzzy control, pattern recognition, knowledge-based neural networks, relational computing, and software engineering. He has published numerous papers in the above areas.

Dr. Pedrycz was a recipient of the IEEE Canada Computer Engineering Medal, the Cajastur Prize for Soft Computing from the European Centre for Soft Computing, the Killam Prize, and the Fuzzy Pioneer Award from the IEEE Computational Intelligence Society. He is intensively involved in editorial activities. He is an Editor-in-Chief of Information Sciences, an
Editor-in-Chief of WIREs Data Mining and Knowledge Discovery (Wiley) and the International Journal of Granular Computing (Springer). He currently serves as a member of a number of editorial boards of other international journals. He is a fellow of the Royal Society of Canada.
\end{IEEEbiography}

\begin{IEEEbiography}[{\includegraphics[width=1in,height=1.25in,clip,keepaspectratio]{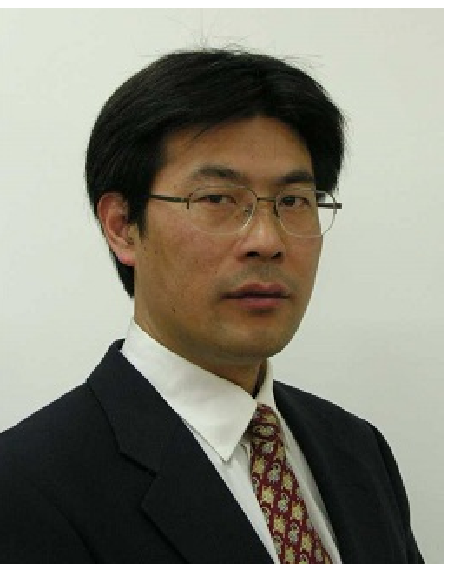}}]{ZhiWu Li}
(M'06-SM'07-F'16) received the B.S. degree in mechanical engineering, the M.S. degree in automatic control, and the Ph.D. degree in manufacturing engineering from Xidian University, Xi'an, China, in 1989, 1992, and 1995, respectively.

He joined Xidian University in 1992. He is also currently with the Institute of Systems Engineering, Macau University of Science and Technology, Macau, China. He was a Visiting Professor with the University of Toronto, Toronto, ON, Canada, the Technion-Israel Institute of Technology, Haifa, Israel, the Martin-Luther University of Halle-Wittenburg, Halle, Germany, Conservatoire National des Arts et M\'{e}tiers, Paris, France, and Meliksah Universitesi, Kayseri, Turkey. His current research interests include Petri net theory and application, supervisory control of discrete-event systems, workflow modeling and analysis, system reconfiguration, game theory, and data and process mining.

Dr. Li was a recipient of an Alexander von Humboldt Research Grant, Alexander von Humboldt Foundation, Germany. He is listed in Marquis Who's Who in the World, 27th Edition, 2010. He serves as a Frequent Reviewer of 90+ international journals, including Automatica and a number of the IEEE \textsc{Transactions} as well as many international conferences. He is the Founding Chair of Xi'an Chapter of IEEE Systems, Man, and Cybernetics Society. He is a member of Discrete-Event Systems Technical Committee of the IEEE Systems, Man, and Cybernetics Society and IFAC Technical Committee on Discrete-Event and Hybrid Systems, from 2011 to 2014.
\end{IEEEbiography}

\begin{IEEEbiography}[{\includegraphics[width=1in,height=1.25in,clip,keepaspectratio]{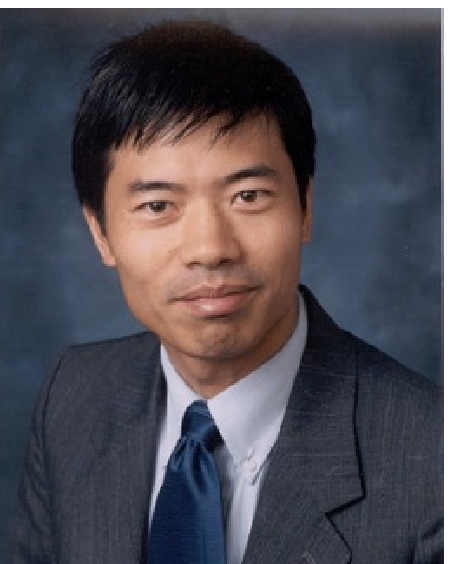}}]{MengChu Zhou}
(S'88-M'90-SM'93-F'03) received his B.S. degree in Control Engineering from Nanjing University of Science and Technology, Nanjing, China in 1983, M.S. degree in Automatic Control from Beijing Institute of Technology, Beijing, China in 1986, and Ph. D. degree in Computer and Systems Engineering from Rensselaer Polytechnic Institute, Troy, NY in 1990.

He joined New Jersey Institute of Technology (NJIT), Newark, NJ in 1990, and is now a Distinguished Professor of Electrical and Computer Engineering. His research interests are in Petri nets, intelligent automation, Internet of Things, big data, web services, and intelligent transportation.

He has over 800 publications including 12 books, 500+ journal papers (400+ in IEEE \textsc{Transactions}), 23 patents and 29 book-chapters. He is the founding Editor of IEEE Press Book Series on Systems Science and Engineering and Editor-in-Chief of IEEE/CAA Journal of Automatica Sinica. He is a recipient of Humboldt Research Award for US Senior Scientists from Alexander von Humboldt Foundation, Franklin V. Taylor Memorial Award and the Norbert Wiener Award from IEEE Systems, Man and Cybernetics Society. He is founding Co-chair of Enterprise Information Systems Technical Committee (TC) and Environmental Sensing, Networking, and Decision-making TC of IEEE SMC Society. He has been among most highly cited scholars for years and ranked top one in the field of engineering worldwide in 2012 by Web of Science/Thomson Reuters and now Clarivate Analytics. He is a life member of Chinese Association for Science and Technology-USA and served as its President in 1999. He is a Fellow of International Federation of Automatic Control (IFAC), American Association for the Advancement of Science (AAAS) and Chinese Association of Automation (CAA).
\end{IEEEbiography}

% You can push biographies down or up by placing
% a \vfill before or after them. The appropriate
% use of \vfill depends on what kind of text is
% on the last page and whether or not the columns
% are being equalized.

%\vfill

% Can be used to pull up biographies so that the bottom of the last one
% is flush with the other column.
%\enlargethispage{-5in}

% that's all folks
\end{document}